\documentclass[11pt,english]{article}
\usepackage{amsmath,amssymb,amsthm}
\usepackage{multicol}
\usepackage[margin=1in]{geometry}
\usepackage{graphicx,color}
\usepackage{enumitem}
\usepackage{fullpage}
\usepackage[noblocks]{authblk}
\usepackage{tcolorbox}
\usepackage{babel}
\usepackage{wrapfig}
\usepackage{MnSymbol}
\usepackage{mdframed}
\usepackage{mathtools}
\usepackage{floatrow}
\usepackage{multirow}
\usepackage{multicol}
\usepackage{subcaption}

\usepackage[leftcaption]{sidecap}

\usepackage[ruled,linesnumbered,vlined]{algorithm2e}
\usepackage{tablefootnote}
\usepackage{thm-restate}
\usepackage{bbding}
\usepackage{pifont}
\usepackage{nicefrac}
\usepackage{undertilde}
\renewcommand{\arraystretch}{1.2}

\definecolor{Darkblue}{rgb}{0,0,0.4}
\definecolor{Brown}{cmyk}{0,0.61,1.,0.60}
\definecolor{Purple}{cmyk}{0.45,0.86,0,0}
\definecolor{Darkgreen}{rgb}{0.133,0.543,0.133}

\usepackage[colorlinks,linkcolor=Darkblue,filecolor=blue,citecolor=blue,urlcolor=Darkblue,pagebackref]{hyperref}
\usepackage[nameinlink]{cleveref}

\usepackage[colorinlistoftodos,prependcaption,textsize=tiny]{todonotes}

\newif\ifdraft 
\draftfalse

\newcommand{\namedref}[2]{\hyperref[#2]{#1~\ref*{#2}}}
\newcommand{\propref}[1]{\hyperref[#1]{property~(\ref*{#1})}}

\newcommand{\blind}{the authors}
\newcommand{\Blind}{The authors}
\newtheorem{theorem}{Theorem}
\newtheorem{lemma}{Lemma}
\newtheorem{definition}{Definition}
\newtheorem{claim}{Claim}
\newtheorem{observation}{Observation}

\newtheorem{question}{Question}

\newtheorem{remark}{Remark}

\newcommand{\poly}{\mathrm{poly}}
\newcommand{\polylog}{\mathrm{polylog}}

\newcommand{\R}{\mathbb{R}}
\newcommand{\N}{\mathbb{N}}

\newcommand{\supp}{\mathrm{supp}}

\newcommand{\dm}{\mathrm{diam}}

\newcommand{\cost}{\mathsf{Cost}}

\newcommand{\tw}{\mathsf{tw}}
\newcommand{\pw}{\mathsf{pw}}
\newcommand{\leavs}{\mathsf{leaves}}

\renewcommand{\int}{\mathsf{Int}}

\newcommand{\hop}{\mathrm{hop}}

\definecolor{forestgreen}{rgb}{0.13, 0.55, 0.13}

\SetCommentSty{mycommfont}

\DeclareMathOperator*{\argmax}{arg\,max}

\def\eps{\epsilon}

\DeclareMathAlphabet{\mathpzc}{OT1}{pzc}{m}{it}
\newcommand{\etal}{{\em et al. \xspace}}

\newcommand{\SPD}{\textsc{SPD}~}
\newcommand{\RSPD}{\textsc{RSPD}~}

\newlength{\dhatheight}

\newcommand {\ignore} [1] {}

\newcommand{\initOneLiners}{%
	\setlength{\itemsep}{0pt}
	\setlength{\parsep }{0pt}
	\setlength{\topsep }{0pt}
}
\newenvironment{OneLiners}[1][\ensuremath{\bullet}]
{\begin{list}
		{#1}
		{\initOneLiners}}
	{\end{list}}

\title{Low Treewidth Embeddings of Planar and Minor-Free Metrics}
\author{Arnold Filtser%\thanks{The research was supported by the Simons Foundation.}
}
\affil{Bar-Ilan University, \texttt{arnold273@gmail.com}}
\author{Hung Le\thanks{The research was supported by the start-up grant of UMass Amherst and  and by the National Science Foundation under Grant No. CCF-2121952.}}
\affil{University of Massachusetts at Amherst, \texttt{hungle@cs.umass.edu}}

\date{}
\begin{document}
\maketitle
\begin{abstract}
Cohen-Addad, Filtser, Klein and Le [FOCS'20] constructed a stochastic embedding of minor-free graphs of diameter $D$ into graphs of treewidth $O_{\eps}(\log n)$ with expected additive distortion $+\epsilon D$. Cohen-Addad  \etal then used the embedding  to design the first quasi-polynomial time approximation scheme (QPTAS) for the capacitated vehicle routing problem. Filtser and Le [STOC'21] used the embedding (in a different way) to design a QPTAS for the metric Baker's problems in minor-free graphs. In this work, we devise a new embedding technique to improve the treewidth bound of  Cohen-Addad  \etal exponentially to $O_{\eps}(\log\log n)^2$. As a corollary, we obtain the first efficient PTAS for the capacitated vehicle routing problem in minor-free graphs. We also significantly improve the running time of the QPTAS for the metric Baker's problems in minor-free graphs from   $n^{O_{\eps}(\log(n))}$ to  $n^{O_{\eps}(\log\log(n))^3}$.

Applying our embedding technique to planar graphs, we obtain a deterministic embedding of planar graphs of diameter $D$ into graphs of treewidth $O((\log\log n)^2)/\eps)$ and  additive distortion $+\epsilon D$ that can be constructed in nearly linear time.  Important corollaries of our result  include a bicriteria PTAS for metric Baker's problems and a PTAS for the vehicle routing problem with bounded capacity in planar graphs, both run in \emph{almost-linear} time. The running time of our algorithms is significantly better than previous algorithms that require quadratic time.

A key idea in our embedding is the construction of an (exact) emulator  for tree metrics with treewidth $O(\log\log n)$ and hop-diameter $O(\log \log n)$. This result may be of independent interest.

\end{abstract}

\newpage

%\vfill
%\begin{multicols}{2}
%	{\small \setcounter{tocdepth}{2} \tableofcontents}
%\end{multicols}

\setcounter{tocdepth}{2} 
\tableofcontents
    \newpage
    \pagenumbering{arabic}

\section{Introduction}
Metric embedding is an influential algorithmic tool that has been applied to many different settings, for example,   approximation/sublinear/online/distributed algorithms \cite{LLR95,AMS99Sketch,BCLLM18,KKMPT12}, machine learning \cite{GKK17}, computational biology \cite{HBKKW03}, and computer vision\cite{AS03}. The fundamental idea of metric embedding in solving an algorithmic problem is to embed an input metric space to a host metric space that is  ``simpler'' than the input metric space, solve the problem in the (simple) host metric space and then map the solution back to a solution of the input metric space. In this algorithmic pipeline, the structure of the host metric space plays a decisive role.

In their seminal result, Fakcharoenphol, Rao and Talwar~\cite{FRT04} (improving over Bartal \cite{Bar96,Bartal98}, see also \cite{Bartal04}) constructed a stochastic embedding of an arbitrary $n$-point metric space to a tree with expected multiplicative distortion $O(\log n)$; the distortion was shown to be optimal \cite{Bar96}. One may hope to get better distortion by constraining the structure of the input metric space,  or enriching the host space. Shattering such hopes, Carroll and Goel~\cite{CG04} (implicitly)
showed that there in an infinite family of planar graphs, such that every deterministic embedding into treewidth $n^{\frac13}$ graphs will have multiplicative distortion $\Omega(n^{\frac13})$.\footnote{This lower bound is achieved by applying \Cref{lm:CGLB} on the $n-1$-subdivision of the  $n\times n$ grid.}
Furthermore, Chakrabarti, Jaffe, Lee, and Vincent~\cite{CJLV08} and Carroll and Goel~\cite{CG04} showed that any stochastic embedding of planar graphs into graphs with \emph{constant} treewidth requires expected distortion $\Omega(\log n)$. 
In fact, one can tweak the construction of Chakrabarti \etal~\cite{CJLV08} to show that even embedding into poly-logarithmic treewidth graph still requires poly-logarithmic expected distortion. Specifically, we show that any stochastic embedding of planar graphs into treewidth $\approx \log^{\frac13}n$ graphs incurs expected distortion $\approx \log^{\frac13}n$ (see \Cref{thm:multLBStochastic}).

Bypassing this roadblock,  Fox-Epstein, Klein, and Schild~\cite{FKS19}  studied \emph{additive embeddings}:   a $\Delta$-additive embedding $f: V(G)\rightarrow V(H)$ of a graph $G$ to a graph $H$ is an embedding such that:
\begin{equation*}\label{eq:additive}
	\mbox{for every }u,v\in V(G),\qquad d_G(u,v) \leq d_H(f(u), f(v)) \leq d_G(u,v) + \Delta~.
\end{equation*}

The parameter $\Delta$ is the \emph{additive distortion} of the $\Delta$-additive embedding $f$. Fox-Epstein \etal~\cite{FKS19} showed that planar graphs of diameter $D$ admit a (deterministic) $(\eps D)$-additive embedding  into graphs of treewidth $O(\epsilon^{-c})$ for some universal constant $c$. The constant treewidth bound (for a constant $\eps$) sharply contrasts  additive embeddings with multiplicative embeddings  (where the treewidth is polynomial).  Their motivation for developing the additive embedding was to design PTASes for the metric Baker's problems in planar graphs.  

In a seminal paper~\cite{Baker94}, Baker designed PTASes for several problems in planar graphs such as independent set, dominating set, and vertex cover, where vertices have non-negative measures (or weights).
Note that these problems are APX-hard in general graphs. Baker's results subsequently inspired the development of powerful algorithmic frameworks for planar graphs, such as deletion decomposition~\cite{Baker94,Eppstein2000,DHK05}, contraction decomposition~\cite{Klein05,DHM07,DHK11}, and bidimensionality~\cite{DH05,FLRS11}. Metric Baker's problems generalize Baker's problems in that vertices in the solution must be at least/at most a distance $\rho$ from each other for some input parameter $\rho$. The most well-studied examples of  metric Baker's problems include  $\rho$-independent set, $\rho$-dominating set, $(k,r)$-center. Metric Baker problems have been studied in the context of parameterized complexity~\cite{DFHT05,MP15,BL16,KLP20} where (a) the input graphs are restricted to subclasses of minor-free graphs, such as planar and bounded treewidth graphs, and (b) parameters such as $\rho$ and/or the size of the optimal solution are small. When $\rho$ is a constant, Baker's layering technique can be applied to obtain a linear time PTAS for unweighted planar graphs ~\cite{Baker94} and efficient PTASes for unweighted minor-free graphs~\cite{DHK05}.  However, the most challenging case is when the graph is weighted, and $\rho$ is part of the input; even when restricted to bounded treewidth graph, (single-criteria) PTASes are not known for metric Baker problems.\footnote{In bounded treewidth graphs, no single-criteria algorithm for metric Baker problems are known where the objective of the approximation is the measure of the set. However, there is a single-criteria PTAS where the objective of the approximation is the parameter $\rho$ (and the measure is fixed). See \Cref{lm:KLS-IS-DS}.} 
Marx and Pilipczuk~\cite{MP15} showed that, under Exponential Time Hypothesis (ETH), $\rho$-independent/dominating set problems cannot be solved in time $f(k)n^{o(\sqrt{k})}$ when the solution size is at most $k$. As observed by Fox-Epstein \etal \cite{FKS19},  the result of Marx and Pilipczuk~\cite{MP15} implies that, under ETH, there is no (single-criteria) efficient PTAS for   $\rho$-independent/dominating set problems in planar graphs. However, for the case of uniform measure (i.e. $\forall v$, $\mu(v)=1$), a (non-efficient) PTAS can be obtained via local search~\cite{FL21}. 

Fox-Epstein \etal \cite{FKS19} bypassed the ETH lower bound by designing a \emph{bi-criteria} efficient  PTAS for $\rho$-independent/dominating set problems in planar graphs using their additive distortion embedding of planar graphs into bounded treewidth graphs. Since the treewidth of their embedding is $O(\eps^{-c})$ for some constant $c\geq 19$ (see Section 6.5 in~\cite{FKS19}), and the running time to construct  the embedding  is $n^{O(1)}$ for an unspecified constant in the exponent due to the embedding step, the running time of their PTAS is $2^{O(\eps^{-(c+1)})}n^{O(1)}$~\cite{FKS19}.   In their paper, they noted:

\begin{quote}``Admittedly, in our current proof, the treewidth is bounded by a polynomial of very high degree in $1/\eps$. There is some irony in the fact that our approach to achieving an efficient PTAS yields an algorithm that is inefficient in the constant dependence on $\eps$.''
\end{quote}

Given the state of affairs, the following problem nauturally arises:
	
	\begin{question}\label{q:baker-planar} Can we design a PTAS for metric Baker's problems with (almost) linear running time? Can we obtain a PTAS with a more practical dependency on $\eps$?
	\end{question}

Cohen-Addad, Filtser, Klein, and Le \cite{CFKL20} studied additive embeddings in a more general setting of  $K_r$-minor-free graphs for any fixed $r$. They proved a strong lower bound on treewidth against \emph{deterministic} additive embeddings. Specifically, they showed (Theorem 3 in~\cite{CFKL20}) that there is an $n$-vertex $K_6$-minor-free graph such that any deterministic additive embedding into a graph of treewidth $o(\sqrt{n})$ must incur a distortion at least $\frac{D}{20}$. On the other hand, they showed that randomness helps reduce the treewidth exponentially. In particular, they constructed a \emph{stochastic} additive embedding of $K_r$-minor-free graphs into graphs with treewidth $O(\frac{\log n}{\epsilon^2})$ and \emph{expected} additive distortion $+\epsilon D$. 
Specifically, there is distribution $\mathcal{D}$ over dominating embeddings (i.e. no distances shrink) into treewidth $O(\frac{\log n}{\epsilon^2})$ graphs such that $\forall u,v,~\mathbb{E}_{f,H_f}[d_{H_f}(f(u),d(v))]\le d_G(u,v)+\eps D$.

Their primary motivation was to design a quasi-polynomial time approximation scheme (QPTAS) for the bounded-capacity vehicle routing problem (VRP) in  $K_r$-minor-free graphs. In  this problem, we are given an edge-weighted graph $G=(V,E,w)$,  a set of clients $K\subseteq V$, a depot $r\in V$, and the capacity  $Q\in \mathbb{Z}_+$ of the vehicle. We are tasked with finding a collection of tours $\mathcal{S} = \{R_1,R_2,\ldots\}$ of \emph{minimum cost} such that each tour, starting from $r$ and ending at $r$, visits at most $Q$ clients  and every client is visited by at least one tour; the cost of $\mathcal{S}$ is the total  weight of all tours in $\mathcal{S}$.   (See \Cref{def:vehicle-routing} for a more formal definition.) The VRP was introduced by Dantzig and Ramser~\cite{DR59} and has been extensively studied since then; see the survey by Fisher~\cite{Fisher95}. The problem is APX-hard, as it is a generalization of the  Travelling Salesperson Problem (TSP) when $|Q| = |V|$, which is APX-hard~\cite{PY93}, and admits a constant factor approximation~\cite{HR85}. To get a $(1+\epsilon)$-approximation, it is necessary to restrict the structures of the input graph.  Fundamental graph structures that have long been studied are low dimensional Euclidean (or doubling) spaces, planarity and minor-freeness. 

When the capacity $Q$ is a part of the input, a QPTAS for VRP in Euclidean space of constant dimension is known~\cite{DM15,ACL09}; it remains a major open problem to design a PTAS for VRP even for the Euclidean plane. For trees, a PTAS was only obtained by a recent work of Mathieu and Zhou~\cite{MZ21}, improving upon the QPTAS of Jayaprakash and Salavatipour~\cite{JS22}. No PTAS is known beyond trees, such as planar graphs or bounded treewidth graphs. For the \emph{unsplittable} demand version of the problem on trees, Becker~\cite{Becker18} showed that the problem is APX-hard.  

Going beyond trees, it is natural to restrict the problem further by considering \emph{constant} $Q$.  In this regime, Becker \etal \cite{BKS19} designed the first (randomized) PTAS for bounded-capacity VRP  in planar graphs with running time $n^{O_{\epsilon}(1)}$, improving upon the earlier QPTAS by Becker, Klein and Saulpic~\cite{BKS17}. Recently, Cohen-Addad \etal~\cite{CFKL20} obtained the first efficient PTAS for the problem in planar graphs with running time $O_{\epsilon}(1)n^{O(1)}$. Their algorithm uses the embedding of~\cite{FKS19} as a blackbox, and hence, suffers the drawback of the embedding: the exponent of $n$ in the running time is  unspecified due to the embedding step. 

In $K_r$-minor-free graphs, Cohen-Addad \etal\cite{CFKL20} designed a QPTAS for bounded-capacity VRP in $K_r$-minor-free graphs using their (stochastic) embedding  of $K_r$-minor-free graphs into graphs with treewidth $O(\frac{\log n}{\epsilon^2})$. More precisely, their algorithm has running time $(\log n)^\tw n^{O(1)}$ where $\tw$ is the treewidth of the embedding. Thus, any significant improvement over treewidth bound $O(\frac{\log n}{\epsilon^2})$, such as $O(\log n/\log\log n)$, would lead to a PTAS. Their work left the following questions as open problems.

	\begin{question}\label{q:minor-emb}  Can we devise an additive embedding of  $K_r$-minor-free graphs into graphs with treewidth $O(\log(n)/\log\log(n))$? Can we design a PTAS for the   bounded-capacity VRP in $K_r$-minor-free graphs? Can we improve the running time of the PTAS for the  bounded-capacity VRP in planar graphs to (almost) linear?
	\end{question}

We remark that the dependency on $\eps$ of all  known approximation schemes for bounded-capacity VRP problem is \emph{doubly exponential} in $1/\eps$~\cite{BKS19,CFKL20}. Reducing the dependency to singly exponential in $1/\eps$ is a fascinating open problem.

One drawback of the stochastic embedding with additive distortion is that it cannot be applied directly to design (bicriteria) PTASes for metric Baker problems in minor-free graphs. To remedy this drawback, \blind \cite{FL21} recently introduced \emph{clan embedding} and \emph{Ramsey type embedding} with additive distortions. A clan embedding of a graph $G$ to a graph $H$ is pair of maps $(f,\chi)$ where $f$ is  a \emph{one-to-many} embedding $f: V(G) \rightarrow 2^{V(H)}$ that maps each vertex $x \in V(G)$ to a subset of vertices $f(x)\subseteq V(H)$ in $H$, called \emph{copies} of $x$ (where each set $f(x)\neq\emptyset$ is not empty, and every two sets of $x\ne y$ are $f(x)\cap f(y)=\emptyset$ are disjoint), and $\chi: V(H)\rightarrow V(H)$ maps each vertex $x$ to a vertex $\chi(x) \in f(x)$, called the \emph{chief} of $x$. Furthermore, $f$ must be \emph{dominating}:
\begin{equation*}
\forall x,y\in V(G),\qquad d_G(x,y) \leq \min_{x' \in f(x), y' \in f(y)}d_H(x',y')~.
\end{equation*}
A clan embedding $(f,\chi)$ has additive distortion $+\Delta$ if for every $x,y \in V(G)$:
\begin{equation}\label{eq:clanDef}
\min_{x' \in f(x)}d_H(x',\chi(y)) \leq d_G(x,y)  + \Delta
\end{equation}
That is, there is some vertex in the clan of $x$ which is close to the chief of $y$.  Note that the distortion guarantee is in the worst case. That is, \Cref{eq:clanDef} holds for every pair $x,y\in V(G)$, and every embedding $f\in\supp(\mathcal{D})$ in the support.

A Ramsey type embedding is a (stochastic) one-to-one embedding in which there is a subset of vertices $M\subseteq V$ such that every vertex is included in $M$ with a probability at least $1-\delta$, for a given parameter $\delta \in (0,1)$, and for every vertex $u\in M$, the additive distortion of the distance from $u$ to every other vertex in $V$ is $+\Delta$. See \Cref{def:ramsey} for a formal definition.

\Blind~\cite{FL21} showed previously that for any given $K_r$-minor-free graph and parameters $\eps,\delta\in(0,1)$, one can construct a distribution $\mathcal{D}$ over clan embeddings into $O_r(\frac{\log^2n}{\eps\delta})$-treewidth graphs with additive distortion $+\eps D$ ($D$ being the diameter) and such that the expected clan size $\mathbb{E}[|f(x)|]$ of every $x \in V(G)$ is bounded by $1+\delta$.  For Ramsey type embeddings, the treewidth is also $O_r(\frac{\log^2n}{\eps\delta})$ for an additive distortion $+\eps D$~\cite{FL21}.

Using the clan embedding and  Ramsey type embedding, \blind obtained a QPTAS for metric Baker's problem in $K_r$-minor-free graphs with running time $n^{O_r(\eps^{-2}\log(n)\log\log(n))}$~\cite{FL21}.  The precise running time of the algorithm is $(\log n)^{O(\tw)} n^{O(1)}$ where $\tw$ is the treewidth of the embeddings. The key questions are:

\begin{question}\label{q:clan-Ramsey-minor} Can we improve the treewidth bound $O(\log^2(n))$ of the clan embedding and Ramsey embedding? Can we design a PTAS for metric Baker's problems in $K_r$-minor-free graphs? 
\end{question}

\subsection{Our Results}

In this paper, we provide affirmative answers to \Cref{q:baker-planar} and \Cref{q:minor-emb}, while making significant progress toward \Cref{q:clan-Ramsey-minor}. 

We construct  a new stochastic additive embedding of $K_r$-minor-free graphs into graphs with treewidth $O_r(\frac{(\log \log n)^2}{\epsilon^2})$. Our treewidth bound improves \emph{exponentially} over the treewidth bound of Cohen-Addad \etal~\cite{CFKL20}. 
See \Cref{tab:embeddings} for a summary of new and previous embeddings.

\begin{restatable}[Embedding Minor-free Graphs to Low Treewidth Graphs]{theorem}{MinorToTreewidth}
	\label{thm:MinorToTreewidth}
	Given an $n$-vertex $K_r$-minor-free graph $G(V,E,w)$ of diameter $D$, we can construct in polynomial time a stochastic additive embedding $f: V(G) \rightarrow H$ into a distribution over graphs $H$ of treewidth at most $O_r(\epsilon^{-2}(\log \log n)^2)$ and expected additive distortion $\epsilon D$. 
\end{restatable}

 The main bottleneck for not having an almost-linear time algorithm for the embedding in minor-free graphs is that the best algorithm for computing Robertson-Seymour decomposition takes quadratic time~\cite{KKR12}.

Using our stochastic additive embedding, we design the first PTAS  for  the bounded-capacity vehicle routing problem in $K_r$-minor-free graphs; our PTAS indeed is efficient. This resolves \Cref{q:minor-emb} in the affirmative (more on the planar case later in \Cref{thm:VRP-planar}).

\begin{restatable}[PTAS for Bounded-Capacity VRP in Minor-free Graphs]{theorem}{CVRPMinor}
	\label{thm:VRP-minor}  There is a randomized polynomial time approximation scheme for the bounded-capacity VRP in $K_r$-minor-free graphs that runs in $O_{\epsilon,r}(1)\cdot n^{O(1)}$ time.
\end{restatable}

\begin{table}[]
	\begin{tabular}{|c|c|l|l|l|}
		\hline
		\multicolumn{1}{|l|}{\textbf{Family}} & \multicolumn{1}{l|}{\textbf{Type}}      & \textbf{Treewidth}                                & \textbf{Runtime}                      & \textbf{Ref}                         \\ \hline
		\multirow{3}{*}{Planar}               & \multirow{3}{*}{Deterministic}          & $O_r(\epsilon^{-1}\cdot\log n)$       			& $O_{\eps}(n^{O(1)})$                            & \cite{EKM14}                         \\ \cline{3-5} 
		&                                                                               & $O_r(\epsilon^{-c}), c\geq 19$                          	& $O_{\eps}(n^{O(1)})$                            & \cite{FKS19}                         \\ \cline{3-5} 
		&                                                                               & $O_r(\epsilon^{-1}\cdot(\log \log n)^2)$          & $\epsilon^{-2}\cdot \tilde{O}(n)$     & \Cref{thm:PlanarToTreewidth}         \\ \hline
		\multirow{4}{*}{$K_r$-minor free}     & \multirow{2}{*}{Stochastic}             & $O_r(\epsilon^{-2}\cdot\log n)$                 	& $O_{\eps,r}(n^{O(1)})$                            & \cite{CFKL20}                        \\ \cline{3-5} 
		&                                                                               & $O_r(\epsilon^{-2}\cdot(\log \log n)^2)$          & $O_{\eps,r}(n^{O(1)})$                            & \Cref{thm:MinorToTreewidth}          \\ \cline{2-5} 
		& \multirow{2}{*}{\begin{tabular}[c]{@{}c@{}}Ramsey type\\ / Clan\end{tabular}} & $(\eps\delta)^{-1}\cdot O_r(\log^2n)$             & $O_{\eps,r}(n^{O(1)})$                            & \cite{FL21}                          \\ \cline{3-5} 
		&                                                                               & $(\eps\delta)^{-1}\cdot \tilde{O}_{r}(\log n)$ 	& $O_{\eps,r}(n^{O(1)})$                            & \Cref{thm:MinorClan,thm:MinorRamsey} \\ \hline
	\end{tabular}
	\caption{\it\small Summary of current and previous embeddings into low treewidth graphs with additive distortion $+\eps D$.
		For planar graphs, we get an embedding whose treewidth has a minor depdency on $n$; however our running time and dependency on $\eps$ is much improved. For stochastic embeddings of $K_r$-minor-free graphs, we obtain an exponential improvement. 
		For Ramsey-type and clan embeddings, we obtain a quadratic improvement.
		\label{tab:embeddings}}
\end{table}

Our proof of \Cref{thm:VRP-minor} follows the embedding framework of Cohen-Addad \etal\cite{CFKL20}.  In a nutshell, they showed that if  planar graphs with \emph{one vortex} (see \Cref{sec:addNotation} for a formal definition) have an additive embedding with treewidth $k(n,\eps)$, then $K_r$-minor-free graphs have a stochastic additive embedding with treewidth roughly $O_r(k(n,O_r(\eps^2)) + \log(n))$. That is, the reduction incurs an additive factor of $O(\log n)$.  Thus, there are two issues one has to resolve to reduce the treewidth to $o(\log(n))$: (i) construct an embedding of  planar graphs with one vortex  that has treewidth $k(n,\eps) = o(\log(n))$ and (ii) remove the loss $O(\log(n))$ in the reduction of Cohen-Addad \etal\cite{CFKL20}. By a relatively simple idea (see \Cref{sec:MinorToTW}), we could remove the additive term $O(\log(n))$ in the reduction. We are left with constructing an embedding of planar graphs with one vortex, which is the main barrier we overcome in this work.

One can intuitively think of  planar graphs with one vortex as \emph{noisy} planar graphs, where the vortex is a kind of low-complexity\footnote{Low-complexity means that the vortex has bounded pathwidth.} noise that affects the planar embedding in a local area, i.e., a face. From this point of view, we need an embedding of planar graphs that is \emph{robust} to the noise. The embedding for planar graphs of Fox-Epstein \etal~\cite{FKS19} relies heavily on planarity to perform topological operations such as cutting paths open, and decomposing the graphs into so-called \emph{bars and cages}. As a result, adding a vortex to a planar graph makes their embedding inapplicable. Cohen-Addad \etal\cite{CFKL20} instead use \emph{balanced shortest path separators} in their embedding of planar graphs with one vortex. Using shortest path separators results in an embedding technique that is robust to the noise caused by the vortex. However,  their embedding has treewidth $O(\log(n)/\eps)$ due to that recursively decomposing the graphs using balanced separators gives a decomposition tree of depth  $O(\log n)$.  This  is a universal phenomenon in almost all techniques that rely on balanced separators: the depth $O(\log n)$ factors in many known algorithms in planar graphs~\cite{GKP95,AGKKW98,EKM14,FNM15,CCKMM16}. It seems that to get a treewidth $o(\log n)$, one needs to avoid using balanced separators in the construction. 

Surprisingly perhaps, we can still use balanced shortest path separators to get an embedding with treewidth $O((\log\log n)^2/\eps)$ for planar graphs  with one vortex (see \Cref{lm:nearlyPlanar-emb} for a formal statement). Our key idea is to ``shortcut'' the decomposition tree of depth $O(\log n)$ by adding edges between nodes in such a way that, for every two  nodes in the tree, there is a path in the shortcut tree with $O(\log\log n)$ edges, a.k.a. $O(\log\log n)$ \emph{hops}. To keep the treewidth of the embedding small, we require that  the resulting treewidth of the decomposition tree (after adding the shortcuts) remains small; the decomposition tree is a tree and hence has treewidth $1$. We show that we can add shortcuts in such a way that the resulting treewidth is $O(\log\log(n))$. The two factors of $O(\log\log(n))$  --- one from the hop length and one from the treewidth blowup due to shortcutting --- result in treewidth of $O((\log\log n)^2/\eps)$ of the embedding. We formulate these ideas in terms of constructing an emulator for trees with treewidth $O(\log\log n)$ and hop-diameter $O(\log \log n)$; see \Cref{subsec:technique} for a more formal discussion. The main conceptual message of our technique is that it is possible to get around the depth barrier  $O(\log n)$ of balanced separators by shortcutting the decomposition tree. We belive that this idea would find further use in designing planar graph algorithms.

Applying our technique to planar graphs, we obtain an additive embedding with treewidth $O((\log\log)^2/\eps)$ that can be constructed in nearly linear time.

\begin{restatable}[Embedding Planar Graphs to Bounded Treewidth Graphs]{theorem}{PlanarToTreewidth}
	\label{thm:PlanarToTreewidth}
	Given an $n$-vertex planar graph $G(V,E,w)$ of diameter $D$ and a parameter $\epsilon\in(0,1)$, there is a deterministic embedding ${f:V(G) \rightarrow H}$ into a graph $H$ of treewidth $O(\epsilon^{-1}(\log \log n)^2)$ and additive distortion $+\epsilon D$. \\
	Furthermore, $f$ can be deterministically constructed in $O(n\cdot\frac{\log^3n}{\epsilon^2})$ time.	
\end{restatable}

Our embedding, while having a minor dependency on $n$, offers three advantages over the embedding of \cite{FKS19}. First, our embedding can be constructed in nearly linear running time, removing the running time bottleneck of algorithms that uses the embedding of  \cite{FKS19}. Second,  our embedding algorithm is much simpler than that of \cite{FKS19} as it only uses the recursive shortest path separator decomposition.  Third, the dependency on $\epsilon$ of the treewidth is linear, which we show to be optimal by the following theorem.

\begin{restatable}[Embedding Planar Graphs Lower Bound]{theorem}{PlanarEmbeddingLB}
	\label{thm:PlanarEmbeddingLB} For any $\epsilon \in (0,\frac12)$ and any  $n = \Omega(1/\eps^2)$, there exists an unweighted $n$-vertex planar graph $G(V,E)$ of diameter $D \leq (1/\eps + 2)$ such that  for any deterministic dominating embedding $f: V(G) \rightarrow H$ into a graph $H$ with additive distortion $+\epsilon D$, the treewidth of $H$ is $\Omega(1/\eps)$. 	
\end{restatable}

We note that Eisenstat, Klein, and Mathieu \cite{EKM14} implicitly constructed an embedding with additive distortion $+\epsilon D$ and treewidth $O(\epsilon^{-1}\cdot\log n)$ (see also~\cite{CFKL20}).  However, in the applications of designing PTASes, the running time is at least exponential in the treewidth, and hence, this embedding only implies ineffcient PTASes or QPTASes. The dependency on $n$ of the treewidth of our  embedding in \Cref{thm:PlanarToTreewidth} is exponentially smaller. 

Our new embedding result (\Cref{thm:PlanarToTreewidth}) leads to an almost linear time PTAS for metric Baker's problems in planar graphs, thereby giving an affirmative answer to \Cref{q:baker-planar}. See \Cref{tab:Application} for  a comparison with existing results.

\begin{restatable}[PTAS for Metric Baker Problems in Planar Graphs]{theorem}{MetricBakerPlanar}
	\label{thm:MetricBakerPlanar}
	Given an $n$-vertex planar graph $G(V,E,w)$, two parameters $\epsilon \in (0,1)$ and $\rho > 0$, and a measure $\mu: V\rightarrow \mathbb{R}^+$, one can find in $2^{\tilde{O}(1/\eps^{2+\kappa})}\cdot n^{1+o(1)}$ time for \emph{any fixed $\kappa > 0$}: (1) a $(1-\epsilon)\rho$-independent  set $I$ such that for every $\rho$-independent  set $\tilde{I}$, $\mu(I) \geq (1-\epsilon)\mu(\tilde{I})$ and (2) a  $(1+\epsilon)\rho$-dominating set $S$ such that for every $\rho$-dominating set $\tilde{S}$, $\mu(S) \leq (1+\epsilon)\mu(\tilde{S})$.
\end{restatable}

We could choose $\kappa = 0.01$ in \Cref{thm:MetricBakerPlanar} to get a PTAS with running time $2^{\tilde{O}(1/\eps^{2.01})}\cdot n^{1+o(1)}$ for the $\rho$-independent set/dominating set problems. The dependency on $1/\eps$ of the running time of our PTAS is much smaller than that of~\cite{FKS19}.

By using our new embedding in \Cref{thm:PlanarToTreewidth} and some additional ideas, we significantly improve the running time of the PTAS by Cohen-Addad \etal~\cite{CFKL20} to almost linear time as asked in \Cref{q:minor-emb}. See \Cref{tab:Application} for a summary.

\begin{restatable}[PTAS for Bounded-Capacity VRP in Planar Graphs]{theorem}{CVRPPlanar}
	\label{thm:VRP-planar} There is a randomized polynomial time approximation scheme for the bounded-capacity VRP in planar graphs that runs in time $O_{\epsilon}(1)\cdot n^{1+o(1)}$.  
\end{restatable}

\begin{table}
	\renewcommand{\arraystretch}{1.3}
	\begin{minipage}[c]{0.5\linewidth}
		\scalebox{0.89}{\begin{tabular}{c|l|l|}
				\multicolumn{3}{c}{Bounded-capacity vehicle routing problem}                                                                                                                     \\ \hline
				\multicolumn{1}{|c|}{}                                  & \multicolumn{1}{l|}{Running time}                                                  & \multicolumn{1}{l|}{Ref}                   \\ \hline
				\multicolumn{1}{|c|}{Euclidean plane}                   & $2^{O(\epsilon^{-3})} + O(n\log n)$ & \cite{AKTT97}            \\ \hline
				\multicolumn{1}{|c|}{\multirow{4}{*}{Planar graph}}     & $n^{\poly(\eps^{-1}\log n)}$                            & \cite{BKS17}             \\ \cline{2-3} 
				\multicolumn{1}{|c|}{}                                  & $n^{O_{\epsilon}(1)}$              & \cite{BKS19}             \\ \cline{2-3} 
				\multicolumn{1}{|c|}{}                                  & $O_{\epsilon}(1)\cdot n^{O(1)}$          & \cite{CFKL20}            \\ \cline{2-3} 
				\multicolumn{1}{|c|}{}                                  & $O_{\epsilon}(1)\cdot n^{1+o(1)}$        & \Cref{thm:VRP-planar}    \\ \hline
				\multicolumn{1}{|c|}{\multirow{2}{*}{$K_r$-minor free}} & $n^{O_{r,\eps}(\log\log n)}$                            & \cite{CFKL20}            \\ \cline{2-3} 
				\multicolumn{1}{|c|}{}                                  & $O_{\epsilon,r}(1)\cdot n^{O(1)}$        & \Cref{thm:VRP-minor} \\ \hline
		\end{tabular}}
	\end{minipage}%
	\begin{minipage}[c]{0.5\linewidth}
		\renewcommand{\arraystretch}{1.67}
		\scalebox{0.89}{
			\begin{tabular}{c|l|l|}
				\multicolumn{3}{c}{Metric $\rho$ dominating/independent set problems}                                                                                                                     \\ \hline
				\multicolumn{1}{|c|}{}                                  & \multicolumn{1}{l|}{Running time}                                                  & \multicolumn{1}{l|}{Ref}                   \\ \hline
				\multicolumn{1}{|c|}{\multirow{3}{*}{Planar graph}}     & $n^{O_\eps(1)}$ & \cite{EKM14}                 \\ \cline{2-3} 
				\multicolumn{1}{|c|}{}&\scalebox{0.89}{\renewcommand{\arraystretch}{1.2}\begin{tabular}[c]{@{}l@{}} $2^{O(\eps^{-c+1})}\cdot n^{O(1)}$ \\ for $c\geq 19$\end{tabular}    }   & \cite{FKS19}                 \\ \cline{2-3} 					
				\multicolumn{1}{|c|}{}                                  & \scalebox{0.89}{\renewcommand{\arraystretch}{1.35}\begin{tabular}[c]{@{}l@{}}$2^{\tilde{O}(\eps^{-(2+\kappa)})}\cdot n^{1+o(1)}$,\\ any fixed $\kappa > 0$\end{tabular}    }                                                                                                                   & \Cref{thm:MetricBakerPlanar} \\ \hline
				\multicolumn{1}{|c|}{\multirow{2}{*}{$K_r$-minor free}} & $n^{\tilde{O}(\eps^{-2}\log(n))}$                                                                                    & \cite{FL21}                  \\ \cline{2-3} 
				\multicolumn{1}{|c|}{}                                  & $n^{\tilde{O}(\eps^{-2}(\log\log(n))^3)}$                                                      & \Cref{thm:MetricBakerMinor}  \\ \hline
		\end{tabular}
	}
	\end{minipage}
	\caption{\it On the left illustrated the running time of all the approximation schemes for the bounded-capacity vehicle routing problem. 
		On the right illustrated the running time of the  bicriteria approximation schemes for the metric $\rho$ dominating/independent set problems in their full generality.
		\label{tab:Application}}
\end{table}

In clan embeddings and Ramsey type embeddings, we observe that  the construction of \blind \cite{FL21}, in combination with the construction of Cohen-Addad \etal\cite{CFKL20}, can be seen as providing a reduction from  planar graphs with one vortex to $K_r$-minor-free graphs. Roughly speaking, the reduction implies that if   planar graphs with one vortex have an embedding with treewidth $t(\eps,n)$ and additive distortion $+\eps D$, then $K_r$-minor-free graphs has a clan embedding/Ramsey type embedding with treewidth $t(O_r(\frac{\delta \eps}{\log(n)}),n) + O_r(\log(n))$ and additive distortion $+\eps D$.  \Blind \cite{FL21} used the embedding of planar graphs with one vortex by Cohen-Addad \etal\cite{CFKL20}, which has treewidth  $t(\eps,n) = O(\frac{\log n}{\eps})$, to get a clan embedding/Ramsey type embedding of treewidth $O_r(\frac{\log n}{(\delta\eps)/\log(n)}) + O_r(\log n) = O_r(\frac{\log^2 n}{\delta\eps})$.  

 There are two fundamental barriers if one wants to improve the treewidth bound of $O_r(\frac{\log^2 n}{\delta\eps})$: improving the embedding of planar graphs with one vortex and removing the (both multiplicative and additive) loss $O(\log n)$ in the reduction step. Our embedding for planar graphs with one vortex (\Cref{lm:nearlyPlanar-emb}) overcomes the former barrier. Specifically, by plugging in our embedding for planar graphs with one vortex, we obtain treewidth $O_r(\frac{\log n(\log\log n)^2}{\eps\delta})$ in both clan embedding and  Ramsey type embedding; the improvement is from quadratic in $\log(n)$ to nearly linear in $\log(n)$. While we are not able to fully answer \Cref{q:clan-Ramsey-minor} due to the second barrier, we are one step closer to its resolution.

\begin{restatable}[Clan Embeddings of Minor-free Graphs into Low Treewidth Graphs]{theorem}{ClanMinorToTreewidth}
	\label{thm:MinorClan}
	Given a $K_r$-minor-free $n$-vertex graph $G=(V,E,w)$ of diameter $D$  and parameters $\eps\in(0,\frac14)$, $\delta\in(0,1)$, there is a distribution $\mathcal{D}$ over clan embeddings $(f,\chi)$ with additive distortion $+\eps D$ into graphs of treewidth $O_{r}(\frac{\log n(\log\log n)^2}{\eps\delta})$  such that for every $v\in V$, $\mathbb{E}[|f(v)|]\le 1+\delta$.
\end{restatable}

\begin{restatable}[Ramsey Type Embeddings of Minor-free Graphs into Low Treewidth Graphs]{theorem}{RamseyMinorToTreewidth}
	\label{thm:MinorRamsey}
	Given an $n$-vertex $K_r$-minor-free graph $G=(V,E,w)$ with diameter $D$ and parameters ${\eps\in(0,\frac14)}$, $\delta\in(0,1)$, there is a distribution over dominating embeddings $g:G\rightarrow H$ into treewidth $O_{r}(\frac{\log n(\log\log n)^2}{\eps\delta})$ graphs, such that there is a subset $M\subseteq V$ of vertices for which the following claims hold:
	\begin{enumerate}[noitemsep]
		\item For every $u\in V$, $\Pr[u\in M]\ge 1-\delta$.
		\item For every $u\in M$ and $v\in V$, $d_H(g(u),g(v))\le d_G(u,v)+\eps D$.
	\end{enumerate}
\end{restatable}

Using our new clan and Ramsey embeddings in \Cref{thm:MinorClan} and \Cref{thm:MinorRamsey} on top of the algorithm from \cite{FL21},  we improve the running time of the QPTAS for the metric Baker's problems in $K_r$-minor-free graphs \cite{FL21} from $n^{\tilde{O}(\eps^{-2}\log(n))}$ to $n^{\tilde{O}(\eps^{-2}(\log\log(n))^3)}$. 

\begin{restatable}[QPTAS for Metric Baker Problems in Minor-free Graphs]{theorem}{MetricBakerMinor}
	\label{thm:MetricBakerMinor} %$\rho$-independent/dominating set problems admit  a bicriteria  quasi-polynomial time approximation scheme in $K_r$-minor-free graphs. 	 Specifically, 
	Given an $n$-vertex $K_r$-minor-free graph $G(V,E,w)$, two parameters $\epsilon \in (0,1)$ and $\rho > 0$, and a vertex weight function $\mu: V\rightarrow \mathbb{R}^+$, one can find in $n^{\tilde{O}_r(\eps^{-2}(\log\log(n))^3)}$ time:\\
	(1) a $(1-\epsilon)\rho$-independent  set $I$ such that for every $\rho$-independent  set $\tilde{I}$, $\mu(I) \geq (1-\epsilon)\mu(\tilde{I})$, and\\
	(2) a  $(1+\epsilon)\rho$-dominating set $S$ such that for every $\rho$-dominating set $\tilde{S}$, $\mu(S) \leq (1+\epsilon)\mu(\tilde{S})$.
\end{restatable}

\subsection{Techniques}\label{subsec:technique}

A key technical component in all the embeddings 
%specifically, \Cref{thm:PlanarToTreewidth,thm:MinorToTreewidth,thm:MinorClan,thm:MinorRamsey},
%\atodo{We don't actually use it for \Cref{thm:MinorClan,thm:MinorRamsey}, only implicitly inside the black box.}\htodo{But it is still accurate to say that it was used though indirectly.}\atodo{I removed the specifically... It doesn't feel right to say specifically use when it is only indirectly... I am fine with the current formulation} 
is an emulator for trees with \emph{low treewidth} and \emph{low hop path} between every pair of vertices. 
For a path $P$, denote by $\hop(P)$ the number of edges in the path $P$. For a pair of vertices $u,v\in V$, and parameter $h\in\mathbb{N}$, denote by $$d_G^{(h)}(u,v)=\min\{w(P)\mid P\mbox{ is a }u-v\mbox{ path, and }\hop(P)\le h\}$$ the weight of a minimum $u-v$ path with at most $h$ edges. If no such path exist, $d_G^{(h)}(u,v)=\infty$.

Given an edge-weighted tree $T=(V_T,E_T,w_T)$, we say that an edge-weighted graph $K = (V_K,E_K,w_T)$ is an \emph{$h$-hop emulator} for $T$ if $V_K = V_T$ %\atodo{equal, not $\subseteq$?}\htodo{This is equal in our construction.} 
and $d^{(h)}_K(u,v) = d_T(u,v)$ for every pair of vertices $u,v \in V$. We show that every tree admits a low-hop emulator with small treewidth.

\begin{restatable}{theorem}{TreeSpanner}
	\label{thm:TreeSpanner}  Given an edge-weighted $n$-vertex tree $T$, there is an $O(\log\log n)$-hop emulator $K_T$ for $T$ that has treewidth $O(\log \log n)$. Furthermore, $K_T$ can be constructed in $O(n\cdot\log\log n)$ time.
\end{restatable}

Constructing a low hop emulator for trees is a basic problem that was studied by various authors~\cite{Chazelle87,AS87,BTS94,Thorup97,NS07,CG08,Solomon11,FL22} %\atodo{Could we avoid citing \cite{AS87}? it was never publish.} 
in different contexts. The goal of these constructions is to minimize the number of edges of the emulator for a given hop bound, with no consideration to the treewidth of the resulting graph.
To the best of our knowledge, the only previously constructed low-hop emulator with bounded treewidth is a folklore recursive construction \footnote{Consider the following construction: take the centroid $v$ of the tree (a vertex $v$ such that each connected component in $T\setminus\{v\}$ has at most $\frac23n$ vertices), add edges from $v$ to all the vertices, and recurse on each connected component of $T\setminus\{v\}$. The result is a $2$-hop emulator with treewidth $O(\log n)$.}
of a $2$-hop emulator with treewidth $O(\log n)$. Our emulator seeks to balance between the treewidth and the hop bound, and achieves an exponential reduction in the treewidth (at the expense of an increase in the hop bound). 
Due to its fundamental nature, we believe that \Cref{thm:TreeSpanner} will find more applications in other contexts.

%A key idea in the proof of \Cref{thm:TreeSpanner} is a novel decomposition of trees into subtrees by removing at most $\ell$ vertices such that each subtree has at most 2 boundary vertices and size  $O(\frac{n}{\ell})$. While similar decompositions of trees exist in the literature, e.g.~\cite{Chazelle87,AS87,BTS94,NS07,Solomon11}, they only have guarantees on the size of each subtree. In our construction, it is important for each subtree to have at most $2$ boundary vertices.

 %The extension to planar graphs with one vortex, a key technical barrier in the low treewidth embedding of minor-free graphs, will follow naturally.

\paragraph{Embedding  of planar graphs.~} To avoid notational clutter, we describe our embedding technique for planar graphs instead of addressing planar graphs with one vortex directly (which is the main motivation of our technique). We will use the low hop emulator in \Cref{thm:TreeSpanner} to obtain a simple and efficient construction as claimed in \Cref{thm:PlanarToTreewidth}. Our construction not only addresses \Cref{q:baker-planar}, but it is also robust enough to be easily extensible to planar graphs with one vortex, which, as discussed above, is the fundamental barrier towards an embedding of $K_r$-minor-free graphs with low treewidth.

Given a planar graph $G = (V,E,w)$ with diameter $D$, we construct a recursive decomposition $\Phi$ of $G$ using shortest path separators~\cite{Thorup04,Klein02}. Specifically, we pick a global root vertex $r\in V$. $\Phi$ has a tree structure of depth $O(\log n)$, and each node $\alpha \in \Phi$ is associated with a cluster of $G$, the ``border'' of which consists of $O(1)$ shortest paths rooted at $r$. This set of border shortest paths is denoted by $\mathcal{Q}_\alpha$.
For each node $\alpha$, all the vertices not in $\mathcal{Q}_\alpha$ are called \emph{internal vertices}, and denoted $I_\alpha$.
Each leaf note $\alpha\in\Phi$ contains only a constant number of internal vertices.

Following the standard techniques~\cite{EKM14,CFKL20}, for each path $Q \in \mathcal{Q}_\alpha$ associated with a node $\alpha \in \Phi$, one would place a set of equally-spaced vertices of size $O(\frac{\log(n)}{\eps})$, called \emph{portals}, such that the distance between any two nearby portals is $O(\eps D /\log(n))$. Since $|\mathcal{Q}_{\alpha}| = O(1)$, each node of $\Phi$ has $O(\log(n)/\eps)$ portals. We form a tree embedding of \emph{treewidth $O(\log(n)/\eps)$} from portals of $\Phi$ by considering every edge  $\{\alpha,\beta\}$ in $\Phi$ and adding all pairwise edges between portals of $\alpha$ and $\beta$. For the distortion argument, we consider any shortest path $P$ of $G$. As $\Phi$ has depth $O(\log n)$, $P$ intersects at most $O(\log n)$ shortest paths on the boundary of $O(\log n)$ nodes of $\Phi$. For any shortest path $Q\in \mathcal{Q}_\alpha$ that intersects $P$, we route $P$ through the portal closest to the intersection vertex of $P$ and $Q$, incurring an additive distortion of  $O(\eps D /\log(n))$.  This means that the total distortion due to routing through portals  is $O(\log n) \cdot O(\eps D /\log(n))  = O(\eps) D$. 

 Our key idea is to use \Cref{thm:TreeSpanner} to construct an $O(\log\log n)$-hop emulator $K_{\Phi}$ for $\Phi$ along with a tree decomposition $\mathcal{T}$ for $K_{\Phi}$ of width $\tw(\mathcal{T}) = O(\log\log n)$.  For each path $Q \in \mathcal{Q}_\alpha$ in each node $\alpha \in \Phi$,  we now only place $O(\frac{\log\log(n)}{\eps})$ portals such that the distance between any two nearby portals is $O(\eps D /\log\log(n))$. For each edge  $\{\alpha,\beta\}$ \emph{in $K_{\Phi}$}, as opposed to only consider edges in $\Phi$ in the standard technique, we add all pairwise edges between portals of $\alpha$ and $\beta$ to obtain the host graph $H$. By replacing each node $\alpha \in \Phi$ in each bag of $\mathcal{T}$ with the portals of $\alpha$,  we obtain a tree decomposition $\mathcal{X}_H$ of $H$.  Since $\mathcal{T}$ has width $O(\log\log n)$ and each node has $O(\log\log(n)/\eps)$ portals, the treewidth of $\mathcal{X}_H$ is $O(\log\log(n)) \cdot O(\log\log(n)/\eps) = O((\log\log n)^2/\eps)$.

 For the distortion argument, we crucially use the property that $K_{\Phi}$ has hop diameter $O(\log\log(n))$. Consider any shortest path $P$ of $G$, whose endpoints are internal vertices of two leaf nodes, say $\alpha_1$ and $\alpha_2$,  of $\Phi$. There is a (shortest) path consisting of $k = O(\log n)$ nodes from  $\alpha_1$ to $\alpha_2$ in $K_{\Phi}$. This means that it suffices to consider $k$ shortest paths $\{Q_1,\ldots, Q_{k}\}$ associated with these nodes intersecting $P$.  By the construction, there are  edges in $H$ between portals of $Q_i$ and $Q_{i+1}$ for every $1\leq i \leq k-1$. Thus, by routing an intersection vertex of $P\cap Q_i$ to the nearest portal of $Q_i$ for every $i\in [k]$, we obtain a path $P_H$ in $H$ between the endpoints of $P$. As the portal distance is  $O(\eps D /\log\log(n))$, each routing step incurs   $O(\eps D /\log\log(n))$ additive distortion. Since the number of paths is $k = O(\log\log n)$, the total distortion is $k \cdot O(\eps D /\log\log(n)) = O(\eps) D$, as desired.

 Conceptually,  one can think of our shortcutting technique as an effective way to reduce the recursion depth $O(\log n)$ by balanced separators to $O(\log\log(n))$. Since using balanced separators is a fundamental technique in designing planar graphs algorithms, we believe that our idea or its variants could be used to improve other algorithms where the depth  $O(\log n)$ of the recursion tree is the main barrier. 
 
\paragraph{Extension to $K_r$-minor-free graphs.~} Since planar graphs with one vortex have balanced shortest path separators~\cite{CFKL20}, the embedding algorithm for planar graphs described above is  readily applicable to these graphs; the treewidth of the embedding is $O((\log\log n)^2/\eps)$. If we  plug this embedding into the framework of Cohen-Addad \etal\cite{CFKL20}, we obtain a stochastic embedding of $K_r$-minor-free graphs into $O_r(\log\log n)^2\eps^{-2}) + O_r(\log n)$ treewidth graphs. The additive $ O_r(\log n)$ term is due to a reduction step involving \emph{clique-sums} (Lemma 19 in \cite{CFKL20}). %By a relatively simple argument, we can remove the $O_r(\log n)$ additive term, which ultimately gives \Cref{thm:MinorToTreewidth}. 
%\atodo{Can we say something else? This description is not informative, and also don't give us any points.}\htodoin{One thing we can write more is about removing the  $O_r(\log n)$ additive term. Saying more means that we describe  ~\cite{CFKL20} in more details, which gives an inaccurate impression that our contribution is incremental. The following paragraph is my attempt to  describe our ideal to remove $O_r(\log n)$ additive term}

Our idea to remove the $+O_r(\log n)$ additive term using Robertson-Seymour decomposition is as follows. (See \Cref{sec:addNotation} for the formal definitions of concepts in this paragraph.) By Robertson-Seymour decomposition, $G$ can be decomposed into a tree $\mathbb{T}$ of nearly embeddable graphs, called \emph{pieces}, that are glued together via \emph{adhesions} of constant size. Our new embedding of planar graphs with one vortex implies that nearly embeddable graphs have embeddings into treewidth   $O((\log\log n)^2/\eps)$. To glues these embeddings together, we first root the tree $\mathbb{T}$. Then for each piece $G_i$, we simply add the adhesion $J_0$ between $G_i$ and its parent bag $G_0$ to every bag in the tree decomposition of the host graph $H_i$ in the embedding of $G_i$; this increases the treewidth of $G_0$ by only $+O(h(r))$, and the overall treewidth of the final embedding is also increased only by $+O(h(r))$. 
	For the stretch, let $Q_{uv}$ be a shortest path between two vertices $u$ and $v$. We show by induction that for every vertex $u \in G_a$ and $v\in G_b$ for any two pieces $G_a,G_b$ such that $G_i$ is their lowest common ancestor, there is a vertex $x \in Q_{uv}\cap G_i$ and  $y \in Q_{uv}\cap G_i$, the distances between $u$ and $x$, and between $v$ and $y$ are preserved \emph{exactly}. Thus, the distortion for $d_G(u,v)$ is just $+\eps D$.  Since we apply this construction to \emph{every} piece of $\mathbb{T}$, $+\eps D$ is also the distortion of the final embedding. 
%\htodo{If we accept the above paragraph, I would go back and change \Cref{sec:MinorToTW}.}
%\atodo{I think it is good.}

\paragraph{Bounded-capacity VRP in planar graphs.~} For bounded-capacity VRP in planar graphs, the main ingredient is a rooted-stochastic embedding (see \Cref{def:rootedStochastic}), where given a root $r$ (the depot in the VRP problem), there is a distribution over embeddings such that every pair of vertices $u,v$, have an additive distortion of $\eps\cdot (d_G(r,u) + d_G(r,v))$ in expectation. Efficient construction of such embedding is the main bottleneck in constructing near-linear time PTAS for the bounded-capacity VRP.
The stochastic embedding was introduced by Becker \etal \cite{BKS19}. The main idea of their embedding is to randomly slice graphs into vertex-disjoint bands, construct an additive embedding into a bounded treewidth graph for each band, and finally combine the resulting embeddings. 
There are two steps in the algorithm of Becker \etal \cite{BKS19} that incur running time $\Omega(n^2)$. First, the algorithm uses the additive embedding of  Fox-Epstein \etal \cite{FKS19},  which has a large (but polynomial) running time. We improve this step by using the embedding in \Cref{thm:PlanarToTreewidth}. Second, for each band $B$, they construct a new graph $G_B$ by taking the union of the shortest paths of all pairs of vertices in $B$. Since these shortest paths may contain vertices not in $B$, the size of $G_B$ could be $\Omega(n)$; as a result, the total size of all subgraphs constructed from the bands is $\Omega(n^2)$. We improve this step by showing that it suffices to embed the graph induced by the band $G[B]$ plus one more vertex, and hence the total number of vertices of all subgraphs is $O(n)$. 
The key idea in our proof is to observe that if a vertex falls far from the boundary of the band it belongs to, it has small additive distortion w.r.t. all other vertices. We then carefully choose the parameters so that each vertex is successful in this aspect with probability $1-\eps$.
Interestingly, we obtain here a strong Ramsey-type guarantee: each vertex $v$, with probability $1-\eps$, has a small additive distortion w.r.t. all other vertices (instead of a small expected distortion w.r.t. each specific vertex).

\subsection{Related Work}\label{sec:related}
%\atodoin{Stopped here}
Different types of embeddings were studied for planar and minor free graphs. $K_r$-minor-free graphs embed into $\ell_p$ space with multiplicative distortion $O_r(\log^{\min\{\frac12,\frac1p\}}n)$ \cite{Rao99,KLMN04,AGGNT19,AFGN18,Fil20face}, in particular, they embed into $\ell_\infty$ of dimension $O_r(\log^2n)$
with a constant multiplicative distortion. They also admit spanners with multiplicative distortion $1+\eps$ and $\tilde{O}_r(\eps^{-3})$ lightness \cite{BLW17}. On the other hand, there are other graph families that embed well into bounded treewidth graphs.
Talwar \cite{Tal04} showed that graphs with doubling dimension $d$ and aspect ratio $\Phi$ \footnote{The aspect ratio of a metric space $(X,d)$ is the ratio between the maximum and minimum pairwise distances $\frac{\max_{x,y}d(x,y)}{\min_{x,y}d(x,y)}$.\label{foot:aspectRatio}}  stochastically embed into graphs with treewidth $\eps^{-O(d\log d)}\cdot\log^d\Phi$ and expected distortion $1+\eps$.
Similar embeddings are known for graphs with highway dimension $h$ \cite{FFKP18} (into treewidth $(\log \Phi)^{-O(\log^2\frac h\eps)}$ graphs), and graphs with correlation dimension $k$ \cite{CG12} (into treewidth $\tilde{O}_{k,\eps}(\sqrt{n})$ graphs).

Ramsey type embeddings of general graphs into ultrametrics and trees were extensively studied \cite{BFM86,BBM06,BLMN03,MN07,NT12,BGS16,ACEFN20,FL21}.
Clan embeddings were also studied for trees \cite{FL21}. Clan embeddings are somewhat similar to the previously introduced multi-embeddings \cite{BM04,Bar21} in that each vertex is mapped to a subset of vertices. However, multi-embeddings lack the notion of chief, which is crucial in all our applications. See also \cite{HHZ21copies}.
Finally, clan and Ramsey type embeddings were also studied in the context of hop-constrained metric embeddings \cite{HHZ21hop,Fil21}.

%\subsection{Organization}
%
%In Section~\ref{sec:tree-spanner} we construct a low hop emulator for trees as described in \Cref{thm:TreeSpanner}. In \Cref{sec:PlanarToTW}, we construct a low-treewidth embedding for planar graphs, and give several applications to metric Baker's problems and the bounded-capacity VPR in planar graphs. We extend the embedding to $K_r$-minor-free graphs in \Cref{sec:MinorToTW}. We also presented an application of the embedding to the bounded-capacity VPR in minor-free graphs in this section. In \Cref{sec:RamseyClan}, we construct clan embedding and Ramsey type embedding for minor-free graphs, and show how to obtain an improved QPTAS for metric Baker's problems in minor-free graphs. In \Cref{sec:TwLb}, we present lower bounds for additive and multiplicative embeddings of planar graphs. 
%\atodoin{There was an organization section here. I removed it as it is redundant. I added more references to the table of contents so it should be enough.}

	\section{Preliminaries} \label{sec:prelim}
	$\tilde{O}$ notation hides poly-logarithmic factors, that is $\tilde{O}(g)=O(g)\cdot\polylog(g)$, while 
	$O_r$ notation hides factors in $r$, e.g. $O_r(m)=O(m)\cdot f(r)$ for some function $f$ of $r$. All logarithms are at base $2$ (unless specified otherwise).
	
	A (quasi-)polynomial time approximation scheme, or PTAS (QPTAS),  for a problem is a family of algorithms that, given any \emph{fixed} $\eps  > 0$, there is an algorithm $A_{\eps}$ in the family that runs in \mbox{(quasi-)}polynomial time and finds a $(1+\eps)$-approximation for the problem. A PTAS is \emph{efficient} if the running time of $A_{\eps}$ is of the form $f(\eps)n^{O(1)}$ for some function $f(\eps)$ that only depends on $\eps$. 
	
	We consider connected undirected graphs $G=(V,E)$ with edge weights $w_G: E \to \R_{\ge 0}$.
	A graph is called unweighted if all its edges have unit weight.
	Additionally, we denote $G$'s vertex set and edge set by $V(G)$ and $E(G)$, respectively. Often, we will abuse notation and write $G$ instead of $V(G)$.
	Let $T(V_T,E_T)$ be a tree. For any two vertices $u,v \in V_T$, we denoted by $T[u,v]$ the unique path between $u$ and $v$ in $T$.
	$d_{G}$ denotes the shortest path metric in $G$, i.e., $d_G(u,v)$ is the shortest distance between $u$ to $v$ in $G$. 
	The \emph{metric complement} of $G(V,E,w)$ is an edge-weighted complete graph with vertex set $V$ and the weight of each edge is the distance between its endpoints in $G(V,E,w)$.
	Note that every metric space can be represented as the shortest path metric of a weighted complete graph. We will use the notions of metric spaces, and weighted graphs interchangeably.
	When the graph is clear from the context, we might use $w$ to refer to $w_G$, and $d$ to refer to $d_G$. 	$G[S]$ denotes the induced subgraph by $S$.

	\begin{definition}[Tree decomposition]\label{def:tree-decomp}
		 A tree decomposition of a graph $G(V,E)$ is a tree $\mathcal{T}(\mathcal{B}, \mathcal{E})$ where $\mathcal{B}$ is a collection of subset of vertices of $V$, called \emph{bags}, that satisfies the following conditions:
		\begin{enumerate}[noitemsep]
			\item $\cup_{B\in \mathcal{B}}B  = V$.
			\item For every edge $(u,v)\in E$, there is a bag $B\in \mathcal{B}$ such that $\{u,v\}\subseteq B$.
			\item For every vertex $u\in V$, the set of bags $\{B\in \mathcal{B}: u \in B\}$ containing $u$ induces a connected subtree of $\mathcal{T}(\mathcal{B}, \mathcal{E})$.
		\end{enumerate}
	\end{definition}
	The \emph{width} of a tree decomposition $\mathcal{T}$ is $\max_{ i \in V(\mathcal{T})}|X_i| -1$ and the treewidth of $G$, denoted by $\tw$, is the minimum width among all possible tree decompositions of $G$. A \emph{path decomposition} of a graph $G(V,E)$ is a tree decomposition where the underlying tree is a path. The pathwidth of $G$, denoted by $\pw$, is defined accordingly.
	
	 A metric embedding is a function $f:X\rightarrow Y$ between the points of two metric spaces $(X,d_X)$ and $(Y,d_Y)$.
	 A metric embedding $f$ is said to be \emph{dominating} if for every pair of points $x,y\in X$, it holds that $d_X(x,y)\le d_Y(f(x),f(y))$. 
	 We say the a metric embedding $f$ has additive distortion $+\Delta$ if $f$ is dominating, and for every $x,y\in X$, $d_Y(f(x),f(y))\le d_X(x,y)+\Delta$.
	 Stochastic embedding is a distribution $\mathcal{D}$ over dominating embeddings $f: V(G)\rightarrow V(H_f)$. Stochastic embedding $\mathcal{D}$ has expected additive distortion $+\Delta$, if for every pair of points $x,y\in X$, it holds that $\mathbb{E}_{f\sim\mathcal{D}}[d_{H_f}(f(x),f(y))]\le d_G(x,y)+\Delta$.

\paragraph{Bounded-Capacity Vehicle Routing Problem.~} A \emph{tour} $R$ in a graph $G(V,E,w)$ is a closed walk, and the cost of $R$ is $\cost(R) = \sum_{e \in E(R)}x_e w(e)$ where $x_e$ is the number of times that $R$ visits edge $e$. 

\begin{definition}[Vehicle Routing Problem (VRP)]\label{def:vehicle-routing} An instance of the  vehicle routing problem (VRP) is a tuple $(G(V,E,w),K,r,Q)$ where $K\subseteq V$ is a set of \emph{clients}, $r\in V$ is the \emph{depot}, and $Q\in \mathbb{Z}_+$ is the \emph{capacity} of the vehicle. A feasible solution of an instance of the VRP, denoted by $\mathcal{S}$ ,is a tuple $(\mathcal{K}, \mathcal{R}, q)$ where:
	\begin{itemize}
		\item $\mathcal{K} = \{K_1,K_2,\ldots, K_q\}$ is a family of $q$ subset of terminals such that $\cup_{j=1}^q K_j = K$ and $|K_j| \leq Q$ for every $j \in [q]$.%\atodo{This should be a union rather than a sum?}
		
		\item $\mathcal{R} = \{R_1, \ldots, R_q\}$ is the set of $q$ tours such that each tour $R_j$ (a) starts and ends at the depot $r$ and (2) visits all terminals in $K_j$ for every $j \in [q]$.  
	\end{itemize}
	The cost of $\mathcal{S}$ is: $\cost(\mathcal{S}) = \sum_{j=1}^q \cost(R_j)$. The goal is to find a solution with minimum cost for the VRP in $G(V,E,w)$.
\end{definition}

\section{Spanners for Trees of Low Treewidth and Hop Diameter}\label{sec:tree-spanner}

In this section, we focus on proving~\Cref{thm:TreeSpanner} that we restate here for convenience. 

\TreeSpanner*

Our proof of  \Cref{thm:TreeSpanner} uses the following lemma, which states that in every tree, we can remove a small number of vertices, such that every remaining connected component has at most two boundary vertices. This lemma was used, sometimes implicitly, in the literature~\cite{AHLT05,FLSZ18,NPSW22}. We include a proof for completeness. 

\begin{lemma}\label{lm:tree-division} Given any parameter $\ell\in\mathbb{N}$ and an $n$-vertex tree $T$, then in $O(n)$ time we can find a subset $X$ of at most $\frac{2}{\ell+1}n-1=O(\frac{n}{\ell})$ vertices such that every connected component $C$ of $T\setminus X$ is of size $|C|\le \ell$, and $C$ has at most two outgoing edges towards $X$.
\end{lemma}
\noindent\emph{Proof.}
We begin the proof of \Cref{lm:tree-division} with the following claim:
\begin{claim}\label{clm:ChopToEllPieces}  
	Given an $n$-vertex tree $T$, in $O(n)$ time we can find a set $A$ of at most $\frac{n}{\ell+1}$ vertices such that every connected component in $T\setminus A$ contains at most $\ell$ vertices.
\end{claim}	
\begin{wrapfigure}{r}{0.18\textwidth}
	\begin{center}
		\vspace{-18pt}
		\includegraphics[width=0.95\textwidth]{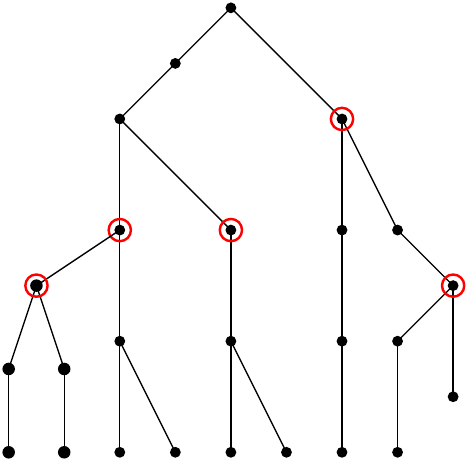}
		\vspace{-10pt}
	\end{center}
	\vspace{-10pt}
\end{wrapfigure}
\noindent\emph{Proof.}
	If $n\le\ell$ then we are done. Otherwise, fix an arbitrary root $r$, and visit the tree in a bottom-up fashion from the leaves towards the root until we find the first vertex $v$ with subtree $T_v$ of size at least $\ell+1$. We then add $v$ to $A$, and continue the process on $T\setminus T_v$. By induction, we will add at most  $\frac{|T\setminus T_v|}{\ell+1}\le \frac{n-(\ell+1)}{\ell+1}=\frac{n}{\ell+1}-1$ additional vertices to $A$. It follows that $|A|\le \frac{n}{\ell+1}$ as required.  $A$ can be constructed in linear time by a simple tree traversal.
	In the illustration on right $\ell=3$, and the vertices in $A$ are encircled in red.
	\qed

\vspace{10pt}We continue the proof of \Cref{lm:tree-division}. We first apply \Cref{clm:ChopToEllPieces} to obtain a set $A$ of size   $\frac{n}{\ell+1}$ in $O(n)$ time such that every connected component of $T\setminus A$ has at most $\ell$ vertices.  The set $X$ contains $A$ and the least common ancestor of every pair of vertices in $A$. $X$ is called the \emph{least common ancestor closure} of $A$ (see \cite{FLSZ18}). Given $A$, $X$ can be found in $O(n)$ time by traversing $T$ in post order. Clearly, every connected component $C$ of $T\setminus X$ has size at most $\ell$. 
Note that if $C$ has outgoing edges towards $x,y\in X$, then necessarily $x$ is an ancestor of $y$, or vice versa (as otherwise their least common ancestor will be in $C$, but this is a contradiction as $C\cap X=\emptyset$).
As it is impossible that $C$ will have three outgoing edges towards $x,y,z\in X$ where $x$ is  an ancestor of $y$, and $y$ is an ancestor of $z$, we conclude that $C$ has at most two outgoing edges towards $X$. 

Finally, we argue by induction on $n$ that $|X|\le2|A|-1$; the lemma will then follow as $2|A|-1 \leq \frac{2}{\ell+1}n-1$.
The base of the induction when $|A|=1$ is clear.
Consider a tree $T$ rooted at a vertex $r$ with children $v_1,\dots,v_s$ whose subtrees are denoted by  $T_1,\dots,T_s$, respectively. First, assume that $r\notin A$.
If all the vertices of $A$ belong to a single subtree $T_i$, then $|X|\le2|A|-1$ by the induction hypothesis on $T_i$.
Otherwise, suppose w.l.o.g. that $A$ intersects $T_1,\dots,T_{s'}$ where $s'\ge 2$.
Then $X$ will contain the root $r$, along with some vertices in each  subtree $T_i$ for $i\in[1,s']$.  Using the induction hypothesis we conclude that: 
$$|X|\le 1+\sum_{i=1}^{s'}(2|T_i\cap A|-1)=2|A|+1-s'\le2|A|-1~.$$
The case where $r\in X$ is similar.
\qed

\vspace{10pt}
We are now ready to prove~\Cref{thm:TreeSpanner}.
\begin{proof}[Proof of~\Cref{thm:TreeSpanner}]
	Fix $c=\frac{2}{\log\frac{3}{2}}$.
	We will prove by induction on $n$, that every $n$-vertex tree $T$ for $n\ge 2$ has a $(1+c\cdot\log\log n)$-hop emulator $K_T$, where $K_T$ has treewidth at most $1+c\cdot \log\log n$. 
	The base case is when $n=1$; we simply return the original tree $K_T=T$ (a singleton emulator with treewidth $0$).
	For the general case, using \Cref{lm:tree-division} with parameter $\ell=\sqrt{2n}-1$,\footnote{For the sake of simplicity, we will ignore integrally issues.} we obtain a set  $X$ of at most $\frac{2}{\ell+1}n-1= \ell$ vertices, such that $T\setminus X$ has a set of connected components $\mathcal{C}$, where every $C\in \mathcal{C}$ is of size $|C|\le \ell$, and has at most two outgoing edges towards $X$.
	Note that for $n\ge 1$, $|X|\leq \ell=\sqrt{2n}-1<n^{\frac23}$.
	For each such connected component $C$, let $T_C=T[C]$ be the induced subtree, and $K_C$ be the emulator constructed using the induction hypothesis, with the tree decomposition $\mathcal{T}_C$.
	Next, we create a new tree $T_X$ with $X$ as a vertex set as follows. For every connected component $C$ with two outgoing edges towards   two vertices $u,v\in X$, we add an edge $\{u,v\}$ to $T_X$ of weight $d_T(u,v)$. One can easily verify that $T_X$ is indeed a tree (in fact it is a minor of $T$), and furthermore, for every two vertices $x,y\in X$, $d_{T_X}(x,y)=d_{T}(x,y)$. We recursively construct  an emulator $K_X$ along with a tree decomposition $\mathcal{T}_X$ for $T_X$.

	We construct the emulator $K_T$ for $T$ as follows: $K_T$ contains the emulator $K_X$, the union of all the emulators $\cup_{C\in \mathcal{C}}K_C$, and in addition, for every $C\in \mathcal{C}$ with outgoing edges towards $\{u,v\}$ (we allow $u=v$), and for every vertex $z\in C$, we add edges $\{z,u\},\{z,v\}$ to $K_T$.
	
	First, we argue that $K_T$ has low hop. Consider a pair of vertices $u,v$ in $T$, and suppose that $u\in C^u$ and $v\in C^v$ where $C^u$ and $C^v$ are in $\mathcal{C}$. (The cases where either $u$ or $v$ is in $X$ follows by the same argument.) 
	If $C^v=C^u$, then by the induction hypothesis on $K_{C^u}$, $d_{K_T}^{(1+c\cdot \log\log n)}(u,v)\le d_{K_{C^{u}}}^{(1+c\cdot \log\log |C_{u}|)}(u,v)=d_{T_{C^{u}}}(u,v)=d_{T}(u,v)$.
	Else, the unique $u-v$ path $P$ in $T$ must go though $X$. Let $x_u,x_v\in P\cap X$ be the closest vertices to $u$ and $v$, respectively. Since $K_T$ contains the edges $\{u,x_u\}$, $\{v,x_v\}$, by  the induction hypothesis on $K_{X}$, we have
	\[
	d_{K_T}^{(3+c\cdot\log\log|X|)}(u,v)\le d_{T}(u,x_{u})+d_{K_{X}}^{(1+c\cdot\log\log|X|)}(x_{u},x_{v})+d_T(x_{v},v)=d_{T}(u,v)~.
	\]
	Observe that $
	3+c\cdot\log\log|X|\le3+c\cdot\log\log n^{\frac{2}{3}}=3+c\cdot\log\frac{2}{3}+c\cdot\log\log n=1+c\cdot\log\log n$. This means that $d_{K_T}^{(1+c\log\log(n))}(u,v) \leq d_{K_T}^{(3+c\cdot\log\log|X|)}(u,v) \leq d_{T}(u,v)$, as desired.   
	
	Next, we argue that $K_T$ has treewidth at most $1+c\cdot\log\log n$.
	We define a tree-decomposition $\mathcal{T}$ as follows: 
	For every cluster $C\in \mathcal{C}$ with outgoing edges towards $\{u,v\}$ (we allow $u=v$), let $\mathcal{T}_C$ be the tree decomposition obtained by the induction hypothesis on $T[C]$. Let $\mathcal{T}'_C$ be the tree decomposition obtained from $\mathcal{T}_C$ by adding the vertices $u,v$ to every bag. 	Consider the tree decomposition $\mathcal{T}_X$ of $K_X$. Since $K_X$ contains the edge $\{u,v\}$, there is a bag $B_{u,v}\in \mathcal{T}_X$ containing both $u$ and $v$. We   add an edge in $\mathcal{T}$ between $B_{u,v}$ and an arbitrary bag in $\mathcal{T}_C$.
	It follows directly from the construction that $\mathcal{T}$ is a valid tree decomposition of $K_T$. Furthermore, the width of the decomposition is bounded by 
	\[
	\max\left\{ 1+c\cdot\log\log|X|,3+c\cdot\log\log\ell\right\} \le1+c\cdot\log\log n~,
	\]
	as required. 
	
	Finally, we bound the running time. Note that by \Cref{lm:tree-division}, we can find $X$ in $O(n)$ time. Constructing the set of connected components $\mathcal{C}$ and the tree $T_X$ takes $O(n)$ time. To find the weights for edges of $T_X$, we use an exact distance oracle for trees with construction time $O(n)$ and query time $O(1)$. It is folklore that such a distance oracle   can be constructed by a simple reduction to the least common ancestor (LCA) data structure~\cite{BF00} (see also~\cite{HT84,SV88}). Constructing $K_T$ from $K_X$ and $\{K_C\}_{C\in\mathcal{C}}$ takes $O(n)$ time as well.   As the depth of the recursion is $O(\log\log n)$, and the total running time of each level is $O(n)$, the overall running time is $O(n\log\log n)$.
\end{proof}

\section{Embedding Planar Graphs and Applications (Proof of \Cref{thm:PlanarToTreewidth})} \label{sec:PlanarToTW}

In this section, we focus on proving~\Cref{thm:PlanarToTreewidth} that we restate here for convenience. 

\PlanarToTreewidth*

Throughout, we fix a planar drawing of $G$. W.l.o.g, we assume that $G$ is triangulated; otherwise, we can triangulate $G$ in linear time and set the weight of the new edges to be $+\infty$.   Let $r$ be an (arbitrary) vertex of $G$ and $T_r$ be a shortest path tree rooted at $r$. We say that a shortest path $Q$ in $G$ is an \emph{$r$-path} if $Q$ is a path in $T_r$ with $r$ as an endpoint.

We start by defining  \emph{$\eta$-rooted shortest path decomposition}.
\begin{definition}[$\eta$-rooted shortest path decomposition  ($\eta$-\RSPD)]\label{def:RSPD}
	An \emph{$\eta$-rooted shortest path decomposition} with root $r$, denoted by $\Phi$ of a graph is a binary tree with the following properties:
	\begin{itemize}[noitemsep]
		\item[\textbf{(P1.)}] $\Phi$ has height $O(\log n)$ and $O(n)$ nodes.
		\item[\textbf{(P2.)}] Each node $\alpha \in \Phi$ is associated with a subgraph $X_{\alpha}$ of $G$, called a \emph{piece}, such that:
		\begin{enumerate}			
			\item[(a)] $X_{\alpha}$ contains at most $\eta$ shortest paths rooted at $r$, called \emph{boundary paths}. Denote by $\mathcal{Q}_\alpha$ be the set of all boundary paths.
			We will abuse notation and denote by $\mathcal{Q}_\alpha$ the union of all the vertices in all the boundary paths.
			Every vertex in $X_\alpha$ which does not lie on a boundary path, is called an \emph{internal vertex}.
			\item[(b)]  If $\alpha$ is the root of $\Phi$, then $X_{\alpha} = G$; if $\alpha$ is a leaf of $\Phi$, then $X_{\alpha}$ has at most $\eta$ internal vertices. Otherwise, $\alpha$ is an internal node with exactly two children $\beta_1,\beta_2$. It holds that $X_{\alpha} = X_{\beta_1}\cup X_{\beta_2}$ and $X_{\beta_1}\cap X_{\beta_2}$ are the boundary paths shared by $X_{\beta_1}$ and $X_{\beta2}$.
			
			\item[(c)]For any vertex $u\in V(X_{\alpha})$ and $v \in V(G)\setminus V(X_{\alpha})$, (any) path between $u$ and $v$ must intersect some vertex that lies on a boundary path of $X_{\alpha}$. 
		\end{enumerate}
	\end{itemize}
\end{definition}
In planar graphs, an $\eta$-\RSPD is also known as a \emph{recursive decomposition} with shortest path separators, see, e.g.,~\cite{AGKKW98,Klein02,KKS11}.  Abraham \etal \cite{AFGN18} (see also \cite{Fil20face}) defined a related notion of shortest path decomposition (SPD). There are several diffences between \RSPD and \SPD:  \SPD the height of the tree is a general parameter $k$; the number of boundary paths in each node is equal to its depth (distance to the root); the tree is not necessarily binary;  and the leaves contain no internal node. \SPD  was used to construct multiplicative embeddings into $\ell_1$ \cite{AFGN18,Fil20face} and scattering partitions \cite{Fil20scattering}. 

The following observation follows directly from the definition.

\begin{observation}\label{obs:P3}Let $\Phi$ be an $\eta$-\RSPD  of $G$. Then $G = \cup_{\alpha \in \leavs(\Phi)}X_{\alpha}$ where $\leavs(\Phi)$ is the set of leaves of the tree $\Phi$.
\end{observation}

For any two nodes $\alpha,\beta \in \Phi$, let $\Phi[\alpha,\beta]$ be the subpath between $\alpha$ and $\beta$ in $\Phi$.
A crucial property of of an \RSPD in our construction is the \emph{separation property}: for any path $P_{uv}$ between two vertices $u\in \mathcal{Q}_{\alpha}$ and $v\in \mathcal{Q}_{\beta}$ of two given nodes $\alpha$ and $\beta$, for any node $\lambda$ on the path between $\alpha$ and $\beta$ in the tree $\Phi$,  $P_{uv}\cap \mathcal{Q}_{\lambda} \not= \emptyset$.

An explicit representation of $\Phi$ could take $\Omega(n^2)$ bits and hence is not computable in nearly linear time. However, $\Phi$ has a \emph{compact representation } that only takes $O(n\eta)$ words of space. Specifically, we store at each node $\alpha  \in \Phi$  at most $\eta$ vertices $B_{\alpha} = \{v_1,v_2,\ldots, v_{\eta'}\}$  such that $\{T_r[r,v_i]\}_{i=1}^{\eta'}$ is the set of $r$-paths $\mathcal{Q}_{\alpha}$. If $\alpha$ is a leaf node of $\Phi$, then $\alpha$  is associated with an extra set of $O(1)$ vertices, denoted by $I_{\alpha}$, that are internal vertices of $X_\alpha$.  Using shortest path separators, Thorup~\cite{Thorup04} showed that one can compute a  compact representation  of an $O(1)$-\RSPD of $G$ in time $O(n\log n)$. (See Section 2.5 in~\cite{Thorup04} for details; the \RSPD is called \emph{frame separator} decomposition of $G(V,E,w)$ in Thorup's paper.)

\begin{lemma}[Thorup~\cite{Thorup04}]\label{lm:comp-recursive-decomp} Given a planar graph $G(V,E,w)$, a (compact representation  of a) $O(1)$-\RSPD $\Phi$ of $G(V,E,w)$ can be computed in $O(n\log n)$ time.
\end{lemma}

In~\Cref{subsec:embedd_AFversion} below, we present an embedding of $G(V,E,w)$ into a low-treewidth graph in nearly linear running time. In \Cref{subsec:rootedStoEmb}, we construct a rooted stochastic embedding for planar graphs, which is used for approximating the vehicle routing problem. In \Cref{subsesc:apps}, we present applications of the two embeddings in designing almost linear time PTASes. 

\subsection{The Embedding Construction }\label{subsec:embedd_AFversion} 

We refer readers to \Cref{subsec:technique} for an overview of the proof. We will construct a one-to-many embedding $f: V\rightarrow 2^{V(H)}$ with additive distortion $+O(\epsilon)\cdot D$; we can recover distortion $\epsilon D$ by scaling $\eps$. By  \Cref{def:one-to-many}, we are required to guarantee that for any two copies of a single vertex $v_1,v_2 \in f(v)$, $d_H(v_1,v_2) = O(\epsilon)\cdot D$. In the end, we can transform $f$ to a one-to-one embedding $f'$ by picking an arbitrary copy $v' \in f(v)$ and set $f'(v) = v'$ for each vertex $v\in V(G)$.

Recall that $T_r$ is a shortest path tree rooted at $r$ of $G(V,E,w)$.  Since $G(V,E,w)$ has diameter $D$, $T_r$ has radius $D$. Let $\delta = \frac{\eps D}{\log\log n}$. We first define a set of vertices $N_r$ called \emph{$\delta$-portals} of $T_r$ as follows: initially, $N_r$ only contains $r$; we then visit every vertex of $T_r$ in the depth-first order, and we add a vertex $v$ to $N_r$ if the nearest ancestor of $v$ in $N_r$ is at a distance larger than $\delta$ from $v$. For each $r$-path $Q$ in $T_r$, we denote by  $N(Q,\delta) = N_r\cap V(Q)$ the set of $\delta$-portals in $Q$.
Note that for every vertex $v\in Q$, there is a $\delta$-portal $u\in N(Q,\delta)$ at a distance at most $\delta$. Furthermore, the distance between a pair of consecutive $\delta$-portals of $Q$ is greater than $\delta$. As the length of $Q$ is bounded by $D$, we conclude that
\begin{equation*}
	|N(Q,\delta)| \leq 
	\left\lfloor\frac{D}{\delta}\right\rfloor\le \frac{\log\log n}{\epsilon}~.
\end{equation*}

The construction works  as follows: let $\Phi$ be an $O(1)$-\RSPD of $G(V,E,w)$ given by \Cref{lm:comp-recursive-decomp}.
For every node $\alpha\in \Phi$, let $Q_1,Q_2\dots$, be the $\eta = O(1)$ $r$-paths constituting $\mathcal{Q}_\alpha$. Denote by $P_\alpha=\cup_{Q\in\mathcal{Q}_\alpha}N(Q,\delta)$ the union of all $\delta$-portals from all the $r$-paths on the boundary of $\alpha$. Note that $|P_\alpha|=O(\eta\cdot\frac{\log\log n}{\epsilon})=O(\frac{\log\log n}{\epsilon})$. 
Using \Cref{thm:TreeSpanner}, we construct an $O(\log\log n)$-hop emulator $K_\Phi$ for $\Phi$ with treewidth $O(\log\log n)$.  We add the following sets of edges to the host graph $H$:
\begin{enumerate}
	\item For every node $\alpha\in \Phi$, and vertex $v\in P_\alpha$, create a copy $\widetilde{v_\alpha}$. Denote this set of copies by $\widetilde{P_\alpha}$.
	Add the edge set $\widetilde{P_\alpha}\times \widetilde{P_\alpha}$ to $H$ (i.e. a clique on the copies).
	\item For every edge $\{\alpha,\beta\}\in K_\Phi$, add the edge set $\widetilde{P_\alpha}\times \widetilde{P_\beta}$ to $H$ (i.e. a bi-clique between the respective copies).
	\item For every leaf node $\alpha\in \Phi$, let $I_\alpha$ be the set of internal vertices. We add to $H$ the vertices in $I_\alpha$  and two edge sets $I_\alpha\times I_\alpha$ and $I_\alpha\times \widetilde{P_\alpha}$.
	\item Denote by $I_{\Phi}=\cup_{\alpha\in\Phi:\alpha\text{ is leaf}}I_{\alpha}$ the set of all vertices that belong to the interior of leaves of $\Phi$. For $v\notin I_\Phi$, by \Cref{obs:P3}, it necessarily belongs to some $r$-path $Q_v$ on the boundary of some leaf node $\alpha_v$. We add $v$ and the edge set $v\times \widetilde{P_{\alpha_v}}$ to $H$.	
\end{enumerate}
In summary, the set of vertices of $H$ is $V_H=V_G\cup\bigcup_{\alpha\in \Phi}\widetilde{P_\alpha}$, while the set of edges is
\begin{align*}
	E_{H} & =\bigcup_{\alpha\in\Phi}\left\{ \left(\widetilde{v_{\alpha}},\widetilde{u_{\alpha}}\right)\mid v,u\in P_{\alpha}\right\} \cup\bigcup_{\left(\alpha,\beta\right)\in\Phi}\left\{ \left(\widetilde{v_{\alpha}},\widetilde{u_{\beta}}\right)\mid v\in P_{\alpha},u\in P_{\beta}\right\} \\
	& \quad\cup\bigcup_{\alpha\in\Phi:\alpha\text{ is leaf}}\left\{ \left(v,\widetilde{u_{\alpha}}\right)\mid v\in I_{\alpha},u\in P_{\alpha}\cup I_{\alpha}\right\} \cup\bigcup_{v\notin I_{\Phi}}\left\{ \left(v,\widetilde{u_{\alpha_{v}}}\right)\mid u\in P_{\alpha_{v}}\right\} 
\end{align*}

The image $f(v)$ of every vertex $v$ consists of $v$ itself (added to $H$ in either Step 3 or Step 4), and all the respective copies (added in Step 1, one per each node $\alpha$ such that $v\in P_\alpha$).
The weights of the edges in $H$ are defined in a natural way: for every edge $(u',v')$ in $H$, $w_H(u',v') = d_G(u,v)$ where $u = f^{-1}(u')$ and $v = f^{-1}(v')$. 

\begin{observation}\label{obs:Ialpha-Palpha} Foor any leaf $\alpha$ in $\Phi$, the image of $I_{\alpha}\cup P_{\alpha}$ induces a clique in $H$. 
\end{observation}

Next, we bound the treewidth of $H$.

\begin{lemma}\label{lm:tw-planar-emb-new} $H$ has treewidth $O(\epsilon^{-1}(\log\log n)^2)$.
\end{lemma}
\begin{proof}
	Let $\mathcal{T}_\Phi(\mathcal{X}_\Phi,\mathcal{E}_\Phi)$ be a tree decomposition of $K_\Phi$ of  width $O(\log\log n)$.
	We create a tree-decomposition $\mathcal{T}$ for $H$ as follows: for every node $\alpha\in\Phi$, and every bag in $\mathcal{T}_\Phi$ containing $\alpha$, we replace $\alpha$ with $\widetilde{P_\alpha}$.
	Next, for leaf node $\alpha$, let $B$ be some arbitrary bag containing $\alpha$. We create a new bag $B_\alpha$, containing the vertices $I_\alpha\cup \widetilde{P_\alpha}$, with a single edge in $\mathcal{T}$ between $B$ and $B_{\alpha}$.
	Finally, for every vertex $v\notin I_\Phi$, let $B_{\alpha_v}$ be some bag containing $\alpha_v$, the leaf node containing $v$. We create a new bag $B_{v}$ containing the vertex $v$ and the set  $\widetilde{P_\alpha}$. We then add a single edge in $\mathcal{T}$ between $B$ and $B_{\alpha_v}$.
	
	As each bag contains at most $O(\log\log n)$ nodes from $\Phi$, and for every $\alpha\in\Phi$, $|\widetilde{P_\alpha}|=O(\eps^{-1}\cdot\log\log n)$, it follows that the width of $\mathcal{T}$ is bounded by $O(\frac{(\log\log n)^2}{\eps})$.
	Finally, it is straightforward to verify that $\mathcal{T}$ is a valid tree decomposition for $H$:  every edge of $H$ is contained in some bag, and  as every copy of every vertex $v$ is associated with a unique node $\alpha\in \Phi$, the connectivity condition  is satisfied.
\end{proof}

We now bound the distortion of the embedding. 
Specifically, we will show that for every pair of vertices $u,v$, and every two copies $u'\in f(u)$, $v'\in f(v)$, $d_G(u,v)\le d_H(u',v')\le d_G(u,v)+O(\eps)\cdot D$.
The lower bound follows directly from the way we assign weights to the edges of $H$. 
Indeed, every $u'$-$v'$ path in $H$ corresponds to a $u$-$v$ walk in $G$ of the same weight.  
Henceforth, we focus on proving the upper bound. 
First, we need the following lemma, which is the generalization of the separation property.

\begin{lemma}\label{lm:separation-planar} Let $\alpha$ and $\beta$ be two nodes in $\Phi$.  Let $P_{uv}$ be any path between two vertices $u$ and $v$ in $G$ such that $u \in \mathcal{Q}_{\alpha}$ and $v\in \mathcal{Q}_{\beta}$. Let $(\alpha=\lambda_1, \lambda_2, \ldots, \lambda_k = \beta)$ be a set of nodes in $\Phi[\alpha,\beta]$ such that $\lambda_{i+1} \in \Phi[\lambda_i,\beta]$ for any $1\le i\le k-1$. Then, there exists a sequence of vertices $(u=x_1, x_2, \ldots, x_k = v)$ such that $x_i \in P_{uv}\cap \mathcal{Q}_{\lambda_i}$ and $x_{i+1}\in P_{uv}[x_i,v]$ for any $1 \leq i \leq k-1$.
\end{lemma}
\begin{proof}
	The proof is by induction. If $k \leq 2$, then the lemma trivially follows. Henceforth, we assume that $k \geq 3$. Observe that $\lambda_2$ is either an ancestor of $\lambda_1$ and/or $\lambda_k$. Assume first that $\lambda_2$ is an ancestor of $\lambda_1$. Then by property P2(b) in~\Cref{def:RSPD}, $u\in X_{\lambda_2}$ and $v$ either belongs to $\mathcal{Q}_{\lambda_2}$ or $v \in V(G)\setminus V(X_{\lambda_2})$. In both cases, $P_{uv}\cap \mathcal{Q}_{\lambda_2} \not= \emptyset$. Let $x_2$ be the last vertex in the path $P_{uv}$ going from $u$ to $v$ such that $x_2 \in P_{uv}\cap \mathcal{Q}_{\lambda_2}$. By applying the induction hypothesis to $P_{x_2v}  = P_{uv}[x_2,v]$, there exists a sequence of vertices $(x_2, \ldots, x_k = v)$ such that $x_i \in P_{x_2v}\cap \mathcal{Q}_{\lambda_i}$ and $x_{i+1}\in P_{x_2v}[x_i,v]$ for any $2 \leq i \leq k-1$. Thus, $(u=x_1, x_2, \ldots, x_k = v)$ is the sequence of vertices claimed by the lemma.
	The case where $\lambda_2$ is an ancestor of $\lambda_k$ follows by the same argument.
\end{proof}

Next, we show that the distortion of portal vertices on the boundaries of nodes in $\Phi$ is in check.

\begin{lemma}\label{lm:boundary-distortion-planar}  Let $\alpha$ and $\beta$ be two nodes in $\Phi$.  Let $\widetilde{u_\alpha} \in \widetilde{P_\alpha}$ and $\widetilde{v_\beta} \in \widetilde{P_\beta}$ be two vertices in $H$, and $u = f^{-1}(\widetilde{u_\alpha}), v = f^{-1}(\widetilde{v_\beta})$. Then, it holds that:
	\begin{equation*}
		d_H(\widetilde{u_\alpha},\widetilde{v_\beta}) \leq d_G(u,v) + O(\epsilon)\cdot D.
	\end{equation*}
\end{lemma}
\begin{proof}
	Let $Q_{uv}$ be a shortest path from $u$ to $v$ in $G$. Recall that $K_{\Phi}$ is an emulator of $\Phi$ with hop diameter $O(\log \log n)$. 
	Let $P = (\alpha=\lambda_1,  \lambda_2 ,\ldots, \lambda_k = \beta)$ be a shortest path from $\alpha$ to $\beta$ in $K_{\Phi}$ such that $k = O(\log \log n)$. Since $K_{\Phi}$ preserves distances between nodes of $\Phi$, i.e, $w_{K_{\Phi}}(P) = w_{\Phi}(\Phi[\alpha,\beta])$, it must be that the nodes on $P$ constitute a subsequence of nodes on 	$\Phi[\alpha,\beta]$. By \Cref{lm:separation-planar}, there exists a sequence of vertices $(u = x_1,x_2, \ldots, x_k = v)$ such that $x_i \in \mathcal{Q}_{\lambda_i}\cap Q_{uv}$ for all $1\leq i \leq k$ and:
	\begin{equation}\label{eq:dist-uv-boundary}
		d_G(u,v) = \sum_{i=1}^{k-1} d_G(x_i,x_{i+1})
	\end{equation}
	For every $i \in [k]$, let $y_i \in P_{\lambda_i}$ be the $\delta$-portal closest to $x_i$; note that $y_1 = u$ and $y_k = v$ because $u$ and $v$ are $\delta$-portals. By the definition of $\delta$-portals,  $d_G(x_i,y_{i}) \leq \delta$. Thus, by the triangle inequality, it holds that:
	\begin{equation}\label{eq:ypair-vs-xpair}
		d_G(y_i,y_{i+1}) \leq d_G(x_i,x_{i+1}) + 2\delta
	\end{equation}
	Recall that $\widetilde{(y_i)_{\lambda_i}}$ be the copy of $y_i$ created for $\lambda_i$. Observe that,  since $(\lambda_i,\lambda_{i+1})$ is an edge in $K_{\Phi}$, by construction, there is an edge between $\widetilde{(y_i)_{\lambda_i}}$ and $\widetilde{(y_{i+1})_{\lambda_{i+1}}}$ of weight $w_H(\widetilde{(y_i)_{\lambda_i}}, \widetilde{(y_{i+1})_{\lambda_{i+1}}}) = d_G(y_i,y_{i+1})$. Thus, we have:
	\begin{align*}
		d_{H}(\widetilde{u_{\alpha}},\widetilde{v_{\beta}}) & =d_{H}(\widetilde{(y_{1})_{\lambda_{1}}},\widetilde{(y_{k})_{\lambda_{k}}})\leq\sum_{i=1}^{k-1}d_{H}(\widetilde{(y_{i})_{\lambda_{i}}},\widetilde{(y_{i+1})_{\lambda_{i+1}}})=\sum_{i=1}^{k-1}d_{G}(y_{i},y_{i+1})\\
		& \le\sum_{i=1}^{k-1}\left(d_{G}(x_{i},x_{i+1})+2\delta\right) \qquad \mbox{(by \Cref{eq:ypair-vs-xpair})}\\
		&=d_{G}(u,v)+2(k-1)\delta \qquad \mbox{(by \Cref{eq:dist-uv-boundary})}\\
		& =d_{G}(u,v)+O(\log\log n)\frac{\epsilon D}{\log\log n}=d_{G}(u,v)+O(\epsilon)D~,
	\end{align*}
	as claimed.
\end{proof}

\begin{lemma}\label{lm:distortion-planar-emb} For any $u,v \in V(G)$ and any $u' \in f(u), v' \in f(v)$, it holds that:
	\begin{equation*}
		d_H(u',v') \leq d_G(u,v) + O(\epsilon) D.
	\end{equation*}
\end{lemma}
\begin{proof}
	The proof is by case analysis. 
	If both $u,v$ are $\delta$-portals, then the lemma holds by \Cref{lm:boundary-distortion-planar}.
	Next, suppose that  both $u,v$ are not $\delta$-portals. 
	By \Cref{obs:P3}, there are two leaf nodes $\alpha,\beta\in \Phi$, such that $u\in X_\alpha$, and $v\in X_\beta$. In particular, there are unique copies of both $u$ and $v$, one for each vertex.
	If $\alpha=\beta$, then $H$ contains the edge $(u,v)$ and we are done.
	Otherwise, let $P_{uv}$ be a shortest $u$-$v$ path in $G$. Then it follows from \Cref{def:RSPD}  that there are two vertices $x_u$ and $x_v$ in $P_{uv}$ such that $x_u\in\mathcal{Q}_\alpha$ and $x_v\in\mathcal{Q}_\beta$ (it might be that $u=x_u$ or $v=x_v$). 
	Furthermore, there are $\delta$-portals $y_u\in P_\alpha$, and $y_v\in P_\beta$, at distances at most $\delta$ from $x_u$ and $x_v$ respectively.	
	By construction in Step 2 and Step 4 (depending on whether $u,v$ are leaves), $H$ contains the edges $(u,\widetilde{(y_u)_\alpha})$, $(v,\widetilde{(y_v)_\beta})$. 
	By \Cref{lm:boundary-distortion-planar}, we have:
	\begin{align*}
		d_{H}(u,v) & \le d_{H}(u,\widetilde{(y_{u})_{\alpha}})+d_{H}(\widetilde{(y_{u})_{\alpha}},\widetilde{(y_{v})_{\beta}})+d_{H}(\widetilde{(y_{v})_{\beta}},v)\\
		& \le d_{G}(u,y_{u})+d_{G}(y_{u},y_{v})+O(\epsilon)\cdot D+d_{G}(y_{v},v)\\
		& \le d_{G}(u,x_{u})+2\cdot d_{G}(x_{u},y_{u})+d_{G}(x_{u},x_{v})+2\cdot d_{G}(x_{v},y_{v})+d_{G}(x_{v},v)+O(\epsilon)\cdot D\\
		& \le d_{G}(u,v)+4\delta+O(\epsilon)\cdot D\le d_{G}(u,v)+O(\epsilon)\cdot D~.
	\end{align*}
Finally, the cases where either $u$ or $v$ is a $\delta$-portal are simpler and can be proved by the same argument.
\end{proof}

We now compute the runtime of our algorithm. First, we compute a (compact representation of) $O(1)$-\RSPD  $\Phi$ using  \Cref{lm:comp-recursive-decomp} in $O(n\log n)$ time. Next, we use  \Cref{thm:TreeSpanner} to compute an emulator $K_\Phi$ in $O(n\cdot\log\log n)$ time, with $O(|\Phi|\cdot\log\log |\Phi|)=O(n\cdot\log\log n)$ edges.
Next, we compute the set of $\delta$-portals. We will use the same shortest path tree $T_r$ rooted at $r$ from \Cref{lm:comp-recursive-decomp} (that can be found in $O(n)$ time~\cite{HKRS97}). For each vertex $v\in T_r$, we compute and store a set of $\delta$-portals $N(T_r[r,v],\delta)$ of the $r$-rooted shortest path $T_r[r,v]$ by depth-first tree traversal in $O(n)$ time. Note that we can afford to store all the portal of the path $T_r[r,v]$ at $v$ explicitly since $|N(T_r[r,v],\delta)| \leq \frac{\log\log n}{\eps}$. 
Given a vertex $v$ with a parent $x_v$ in $T_r$, and the closest $\delta$-portal ancestor $y_v$, if $d_{T_r}(v,y_v)=d_{T_r}(r,v)-d_{T_r}(r,y_v)\le\delta$ then we set $N(T_r[r,v],\delta)=N(T_r[r,x_v],\delta)$, and otherwise set $N(T_r[r,v],\delta)=N(T_r[r,x_v],\delta)\cup\{v\}$.

Next, we construct the one-to-many embedding $f$ into $H$ with tree-decomposition $\mathcal{T}$ following to Steps (1)-(4). The construction is straightforward and takes  \[
|\Phi|\cdot O(1)^{2}\cdot O\left(\frac{\log\log n}{\epsilon}\right)^{2}+|E(K_{\Phi})|\cdot O\left(\frac{\log\log n}{\epsilon}\right)^{2}+O(n)\cdot O\left(\frac{\log\log n}{\epsilon}\right)=O\left(n\cdot\frac{(\log\log n)^{3}}{\epsilon^{2}}\right)~,
\]
time. The last and most time-consuming step is to assign weights to the edges in $H$.
Recall that the weight of an edge in $(u',v')\in H$ where $u = f^{-1}(u')$ and $v = f^{-1}(v')$, is defined to be $w_H(u',v')=d_G(u,v)$. 
Computing the shortest distance between $u$ and $v$ in $G$ takes $\Omega(n)$ time, and using shortest path computation to find the weight of every edge in $H$ incurs $\Omega(n^2)$ time. 
Our solution is to use approximate distances instead of the exact ones to assign to edges of $H$.
Specifically, we will use the \emph{approximate distance oracle} of Thorup\footnote{A similar distance oracle was constructed independently by  Klein~\cite{Klein02}. However, Klein did not specify the construction time explicitly (a crucial property in our context).} \cite{Thorup04} to query the approximate distance between $u$ and $v$ in $O(\frac{1}{\epsilon})$ time. 

\begin{lemma}[Thorup~\cite{Thorup04}, Theorem 3.19]\label{lm:Thorup-oracle} Given an $n$-vertex planar graph $G(V,E,w)$ and a parameter $\epsilon < 1$, one can construct an $O(\frac{n\log n}{\epsilon})$-space data structure  $\mathcal{O}_{G,\eps}$ in $O(\frac{n\log^3(n)}{\epsilon^2})$ time such that given any two vertices $u,v$ in $G$, the data structure, in $O(\frac{1}{\eps})$ time, returns $\mathcal{O}_{G,\eps}(u,v)$ with 
	\begin{equation*}
		d_G(u,v) \leq \mathcal{O}_{G,\eps}(u,v) \leq (1+\eps)d_G(u,v)
	\end{equation*}
\end{lemma}

To assign weights in our graph $H$, we simply construct a $(1+\eps)$-distance oracle $\mathcal{O}_{G,\eps}(u,v)$ for $G$ using \Cref{lm:Thorup-oracle}, and then for every 
edge $(u',v')\in H$ where $u = f^{-1}(u')$ and $v = f^{-1}(v')$, we set $w_H(u',v')=\mathcal{O}_{G,\eps}(u,v)$.
Thus the total time for assigning the weights is $O(\frac{n\log^{3}(n)}{\epsilon^{2}})+|E(H)|\cdot O(\frac{1}{\epsilon})=O(\frac{n\log^{3}(n)}{\epsilon^{2}})$, which is also the overall running time of our algorithm.

Denote the graph constructed using the approximated distances by $\hat{H}$. It remains to argue that $\hat{H}$ has additive distortion $+O(\eps)\cdot D$. 
Note that for every $(u',v')\in H$, $w_H(u',v')\le w_{\hat{H}}(u',v')\le (1+\eps)\cdot w_H(u',v')$. This implies that for every $u',v'\in V(H)$, $d_{H}(u',v')\le d_{\hat{H}}(u',v')\le(1+\eps)\cdot d_{H}(u',v')$. As $H$ has additive distortion $+O(\eps)\cdot D$, it follows that, for every  $u',v'\in V(H)$ where $u = f^{-1}(u')$ and $v = f^{-1}(v')$,
\begin{align*}
	d_{G}(u,v)\le d_{H}(u',v') &{~\le~~ d_{\hat{H}}(u',v')~~}\le(1+\eps)\cdot d_{H}(u',v')\\
	& \phantom{~\le~~ d_{\hat{H}}(u',v')~~}\le(1+\eps)\cdot(d_{G}(u,v)+O(\eps)\cdot D)=d_{G}(u,v)+O(\eps)\cdot D~,
\end{align*}
where the last equality holds as $\eps\cdot d_{G}(u,v)\le \eps\cdot D$. \Cref{thm:PlanarToTreewidth}  now follows.

\subsection{Rooted Stochastic Embedding}\label{subsec:rootedStoEmb}

The main tool we use to design a PTAS for the bounded-capacity VRP is a \emph{rooted stochastic embedding}. 

\begin{definition}[$(r,\rho)$-rooted stochastic embedding]\label{def:rootedStochastic}
	Given a vertex $r \in V(G)$ and a parameter $\rho > 0$, we say that a dominating stochastic embedding $f: V(G)\rightarrow V(H)$ of a graph $G$ into a distribution over graphs $H$ is an \emph{$(r,\rho)$-rooted stochastic embedding} if
	\begin{equation*}\label{eq:root-distortion}
		\forall u,v\in V,\qquad\mathbb{E}[d_H(f(u),f(v))]   \leq d_G(u,v)  + \rho\cdot (d_G(r,u) + d_G(r,v))~.
	\end{equation*}
\end{definition}
We call $\rho$ the distortion parameter. Using \Cref{thm:PlanarToTreewidth}, we obtain the following lemma.

\begin{lemma}\label{lm:Root-Embedding-Planar}Let $G(V,E,w)$ be an $n$-vertex planar graph and $r$ be a distinguished vertex in $V$. For any given parameter $\eps \in (0,\frac14)$, we can construct in $O_{\eps}(n\log^3(n))$  time an $(r,\eps)$-rooted stochastic embedding of $G$ into graphs $H$ of treewidth $O_{\eps}(\log \log n)^2$.
\end{lemma}
\begin{proof}	
	We assume w.l.o.g. that the minimum distance in $G$ is $1$. We construct an embedding $f$ into a graph $H$ such that for every $u,v\in G$, $\mathbb{E}[d_H(f(u),f(v))]   \leq d_G(u,v)  + O(\eps)\cdot (d_G(r,u) + d_G(r,v))$; we can scale $\eps$ to get back distortion  $\eps\cdot (d_G(r,u) + d_G(r,v))$.  
	We begin by randomly ``slicing'' $G$ into a collection of \emph{bands} $\mathcal{B}= \{B_0,B_1,\ldots\}$ as follows. 
	Choose a random $x \in [0,1]$. 
	For every $i\ge 0$, let $U_{i}=(\frac{1}{\eps})^{\frac{i+x}{\eps}}$, and $L_{i}=(\frac{1}{\eps})^{\frac{i-1+x}{\eps}}$.
	We then define $B_i = \{u:  L_i \le d_G(r,u) < U_i\}$ (note that $L_0<1$, and hence every vertex $v\ne r$ belong some $B_i$). 
	Let $G_{0} = G[B_0]$, and for $i\ge 1$, let $G_{i}$ be the graph obtained from $G$ by removing every vertex of distance at least $U_i$ from $r$ and contracting every vertex of distance strictly less than $L_i$ from $r$ into $r$.
	For every vertex $v$ which is the neighbor of $r$ in $G_i$, the weight of the edge is defined to be $d_G(r,v)$, while all the other weights are the same as in $G$.
	Note that for every $v\in B_i$, $d_{G_i}(v,r)=d_G(v,r)$. In particular, $G_i$ has diameter at most $2\cdot U_i$.
	Observe that graphs $\{G_{i}\}_i$ can be constructed in total $O(n)$ time as we can construct them iteratively; each edge is contracted at most one time. All distances from $r$ can be computed in $O(n)$ time by applying a linear time single source shortest path algorithm~\cite{HKRS97}. %\atodo{Can BFS be computes more efficiently in planar graphs?}).
	
	For every $i\ge 0$, we use \Cref{thm:PlanarToTreewidth} with stretch parameter $\delta=(\frac{1}{\eps})^{-\frac{1}{\eps}}=\eps^{\frac1\eps}$ 
	to embed $G_i$ into a graph $H_i$ with treewidth $O(\delta^{-1}(\log \log n)^2)=O_\eps(\log \log n)^2$ and additive distortion $\delta\cdot U_i$. Finally, we construct a single graph $H$ by identifying the (image of the) vertex $r$ in all the graphs $H_i$; the embedding $f$ is defined accordingly.
	In addition, for every vertex $v$, we  add an edge from $f(r)$ to $f(v)$ of weight $d_G(r,v)$. Note that this can increase the overall treewidth by at most $1$. It follows that $H$ has treewidth $O_\eps(\log \log n)^2$.
	The total running time is bounded by $\sum_{i\ge0}\ensuremath{O(|B_{i}|\cdot\frac{\log^{3}n}{\delta^{2}})=O_{\eps}(n\cdot\log^{3}n)}$.
	
	Finally, we bound the expected distortion.
	We will actually show a stronger distortion guarantee: a ``Ramsey-type'' property (see e.g. \cite{BLMN03,MN07,ACEFN20,FL21,Fil21}). Specifically, we
	will show that with probability $1-2\eps$, a vertex $v$ has small distortion
	w.r.t. all other vertices. We say that the embedding is \emph{successful}
	for a vertex $v$ if there is some index $i$ such that $d(r,v)\in[\epsilon^{-1}\cdot L_{i},\epsilon\cdot U_{i}]$.
	To analyze the probability that $v$ is successful, let $a\in[0,1)$ be
	some number such that $d(r,v)=(\frac{1}{\eps})^{\frac{i_{v}+a}{\eps}}$
	for some integer $i_{v}\in\mathbb{N}$. Then $v$ is successful if $(\frac{1}{\eps})^{\frac{i_{v}+a}{\eps}}\in[(\frac{1}{\eps})^{\frac{i-1+x}{\eps}+1},(\frac{1}{\eps})^{\frac{i+x}{\eps}-1}]$
	or equivalently, $i_{v}+a\in[i-1+x+\eps,i+x-\eps]$ for some index $i$. Hence	 
	\[
	\Pr\left[v\text{ is unsuccessful}\right]=\Pr\left[\not\exists i\text{ s.t. }i+x\in[i_{v}+a+\epsilon,i_{v}+a+1-\epsilon]\right]=2\epsilon~.
	\]
	Next, suppose that a vertex $u$ is successful, and let $v$ be any
	other vertex. Let $i$ be the integer such that $u\in B_{i}$, and let $P_{u,v}$ be a shortest $u$-$v$ path. We consider three cases:
	\begin{itemize}
		\item If $P$ is fully contained in $B_{i}$, then $d_{G_{i}}(u,v)=d_{G}(u,v)$,
		implying that
		\begin{align}
			d_{H}(f(u),f(v)) & \le d_{H_{i}}(f(u),f(v))\le d_{G_{i}}(u,v)+\delta\cdot2U_{i}\nonumber\\
			& \le d_{G}(u,v)+L_{i}\le d_{G}(u,v)+2\epsilon\cdot d_{G}(r,u)~.\label{eq:rootedPContained}
		\end{align}
		\item Else, if $P$ contains a vertex $z\in B_{i'}$ for $i'>i$, then \begin{align*}
			d_{G}(r,u) & \le\epsilon\cdot U_{i}\le\epsilon\cdot d_{G}(r,z)\le\epsilon\cdot\left(d_{G}(r,u)+d_{G}(u,z)\right)\\
			& \le\epsilon\cdot\left(d_{G}(r,u)+d_{G}(u,v)\right)\le\epsilon\cdot\left(2\cdot d_{G}(r,u)+d_{G}(r,v)\right)~,
		\end{align*}
		which implies that $d_{G}(r,u)\le\frac{\epsilon}{1-2\epsilon}\cdot d_{G}(r,v)$. Thus, 
		\[
		d_{H}(v,u)\le d_{G}(v,r)+d_{G}(r,u)\le d_{G}(v,u)+2\cdot d_{G}(r,u)\le d_{G}(v,u)+\frac{2\epsilon}{1-2\epsilon}\cdot d_{G}(r,v)~.
		\]
		\item Else, $P$ is contained in $\cup_{i'\le i}B_{i'}$, and it contains
		at least one vertex $z\in B_{i'}$ for $i'<i$. Then $d_{G}(r,z)\le L_{i}\le\epsilon\cdot d_{G}(r,u)$.
		In particular
		\[
		d_{H}(v,u)\le d_{G}(v,r)+d_{G}(u,r)\le d_{G}(v,z)+d_{G}(z,u)+2\cdot d_{G}(r,z)\le d_{G}(v,u)+2\cdot\epsilon\cdot d_{G}(r,u)~.
		\]
	\end{itemize}
	We conclude that if $u$ is successful, then $d_{H}(v,u)\le d_{G}(v,u)+O(\epsilon)\cdot\left(d_{G}(r,v)+d_{G}(r,u)\right)$.
	Otherwise, $d_{H}(v,u)\le d_{G}(v,r)+d_{G}(r,u)$.
	As the probability of being unsuccessful is bounded by $2\epsilon$,
	we conclude that:
	\begin{align*}
		\mathbb{E}[d_{H}(f(u),f(v))] & \leq\Pr\left[v\text{ is successful}\right]\cdot\left(d_{G}(v,u)+O(\epsilon)\cdot\left(d_{G}(r,v)+d_{G}(r,u)\right)\right) \\
		                             & \qquad+\Pr\left[v\text{ is unsuccessful}\right]\cdot\left(d_{G}(v,r)+d_{G}(r,u)\right)                                    \\
		                             & \le d_{G}(v,u)+O(\epsilon)\cdot\left(d_{G}(r,v)+d_{G}(r,u)\right)~.
	\end{align*}

\end{proof}

\subsection{Applications (Proofs of \Cref{thm:MetricBakerPlanar,thm:VRP-planar})}\label{subsesc:apps}

\subsubsection{Metric Baker Problems}

In this section, we prove \Cref{thm:MetricBakerPlanar} that we restate below.

\MetricBakerPlanar*

Katsikarelis, Lampis, and Paschos~\cite{KLP19,KLP20} showed how to solve metric Baker problems in bounded treewidth graphs. Here we consider a slightly more general version where we are given a set of terminal $K\subseteq V(G)$ and we want to find a $\rho$-independent/dominating set for $K$. That is, for the $\rho$-independent set problem, the solution $I$ must be a subset of $K$, and for the $\rho$-dominating set problem, we only require that the vertices in $K$ be dominated (that is every vertex in $K$ will be a at distance at most $\rho$ from some vertex in the dominating set).

\begin{lemma}[Theorem 21~\cite{KLP20} and Theorem 31~\cite{KLP19}]\label{lm:KLS-IS-DS} Given an $n$-vertex graph $G(V,E,w)$ of treewidth $\tw$,  a terminal set $K\subseteq V(G)$, two parameters $\epsilon \in (0,1)$ and $\rho > 0$, and measure $\mu: V\rightarrow \mathbb{R}^+$, one can find in $n\cdot (\frac{\log n}{\eps})^{O(\tw)}$ time:
	\begin{OneLiners}
		\item A $(1-\epsilon)\rho$-independent  set $I\subseteq K$ such that for every $\rho$-independent  set $\tilde{I}\subseteq K$, $\mu(I) \geq \mu(\tilde{I})$.
		\item A$(1+\epsilon)\rho$-dominating set $S$ for $K$ such that for every $\rho$-dominating set $\tilde{S}$ for $K$, $\mu(S) \leq \mu(\tilde{S})$.
	\end{OneLiners}
\end{lemma}
%\sloppy 
We note that the running time of \Cref{lm:KLS-IS-DS} in the work of  Katsikarelis \etal \cite{KLP19, KLP20} is stated as $O^*((\frac{\tw}{\eps})^{O(\tw)})$ where $O^*$ notation hides the polynomial dependency on $n$; the main focus of their paper is to minimize the exponential dependency on the treewidth. The precise running time of their algorithms is  $O((\frac{\log n}{\eps})^{O(\tw)} n)$ (M.~Lampis, personal communication, 2021); see page 18 in~\cite{KLP20} and page 114 in~\cite{KLP19}. While the algorithm is for $\rho$-independent/dominating set with uniform measure, the same algorithm works for $\rho$-independent/dominating set with non-uniform measure by simply storing the measure of the  $\rho$-independent set corresponding to each configuration of the dynamic programming table. 

We now focus on proving \Cref{thm:MetricBakerPlanar}. By scaling every edge of $G(V,E,w)$, we assume that $\rho = 1$. For a cleaner presentation, we assume that $\frac{1}{\eps}$ is an integer.  Let $T_r$ be a shortest path tree rooted at a vertex $r$; $T_r$ can be found in $O(n)$ time~\cite{HKRS97}.

Similar to Fox-Epstein \etal \cite{FKS19},  we follow the classical layering technique of Baker~\cite{Baker94} to reduce to the case where the planar graph has diameter $O(1/\eps)$. We then embed the planar graph to a graph with treewidth $O(\frac{(\log\log n)^2}{\epsilon^2})$ and additive distortion $\eps$. The additive distortion $\eps$ translates into $(1\pm \epsilon)$ distance separation in the bicriteria PTAS. Next, we apply \Cref{lm:KLS-IS-DS} to find a bicriteria PTAS for metric Baker problems in treewidth-$O(\frac{(\log\log n)^2}{\epsilon^2})$   graphs. Finally, we lift the solution to the solution of the input planar graphs.

In the analysis below, we use Young's inequality~\cite{Young12}: 

\begin{theorem}[Young's Inequality]\label{thm:young} Given four real numbers $a\geq 0, b\geq 0, p> 1,q >1$ such that $\frac{1}{p} + \frac{1}{q} = 1$, the following inequality holds:
	\begin{equation*}
		ab\leq \frac{a^{p}}{p} + \frac{b^{q}}{q}.
	\end{equation*}
The equality holds if and only if $a^{p} = b^q$.
\end{theorem}

\paragraph{$\rho$-Independent Set Problem.~}  Let $\sigma$ be a number in $\{0,\ldots, 2/\eps-1\}$. Let  $\mathcal{I}_{-1,\sigma} = [0,\sigma)$, 
$\mathcal{I}^{-}_{-1,\sigma} = [0,\sigma-1)$, and for each $j\geq 0$, we define:
%\begin{equation*}
%			\mathcal{I}_{j,\sigma} = [\frac{j}{\eps} + \sigma,\frac{(j+1)}{\eps} + \sigma) \qquad \mathcal{I}^{-}_{j,\sigma} = [\frac{j}{\eps} + \sigma + 1, \frac{(j+1)}{\eps} + \sigma) 
%\end{equation*} 
\[
\mathcal{I}_{j,\sigma}=[\frac{2}{\eps}\cdot j+\sigma~,~\frac{2}{\eps}\cdot(j+1)+\sigma]
~~\qquad\mathcal{I}_{j,\sigma}^{-}=[\frac{2}{\eps}\cdot j+\sigma+1~,~\frac{2}{\eps}\cdot(j+1)+\sigma-1]
\]  
Let $U_{j,\sigma} = \{v\in V: d_G(r,v) \in \mathcal{I}_{j,\sigma}\}$ and  $U^-_{j,\sigma} = \{v\in V: d_G(r,v) \in \mathcal{I}^{-}_{j,\sigma}\}$ for every $j\geq -1$. Let $V_{\sigma} = \cup_{j=-1}^{\infty}U_{j,\sigma}$ and $V^{-}_{\sigma} = \cup_{j=-1}^{\infty}U^-_{j,\sigma}$. We observe that $V_{\sigma} = V$ for every $\sigma$ and that:

\begin{observation}\label{obs:Baker-layering-IS} Each vertex $v \in V$ is contained in at least $\frac{2}{\eps}-2$ sets in the collection $\{V^-_{\sigma}\}^{\frac{2}{\eps}-1}_{\sigma = 0}$.
\end{observation}

Fix an index $j \geq -1$. Let $G_{j,\sigma}$ be defined as follows. For $j = -1$ set $G_{-1,\sigma} = G[U_{-1,\sigma}]$. In general, $G_{j,\sigma}$ is the graph obtained from $G$ by removing every vertex of distance at least  $\frac{2(j+1)}{\eps} + \sigma$ from $r$ and contracting every vertex of distance strictly less than  $\frac{2j}{\epsilon} + \sigma$ to $r$; we set weight 1 to every edge incident to $r$. Observe that $G_{j,\sigma}$ is a planar graph with	a diameter at most $\frac{4}{\eps}+2$. Furthermore,  $V(G_{j,\sigma}) = U_{j,\sigma} \cup \{r\}$ when $j \geq 0$ and $V(G_{-1,\sigma}) = U_{-1,\sigma}$. Note that $G_{j,\sigma}$ can be constructed for all $j$ in total $O(n)$ time as we can construct them iteratively; each edge is contracted only once.

Next, we find a $(1-\eps)$-independent set $I_{j,\sigma}$ in $G_{j,\sigma}$.
Fix $\delta = \frac{\eps^2}{12}$, using \Cref{thm:PlanarToTreewidth}
we construct an embedding $f_{j,\sigma}$ of $G_{j,\sigma}$ into a graph $H_{j,\sigma}$ with treewidth $O(\frac{(\log\log |G_{j,\sigma}|)^2}{\delta}) = O(\frac{(\log\log n)^2}{\eps^2})$ and additive distortion $ \delta \cdot \dm(G_{j,\sigma}) \leq \frac{\epsilon^2}{12} (
\frac{4}{\eps} + 2) \leq \frac{\eps}{2}$ when $\eps \leq \frac{1}{2}$. The construction time is ${O(|G_{j,\sigma}|\cdot\frac{\log^{3}|G_{j,\sigma}|}{\delta^{2}})=O(|G_{j,\sigma}|\cdot\frac{\log^{3}n}{\epsilon^{4}})}$. For each vertex $u_{H_{j,\sigma}} \in H_{j,\sigma}$, we assign measure $\mu_{H_{j,\sigma}}(u_{H_{j,\sigma}}) = \mu(u)$ if $u_{H_{j,\sigma}} = f_{j,\sigma}(u)$ for some $u \in U^-_{j,\sigma}$, and $\mu_{H_{j,\sigma}}(u_{H_{j,\sigma}}) = 0$  otherwise. 
Let $K_{j,\sigma}=f(U^-_{j,\sigma})$ be a terminal set. 
Using \Cref{lm:KLS-IS-DS}, we find a $(1-\eps/2)$-independent set $I^{H}_{j,\sigma}$ for the terminal set $K_{j,\sigma}$ in $O\left((\frac{\log n}{\eps})^{O(\epsilon^{-2}(\log\log n)^{2})}|G_{j,\sigma}|\right)$ time.
Let $I_{j,\sigma} = f^{-1}(I^{H}_{j,\sigma})$ and $I_{\sigma}=\cup_{j\ge-1}I_{j,\sigma}$.

\begin{claim}\label{clm:IS-dist}$I_{\sigma}$ is a $(1-\eps)$-independent set of $G$. 
\end{claim} 
\begin{proof}
	Let $u$ and $v$ be any two vertices in $I_{\sigma}$ such that $d_G(u,v) < 1$. Our goal is to show that $d_G(u,v)\geq 1-\eps$. Suppose that there exists $j$ such that $u \in G_{j,\sigma}$ and $v\not\in G_{j,\sigma}$. By the construction of $I_{j,\sigma}$, $u \in U^-_{j,\sigma}$ and hence, the distance in $G$ from $u$ to any vertex not in $G_{j,\sigma}$ is at least $1$. In particular, $d_G(u,v)\geq 1$, a contradiction. It follows that there exists $j$ such that both $u$ and $v$ are in $G_{j,\sigma}$. By the construction of $I_{j,\sigma}$, both $u$ and $v$ are in $U^-_{j,\sigma}$. By the definition of $U^-_{j,\sigma}$, any vertex with distance at most $1$ from $u$ or $v$ is contained in $U_{j,\sigma}$. Since $d_G(u,v) < 1$,  any vertex in the shortest path from $u$ to $v$ in $G$ is contained in $U_{j,\sigma}$. It follows that $d_{G}(u,v) = d_{G_{j,\sigma}}(u,v)$.  By \Cref{thm:PlanarToTreewidth}, and the fact that $I^{H}_{j,\sigma}$ is a $1-\frac\eps2$ independent set, $d_{G_{j,\sigma}}(u,v)+\frac{\eps}{2}\geq d_{H_{j,\sigma}}(f_{j,\sigma}(u),f_{j,\sigma}(v))\geq1-\epsilon/2$. It follows that $d_{G_{j,\sigma}}(u,v) \geq 1-\eps$, which implies $d_{G}(u,v) \geq 1-\eps$ as claimed.
\end{proof}

Next, we return the set $I=I_{\sigma}$ for the $\sigma$ which maximizes $\mu(I_{\sigma})$ over all  $\sigma\in[0,\frac2\eps-1]$. By \Cref{clm:IS-dist}, $I$ is a $(1-\eps)$-independent set of $G$. We now bound the cost of $I_{\sigma}$.
Let $\tilde{I}$ be any $1$-independent set in $G$, and $\sigma^* = \argmax_{\sigma}\mu(\tilde{I}\cap V^-_{\sigma})$. By \Cref{obs:Baker-layering-IS}, we have:
\begin{equation*}%\label{eq:approx-IS}
	\mu(\tilde{I}\cap V^-_{\sigma^*}) \geq \frac{\sum_{\sigma =0}^{\frac{2}{\eps}-1} \mu(\tilde{I}\cap V^-_{\sigma})}{2/\eps} = \frac{2/\eps-2}{2/\eps} \mu(\tilde{I}) = (1-\eps)\mu(\tilde{I})~.
\end{equation*}
For every $\sigma$ and $j$,  $f(\tilde{I}\cap U^-_{j,\sigma})$ is a $1$-independent set in $H_{j,\sigma}$ w.r.t. the terminal set $K_{j,\sigma}$, and $\mu_{H_{j,\sigma}}(f(\tilde{I}\cap U^-_{j,\sigma})) = \mu(\tilde{I}\cap U^-_{j,\sigma})$.  We conclude that 
\begin{align*}
	\mu(I) & =\max_{\sigma}\mu(I_{\sigma})\ge\mu(I_{\sigma^{*}})=\sum_{j=-1}^{\infty}\mu(I_{j,\sigma^{*}})=\sum_{j=-1}^{\infty}\mu(f^{-1}(I_{j,\sigma^{*}}^{H}))=\sum_{j=-1}^{\infty}\mu_{H_{j,\sigma^{*}}}(I_{j,\sigma^{*}}^{H})\\
	& \ge\sum_{j=-1}^{\infty}\mu_{H_{j,\sigma^{*}}}(f_{j,\sigma^{*}}(\tilde{I}\cap V_{\sigma^{*}}^{-}))=\mu(\tilde{I}\cap V_{\sigma}^{-})\ge(1-\epsilon)\cdot\mu(\tilde{I})~.
	\end{align*}
The total running time to construct all the embeddings $f_{\sigma,j}$ into graphs $H_{\sigma,j}$, and to run the approximation algorithm from \Cref{lm:KLS-IS-DS}, denoted by  $T_{IS}(n)$,  is 
\begin{align}\label{eq:time-IS-semifinal}
	T_{IS}(n) = & \sum_{\sigma}\sum_{j}O(|G_{j,\sigma}|\cdot\frac{\log^{3}|G_{j,\sigma}|}{\delta^{2}})+O\left((\frac{\log n}{\eps})^{O(\epsilon^{-2}(\log\log n)^{2})}|G_{j,\sigma}|\right)\\
	& \qquad=n\cdot\sum_{\sigma}\left(O(\frac{\log^{3}n}{\eps^{4}})+2^{O(\eps^{-2}\cdot\log\frac{1}{\eps}\cdot(\log\log n)^{3})}\right)\\
	& \qquad=  n\cdot\left(O(\frac{\log^{3}n}{\eps^{5}}) +  (1/\eps)2^{O(\eps^{-2}\cdot\log\frac{1}{\eps}\cdot(\log\log n)^{3})} \right)\\
\end{align}

For a given fixed $\kappa > 0$, by applying Young's inequality (\Cref{thm:young}) with $a = \eps^{-2}\cdot\log\frac{1}{\eps}$, $b = (\log\log n)^3$, $p = 1+\kappa/2$ and $q = 1 + \frac{2}{\kappa}$, we have that:

\begin{align}\label{eq:IS-young-app}
	 \eps^{-2}\cdot\log\frac{1}{\eps}\cdot(\log\log n)^{3} \leq\eps^{-(2+\kappa)}(\log\frac{1}{\eps})^{1+\kappa/2}\frac{2}{2+\kappa} + (\log\log n)^{3 + 6/\kappa}\frac{\kappa}{2+\kappa}   \\
\end{align}

Since $\kappa$ is fixed, combining \Cref{eq:time-IS-semifinal} and \Cref{eq:IS-young-app}, we have that:

\begin{align*}
	T_{IS}(n) = & n\cdot\left(O(\frac{\log^{3}n}{\eps^{5}}) + 2^{\tilde{O}(\eps^{-(2+\kappa)})} \cdot 2^{O((\log\log n)^{3 +6/\kappa})}\right) =  2^{\tilde{O}(\eps^{-(2+\kappa)})}n^{1+o(1)}~,
\end{align*}
as desired.

\paragraph{$\rho$-Dominating Set Problem.~} The algorithm for the $\rho$-dominating set problem is similar to the algorithm for $\rho$-independent set problem. 
For each $\sigma\in \{0,\ldots, \frac{2}{\eps}-1\}$, we define $\mathcal{I}_{-1,\sigma} = [0,\sigma]$ and $\mathcal{I}^+_{-1,\sigma} = [0,\sigma+1]$, and for each $j\geq 0$:
\[
\mathcal{I}_{j,\sigma}=[\frac{2}{\eps}\cdot j+\sigma~,~\frac{2}{\eps}\cdot(j+1)+\sigma]~~\qquad\mathcal{I}_{j,\sigma}^{+}=[\frac{2}{\eps}\cdot j+\sigma-1~,~\frac{2}{\eps}\cdot(j+1)+\sigma+1]
\]  
Let $U_{j,\sigma} = \{v\in V: d_G(r,v) \in \mathcal{I}_{j,\sigma}\}$ and for $j\geq -1$, $U^+_{j,\sigma} = \{v\in V: d_G(r,v) \in \mathcal{I}^{+}_{j,\sigma}\}$. 
For every $\sigma$, set $G_{-1,\sigma} = G[U^+_{-1,\sigma}]$, and for $j\ge 0$,  $G_{j,\sigma}$ is the graph obtained from $G$ by removing every vertex of distance at least  $\frac{2}{\eps}\cdot(j+1)+\sigma+1$ from $r$, and contracting every vertex of distance strictly less than  $\frac{2}{\eps}\cdot j+\sigma-1$ to $r$; we set every edge incident to $r$ the weight $1$. Observe that $G_{j,\sigma}$ is a planar graph with diameter at most $\frac{4}{\eps}+6$. Furthermore,  $V(G_{j,\sigma}) = U^+_{j,\sigma} \cup \{r\}$ when $j \geq 0$ and  $V(G_{-1,\sigma}) = U^+_{-1,\sigma}$.
As in the $\rho$-independent set problem, $G_{j,\sigma}$ can be constructed for all $j$ in total $O(n)$ time.

Next, for every $j$ and $\sigma$, we find a $(1+\eps)$-dominating set $S_{j,\sigma}$ in $G_{j,\sigma}$.
Fix $\delta = \frac{\eps^2}{16}$. We construct an embedding $f_{j,\sigma}$ of $G_{j,\sigma}$ into a graph $H_{j,\sigma}$ with treewidth $O(\frac{(\log\log |G_{j,\sigma}|)^2}{\delta}) = O(\frac{(\log\log n)^2}{\eps^2})$ and additive distortion $\delta\cdot\dm(G_{j,\sigma})\leq\frac{\epsilon^{2}}{16}(\frac{4}{\eps}+6)\leq\frac{\eps}{2}$
when $\eps \leq \frac{1}{2}$; by \Cref{thm:PlanarToTreewidth}, $f_{j,\sigma}$ can be constructed in $O(|G_{j,\sigma}|\cdot\frac{\log^{3}|G_{j,\sigma}|}{\delta^{2}})=O(|G_{j,\sigma}|\cdot\frac{\log^{3}n}{\epsilon^{4}})$
time. For each vertex $u_{H_{j,\sigma}} \in H_{j,\sigma}$, we assign measure $\mu_{H_{j,\sigma}}(u_{H_{j,\sigma}}) = \mu(u)$ if $u_{H_{j,\sigma}} = f_{j,\sigma}(u)$ for some $u \in U^+_{j,\sigma}$, and $\mu_{H_{j,\sigma}}(u_{H_{j,\sigma}}) = \infty$  otherwise. 
Let $K_{j,\sigma}=f(U_{j,\sigma})$ be a terminal set. 
Using \Cref{lm:KLS-IS-DS}, we find a  set $S^H_{j,\sigma}$ which $(1+\eps/2)$-dominates the terminal set $K_{j,\sigma}$ in $O(\frac{\log n}{\eps})^{O(\epsilon^{-2}(\log\log n)^{2})}|G_{j,\sigma}|)$ time, and such that for every $(1+\frac\eps2)$-dominating set $\hat{S}$, $\mu_{H_{j,\sigma}}(S^H_{j,\sigma})\le \mu_{H_{j,\sigma}}(\hat{S})$.
%\atodo{Single criteria!}
% 
Let $S_{j,\sigma} = f^{-1}(S^{H}_{j,\sigma})$.
As $f(U_{j,\sigma})$ is a dominating set of finite measure, and all the vertices out of $f(U^+_{j,\sigma})$ have measure $\infty$, it holds that $S^{H}_{j,\sigma}\subseteq f(U^+_{j,\sigma})$.
Observe that for every vertex $u\in U_{j,\sigma}$, $f(u)\in K_{j,\sigma}$, and hence there is a vertex $f(v)\in S^{H}_{j,\sigma}$ such that $d_{H_{j,\sigma}}(f(u),f(v))\le 1+\frac\eps2$. It follows that $v\in S_{j,\sigma}$, and $d_G(u,v)\le d_{H_{j,\sigma}}(f(u),f(v))\le  1+\eps$.
Set $S_{\sigma}=\cup_{j\ge-1}S_{j,\sigma}$, as $V = \cup_{j=-1}^{\infty}U_{j,\sigma}$, $S_{\sigma}$ is a $1+\eps$ dominating set. We return the set $S=S_{\sigma}$ for the $\sigma$ which minimizes $\mu(S_{\sigma})$ for $\sigma\in[0,\frac2\eps-1]$.

Let $\tilde{S}$ be any $1$-dominating in $G$, and $\sigma^*$ the index that minimizes $\sum_{j=-1}^{\infty}\mu(\tilde{S}\cap U_{j,\sigma^{*}}^{+})$.
Observe that each vertex $v \in V$ is contained in at most $\frac{2}{\eps}+2$ sets in the collection $\{U^+_{j,\sigma}\}_{\sigma\in\{0,\dots,\frac2\eps-1\},j\ge-1}$. Thus
\[
\sum_{j=-1}^{\infty}\mu(\tilde{S}\cap U_{j,\sigma^{*}}^{+})\le\frac{\epsilon}{2}\cdot\sum_{\sigma=0}^{\frac{2}{\epsilon}-1}\sum_{j=-1}^{\infty}\mu(\tilde{S}\cap U_{j,\sigma}^{+})=\frac{\epsilon}{2}\cdot\sum_{v\in\tilde{S}}\mu(v)\sum_{\sigma=0}^{\frac{2}{\epsilon}-1}\sum_{j=-1}^{\infty}\cdot\boldsymbol{1}_{v\in U_{j,\sigma}^{+}}=\frac{\epsilon}{2}\cdot\sum_{v\in\tilde{S}}\mu(v)\cdot(\frac{2}{\epsilon}+2)=(1+\epsilon)\cdot\mu(\tilde{S})~.
\]
By \Cref{lm:KLS-IS-DS}, $\mu(S^H_{j,\sigma})\le\mu(\tilde{S}\cap U^+_{j,\sigma})$. We conclude, 
\[
\mu(S)=\min_{\sigma}\mu(S_{\sigma})\le\mu(S_{\sigma^{*}})=\sum_{j=-1}^{\infty}\mu(S_{j,\sigma^{*}})\le\sum_{j=-1}^{\infty}\mu(\tilde{S}\cap U_{j,\sigma^{*}}^{+})\le(1+\epsilon)\cdot\mu(\tilde{S})~.
\]
The time analysis is exactly the same as the time analysis of the $\rho$-independent set problem.

\subsubsection{Bounded-Capacity Vehicle Routing Problem} \label{subsec:VHR-planar}

Using \Cref{lm:Root-Embedding-Planar}, we are able to reduce the problem to graphs with treewidth $O_{\eps}(\log\log n)^2$.  Cohen-Addad \etal \cite{CFKL20} designed an algorithm to solve the bounded-capacity VRP in $O(Q\cdot\eps^{-1}\log n)^{O(Q\cdot\frac{\tw}{\epsilon})}\cdot n^{O(1)}$ %^2 
time for graphs with treewidth $\tw$ (Theorem 8 in~\cite{CFKL20}, recall that $Q$ is the capacity of the vehicle).  Most of the time is spent on computing distances between $O(n\cdot\tw^2)$ pairs of vertices in $G$. 
An exact distance oracle is a data structure that, after some prepossessing stage, can answer distance queries (exactly) between every pair of vertices.
By using a distance oracle instead of computing all distances directly, we have:
\begin{lemma}[Theorem 8~\cite{CFKL20}, implicit]\label{lm:DP-VRP-TW} Let $(G(V,E,w), K, r, Q)$ be an instance of the VRP problem,  where $G(V,E,w)$ has treewidth $\tw$ and $n$ vertices. Suppose that we can construct an exact distance oracle for $G$ with preprocessing time $P(n,\tw)$ and query time $T(n,\tw)$, then we can find $(1+\eps)$-approximate solution in time $(Q\epsilon^{-1}\log n)^{O(Q\cdot\tw)}\cdot n + O(n\cdot\tw^2)\cdot T(n,\tw) + P(n,\tw)$.  	
\end{lemma}

 Chaudhuri and Zaroliagis~\cite{CZ00} constructed a distance oracle with low preprocessing time and query time for directed graphs with small treewidth. Here, we only need an oracle for undirected graphs.

\begin{lemma}[Theorem 4.2 (ii) in~\cite{CZ00}]\label{lm:distOracle-tw} Let $G(V,E,w)$ be a directed graph with treewidth $\tw$ and $n$ vertices. We can construct a distance oracle with  prepprocessing time $P(n,\tw) = O(n\cdot\tw^3)$ and query time $T(n,\tw) = O(\tw^3\cdot \alpha(n))$ where $\alpha(n)$ is the inverse Ackermann function.
\end{lemma}
By plugging \Cref{lm:distOracle-tw} into \Cref{lm:DP-VRP-TW}, we get:

\begin{lemma}\label{lm:DP-VRP} 	Given  an instance of the VRP problem $(G(V,E,w), K, r, Q)$ where $G(V,E,w)$ has treewidth $\tw$ and $n$ vertices, one can find a $(1+\eps)$-approximate solution in time $(Q\epsilon^{-1}\log n)^{O(Q\cdot\tw)}\cdot n$.  
\end{lemma}

The final missing piece we need for our algorithm for the VRP is a result of Becker \etal \cite{BKS19}, who showed how to solve the VRP from  an $(r,\eps)$-rooted stochastic embedding and an efficient dynamic programming for VRP in small treewidth graphs.

\begin{lemma}[Becker, Klein, and Schild~\cite{BKS19}, implicit]\label{lm:BKS} Let $\eps\in (0,1)$ be a parameter. Let $(G(V,E,w), K, r, Q)$ be an instance of the VRP problem with $n$ vertices and $m$ edges such that $G(V,E,w)$ admits an $(r,\epsilon)$-rooted stochastic embedding into graphs with treewidth $\tau(n,\eps)$ that can be constructed in time $T(n,m,\eps)$. If we can find a $(1+\eps)$-approximate solution of the VRP in graphs with treewidth $\tw$ in time $D(n,Q,\tw,\eps)$, then we can find a $(1+\eps)$-approximate solution for the VRP in $G(V,E,w)$ in expected time $D(n,Q,\tau(n,\eps/(9Q)),\epsilon/3) + T(n,m,\eps/(9Q)) +  O(m)$.  	
\end{lemma}

We now show that there exists a PTAS for the bounded-capacity VRP in almost linear time. We restate \Cref{thm:VRP-planar} below for convenience.

\CVRPPlanar*
\begin{proof} 
	\sloppy By \Cref{lm:Root-Embedding-Planar}, we can construct  an $(r,\epsilon)$-rooted stochastic embedding into graphs with treewidth $\tau(n,\eps) =O_\eps(\log\log n)^2$ in $O_\eps(n\log^3n)$ time. By \Cref{lm:DP-VRP}, $D(n,Q,\tw,\eps) =(Q\epsilon^{-1}\log n)^{O(Q\cdot\tw)}n$. By \Cref{lm:BKS}, we can find a $(1+\eps)$-approximate solution of the VRP in time:
	\begin{align}
		(Q\epsilon^{-1}\log n)^{O_{\frac{\epsilon}{9Q}}(Q\cdot\log\log n)^{2}}\cdot n & =2^{O_{\epsilon}(\log\log n)^{3}}\cdot n~~\qquad\mbox{(since }Q=O(1))\nonumber\\
		& \leq2^{f(\eps)\cdot(\log\log n)^{3}}\cdot n\qquad\mbox{(for some function }f(\eps))\nonumber\\
		& \leq2^{\frac{f(\eps)^{2}+(\log\log n)^{6}}{2}}n=O_{\eps}(1)\cdot2^{\frac{(\log\log n)^{6}}{2}}\cdot n=O_{\eps}(n^{1+o(1)})~.\label{PlanarVRPruntime}
	\end{align}
\end{proof}

\section{Additional Notation and Preliminaries} \label{sec:addNotation}
\subsection{Metric Embeddings} 

We will study a more permitting generalization of metric embedding introduced by Cohen-Addad \etal \cite{CFKL20}, which is called \emph{one-to-many} embedding.\hspace{-10pt}

\begin{definition}[One-to-many embedding]\label{def:one-to-many}
	A \emph{one-to-many embedding} is a function $f:X\rightarrow2^Y$  from the points of a metric space $(X,d_X)$ into non-empty subsets of points of a metric space $(Y,d_Y)$, where the subsets $\{f(x)\}_{x\in X}$ are disjoint.	
	$f^{-1}(x')$ denotes the unique point $x\in X$ such that $x'\in f(x)$. If no such point exists, $f^{-1}(x')=\emptyset$.
	A point $x'\in f(x)$ is called a \emph{copy} of $x$, while $f(x)$ is called the \emph{clan} of $x$.
	For a subset $A\subseteq X$ of vertices, denote $f(A)=\cup_{x\in A}f(x)$.
	
	We say that $f$ is \emph{dominating} if for every pair of points $x,y\in X$, it holds that $d_X(x,y)\le \min_{x'\in f(x),y'\in f(y)}d_Y(x',y')$.
	We say that $f$ has an additive distortion $+\eps D$ if $f$ is dominating and $\forall x,y\in X$, $\max_{x'\in f(x),y'\in f(y)}d_Y(x',y')\le d_X(x,y)+\eps D$.
	
	A stochastic one-to-many embedding is a distribution $\mathcal{D}$ over dominating one-to-many embeddings. We say that a stochastic one-to-many embedding $f$ has an expected additive distortion $+\eps D$, if $\forall x,y\in X$, $\mathbb{E}[\max_{x'\in f(x),y'\in f(y)}d_Y(x',y')]\le d_X(x,y)+\eps D$.
\end{definition}

While the embedding in \Cref{thm:MinorToTreewidth} is one-to-one, it is technically more convenient to construct a one-to-many embedding. This is because the embedding is constructed in multiple steps, where the ``topological complexity'' grows from one step to the next. To transition from one step to the next, we sometimes require the embeddings be clique-preserving (see \Cref{def:cliquePreserving}), a property for which one-to-many embeddings are required.
Note that we  can always make the final embedding $f$ of $G$ one-to-one by retaining exactly one copy in each set $f(v)$ for each $v \in V(G)$.  

Next, we define clan embedding, which will be used in the approximation of metric Baker's problems. This notion was introduced by \blind \cite{FL21}.

\begin{definition}[Clan Embedding]\label{def:clan}
	A clan embedding from metric space $(X,d_X)$ into a metric space $(Y,d_Y)$ is a pair $(f,\chi)$ where $f:X\rightarrow2^{Y}$ is a dominating one-to-many embedding, and $\chi:X\rightarrow Y$ is a classic embedding. For every $x\in X$, we have that $\chi(x)\in f(x)$; here $f(x)$ called the clan of $x$, while $\chi(x)$ is referred to as the chief of the clan of $x$ (or simply the chief of $x$).
	
	\sloppy 
	We say that a clan embedding $f$ has an  additive distortion $+\Delta$ if for every $x,y\in X$, \mbox{$\min_{y'\in f(y)}d_{Y}(y',\chi(x))\le d_{X}(x,y)+\Delta$}. 
	
	A \emph{$(\Delta,\delta)$-clan embedding} is a distribution over clan embeddings with additive distortion $+\Delta$ such that the expected clan size has $\mathbb{E}[|f(x)|] \leq (1+\delta)$ for every $x\in X$. 
\end{definition}

We will construct embeddings for minor-free graphs using a divide-and-concur approach. In order to combine different embeddings into a single one, it will be important that these embeddings are \emph{clique-preserving}. 
\begin{definition}[Clique-copy]
	Consider a one-to-many embedding $f:G\rightarrow 2^H$, and a clique $Q$ in $G$. A subset $Q'\subseteq f(Q)$ is called clique copy of $Q$ if $Q'$ is a clique in $H$, and for every vertex $v\in Q$, $Q'\cap f(v)$ is a singleton.
\end{definition}
\begin{definition}[Clique-preserving embedding]\label{def:cliquePreserving}
	A one-to-many embedding $f:G\rightarrow2^H$ is called clique-preserving embedding if for every clique $Q$ in $G$, $f(Q)$ contains a clique copy of $Q$.
	A clan embedding $(f,\chi)$ is clique-preserving if $f$ is clique-preserving.
\end{definition}

\begin{definition}[Ramsey Type Embedding]\label{def:ramsey}
For a given parameter $\delta \in (0,1)$, a \emph{$(\Delta,\delta)$-Ramsey type embedding} from metric space $(X,d_X)$ into a metric space $(Y,d_Y)$ with additive distortion $+\Delta$ is a distribution over dominating one-to-one embeddings $f:X\rightarrow Y$ such that there is a subset $M\subseteq X$ of vertices for which the following claims hold:
	\begin{enumerate}
		\item For every $u\in X$, $\Pr[u\in M]\ge 1-\delta$.
		\item For every $u\in M$ and $v\in X$, $d_Y(f(u),f(v))\le d_X(u,v)+\Delta$.
	\end{enumerate}
\end{definition}
	
We note that the distortion in \Cref{def:ramsey} is worst-case. 	
	
\subsection{Graph Minor Theory and Known Embeddings}	
We introduce the notation used in the graph minor theory Robertson and Seymour.  A \emph{vortex} is a  graph $W$ equipped with a path decomposition $\{X_1,X_2,\ldots, X_t\}$ (see the discussion after \Cref{def:tree-decomp}) and a sequence of $t$ designated vertices $x_1,\ldots, x_t$, called the \emph{perimeter} of $W$, such that each $x_i \in X_i$ for all $1\leq i \leq t$. The \emph{width} of the vortex is the width of its path decomposition.
	We say that a vortex $W$ is \emph{glued} to a face $F$ of a surface embedded graph $G$ if $W\cap F$ is the perimeter of $W$ whose vertices appear consecutively along the boundary of $F$. 
	
	If $G = G_{\Sigma}\cup W$ consists of planar graph  $G_{\Sigma}$ with a single vortex $W$ of width $h$, we call $G$ a \emph{$h$-vortex planar graph}.
	
	\paragraph{$h$-Multi-vortex-genus graph.~}   a $h$-multi-vortex-genus graph  	is a graph $G = G_{\Sigma}\cup  W_1\cup\dots \cup W_{h}$, where $G_{\Sigma}$ is (cellularly) embedded on a surface $\Sigma$ of genus $h$, and each $W_i$ is a vortex of width at most $h$ glued to a face of $G_{\Sigma}$.

	In constructing a low-treewidth embedding of $K_r$-minor-free graphs,  Cohen-Addad \etal\cite{CFKL20} provided a \emph{deterministic} reduction from planar graphs with one vortex to surface-embedded graphs with many (but constant) vortices. The following lemma is a reinterpretation of Lemma 7 and Lemma 8 in the full version of their paper.
	
	\begin{lemma}[Multiple Vortices and Genus, Lemmas 7 and 8 \cite{CFKL20}, adapted]\label{lm:CKFL-embed-genus-vortex} Suppose that every $n$-vertex $h$-vortex planar graph with diameter $D$ can be deterministically embedded into a graph with treewidth $t(h,\eps,n)$ and additive distortion $+\eps D$ via a one-to-many clique-preserving embedding, then given an $n$-vertex  $h$-multi-vortex-genus graph $G$ of diameter $D$, 
		$G$ can be deterministically embedded into a graph with treewidth $O_h(t(O_h(1),O_h(\eps),O_h(n)))$ and additive distortion $+\eps D$ via a clique-preserving embedding. 
	\end{lemma}

	\paragraph{Nearly $h$-embeddability.~}  A graph $G$ is nearly $h$-embeddable if  there is a set of at most $h$ vertices $A$, called \emph{apices}, such that $G\setminus A$ is an $h$-multi-vortex-genus graph $G_{\Sigma}\cup \{W_1, W_2,\ldots, W_{h}\}$.
	
	Cohen-Addad \etal\cite{CFKL20}  showed how to handle the apices by using randomness, effectively extending their \Cref{lm:CKFL-embed-genus-vortex} to obtain a  reduction from planar graphs with one vortex to nearly $h$-embeddable graphs. The resulting embedding is stochastic.  The stochasticity is necessary due to a lower bound for deterministic embedding of Cohen-Addad \etal \cite{CFKL20} who showed that any (deterministic) embedding into graphs with treewidth $o(\sqrt{n})$ must incur an additive distortion at least $\frac{D}{20}$ (Theorem 3 in \cite{CFKL20}). The following lemma is a reinterpretation of their Lemma 9.
	
	 \begin{lemma}[Lemma 9 \cite{CFKL20}, adapted]\label{lm:nearPlanar-to-nearEmb} If every $h$-vortex planar graphs of $n$ vertices and diameter $D$ can be (deterministically or stochastically) embedded into a graph with treewidth $t(h,\eps,n)$ and distortion $+\eps D$ via a one-to-many clique-preserving embedding, then nearly $h$-embeddable graphs with $n$ vertices and diameter $D$  can be \emph{stochastically} embedded into a distribution over graphs with treewidth $O_h\left(t\left(O_h(1),O_h(\eps^2),n\right)\right) + O_h(\eps^{-2})$ and expected additive distortion $+\eps D$ via a clique-preserving embedding. 
	 \end{lemma}

	\paragraph{$h$-Clique-sum.~} A graph $G$ is an $h$-clique-sum of two graphs $G_1$ and $G_2$, denoted by $G  = G_1\oplus_h G_2$, if there are two cliques of size exactly $h$ each such that $G$ can be obtained by  identifying vertices of the two cliques and removing some clique edges of the resulting identification. The shared clique between $G_1$ and $G_2$ is called the \emph{adhesion} of the two graphs. The adhesion is also called the joint set of the two graphs in~\cite{CFKL20}. 
	
	Note that clique-sum is not a well-defined operation since the clique-sum of two graphs is not unique due to the clique edge deletion step.
	
	The celebrated theorem of Robertson and Seymour (\Cref{thm:RS}, \cite{RS03}) intuitively said that every $K_r$-minor-free graph can be decomposed into nearly embeddable graphs glued together following a tree-like structure by clique-sum operations. 
	
	\begin{theorem}[Theorem~1.3~\cite{RS03}] \label{thm:RS} There is a constant $h = O_r(1)$ such that any $K_r$-minor-free graph $G$ can be decomposed into a tree $\mathbb{T}$ where each node of $\mathbb{T}$ corresponds to a nearly $h$-embeddable graph such that $G = \cup_{X_iX_j \in E(\mathbb{T})} X_i \oplus_h X_j$.
	\end{theorem} 
	
	The graphs corresponding to the nodes in the clique-sum decomposition above are referred to as \emph{pieces}. A piece may not be a subgraph of $G$ since some edge in a clique involving the clique-sum operation of the piece with another piece may not be present in $G$. However, we can modify the graph to guarantee that every piece is a subgraph of $G$ by adding an edge $(u,v)$ in the piece to $G$ if the edge does not appear in $G$. The weight of $(u,v)$  is set to be the shortest distance between $u$ and $v$ in $G$. This operation does not change the Robertson-Seymour decomposition of $G$ nor its metric shortest path.

\section{Stochastic Embedding of Minor-free Graphs and Applications} \label{sec:MinorToTW}

\subsection{Stochastic Embedding of Minor-free Graphs (Proof of \Cref{thm:MinorToTreewidth})}\label{subsec:MinorToTW}

In this section, we construct a stochastic additive embedding for minor-free graphs as claimed in \Cref{thm:MinorToTreewidth}. We will construct a one-to-many embedding for $G$. Afterwards, the embedding can be converted into a one-to-one embedding by retaining a single copy of each vertex. We refer readers to \Cref{subsec:technique} for a high-level overview of our construction.  We start with an embedding for $h$-vortex planar graphs. In a nutshell, we improve the treewidth in the construction of Cohen-Addad \etal \cite{CFKL20} in two places: (1) the embedding of $h$-vortex planar graphs and (2) gluing the embeddings of nearly embeddable graphs, a.k.a. pieces, via clique-sum operations. Our results are described in the following lemmas. 

\begin{lemma}\label{lm:nearlyPlanar-emb}	Given an $n$-vertex $h$-vortex planar graph $G_{\Sigma}\cup W$ of diameter $D$ and a parameter $\epsilon < 1$, there is a polynomial time algorithm that constructs a one-to-many (deterministic) clique-preserving embedding $f: V(G) \rightarrow H$ into a graph $H$ of treewidth at most $O(\frac{h(\log \log n)^2}{\eps})$ and additive distortion $+\epsilon D$.
\end{lemma}

The proof of \Cref{lm:nearlyPlanar-emb} is presented in \Cref{subsec:nearlyPlanar}.

\begin{restatable}{lemma}{nearlyEmbToMinor}\label{lm:nearlyEmb-to-Minor} If  nearly $h$-embeddable graphs of $n$ vertices and diameter $D$ can be stochastically embedded into a graph with treewidth $t(h,\eps,n)$  by a clique-preserving embedding $f$ of expected additive distortion $+\eps D$, then $K_r$-minor-free graphs with $n$ vertices and diameter $D$  can be stochastically embedded into a graph with treewidth $t(h(r),\eps,n) + h(r)$ and expected additive distortion $+\eps D$ for some function $h$ depending on $r$ only.
\end{restatable}

The proof of \Cref{lm:nearlyEmb-to-Minor} is presented in \Cref{subsec:general-minor}. We are now ready to prove \Cref{thm:MinorToTreewidth}, which we restate below for convenience.

\MinorToTreewidth*
\begin{proof} By plugging \Cref{lm:nearlyPlanar-emb} to \Cref{lm:nearPlanar-to-nearEmb} (here $\hat{t}(h,\eps,n) = O(\frac{h(\log \log n)^2}{\eps})$), we obtain a stochastic, clique-preserving  embedding of nearly $h$-embeddable graphs with additive distortion $+\eps D$ and treewidth:
	
	\begin{equation}\label{eq:nearlyem-tw}
		O_h\left(\hat{t}\left(O_h(1),O_h(\eps^2),n\right)\right) + O_h(\eps^{-2}) = O_h(\frac{(\log \log n)^2}{\eps^2})~.
	\end{equation}
By applying \Cref{lm:nearlyEmb-to-Minor} with  $t(h,\eps,n) = O_h(\frac{(\log \log n)^2}{\eps^2})$ (due to \Cref{eq:nearlyem-tw}), we get a stochastic embedding of $K_r$-minor-free graphs with additive distortion $+\eps D$ and treewidth:
	\begin{equation*}
		t(h(r),\eps,n) + h(r) =  O_r(\frac{(\log \log n)^2}{\eps^2}) + O_r(1) = O_r(\frac{(\log \log n)^2}{\eps^2})~,
	\end{equation*}
as desired.
\end{proof}

\subsubsection{$h$-vortex planar graphs}\label{subsec:nearlyPlanar}

In this section, we focus on constructing an additive embedding for $h$-vortex planar graphs as described in \Cref{lm:nearlyPlanar-emb}. We follow the same idea in the construction of an additive embedding for planar graphs in \Cref{sec:PlanarToTW}. Here we rely on (an implicit) construction of an $O(h)$-\RSPD for $h$-vortex planar graphs by Cohen-Addad \etal \cite{CFKL20}  (see Section 5.1. in the \href{https://arxiv.org/pdf/2009.05039.pdf}{full version}).

\begin{lemma}[\cite{CFKL20}]\label{lem:RSPD}Given an edge-weighted  $h$-vortex planar graph $G = G_{\Sigma}\cup W$  with $n$ vertices, an $O(h)$-\RSPD $\Phi$ of $G$ can be constructed in polynomial time.
\end{lemma}

\begin{proof}[Proof of \Cref{lm:nearlyPlanar-emb}]
	The deterministic embedding in \Cref{thm:PlanarToTreewidth} was constructed using \RSPD and did not exploit any planar properties (other than the distance oracle \Cref{lm:Thorup-oracle} for obtaining an efficient implementation). Hence, we can use the exact same construction for $h$-vortex planar graph based on their $O(h)$-\RSPD  guaranteed in \Cref{lem:RSPD}. The only difference is, since each node $\alpha\in\Phi$ can contain up to $O(h)$ shortest paths in its boundary $\mathcal{Q}_\alpha$, the width of the bags in the resulting tree decomposition might increase by up to an $O(h)$ factor. Thus other than the clique-preserving property, all the other properties guaranteed by the lemma hold.

In order to obtain the clique preserving property, we slightly modify the embedding by adding additional edges. We rely on the following structure of the \RSPD, which is implicit in Cohen-Addad \etal \cite{CFKL20}. We include the proof here for completeness.

\begin{lemma}[\cite{CFKL20}]\label{lm:clique-leaf} Let $Q$ be any clique in an $h$-vortex planar graph $G_{\Sigma}\cup W$, and $\Phi$ be an \RSPD of $G_{\Sigma}\cup W$. Then, there is a leaf $\alpha\in \Phi$ such that $Q\subseteq X_{\alpha}$.
\end{lemma}
\begin{proof} We prove by induction that at any level of $\Phi$, there exists a node $\alpha \in \Phi$ such that $Q\subseteq X_{\alpha}$. The base case is when $\alpha$ is the root of $\Phi$. The claim trivially holds since $X_{\alpha} = G_{\Sigma}\cup W$. 
	
Let $\alpha$ is a node at level $\ell$ for some $\ell \geq 0$ such that $Q\subseteq X_{\alpha}$. Let $\beta_1,\beta_2$ be two children of $\alpha$. Let $Y =  V(X_{\beta_1})\cap  V(X_{\beta_2})$. Observe that there is no edge in $G$ between $V(X_{\beta_1})\setminus Y$ and $ V(X_{\beta_2}) \setminus Y$ by property (2c) of  \RSPD (\Cref{def:RSPD}). Thus,  either $Q\cap (V(X_{\beta_1})\setminus Y) = \emptyset$ or $Q\cap (V(X_{\beta_2})\setminus Y) = \emptyset$ since it is a clique. In the former case, $Q\subseteq X_{\beta_2}$, and in the latter case, $Q\subseteq X_{\beta_1}$. 
\end{proof}
Consider the $O(h)$-\RSPD of $G_{\Sigma}\cup W$ given by \Cref{lem:RSPD}. 
$\Phi$ has  $O(n)$ nodes. 
Recall that we portalized all the shortest paths, and then applied the (4-step) construction in \Cref{subsec:embedd_AFversion} to the $O(h)$-\RSPD $\Phi$ to construct an embedding $f$, the host graph $H$, and a tree decomposition $\mathcal{T}$ of $H$.
We remark that there are only $O_h(n)$ maximal cliques in $G_{\Sigma}\cup W$ \cite{ELS10}. 
\footnote{The degeneracy of a graph is the minimal $d$ such that there is an order $v_1v_2\dots v_n$ over the vertices where for every $i$, $v_i$ has at most $d$ neighbors among $v_{i+1}\dots v_n$. 
Eppstein, L{\"o}ffler 
and Strash \cite{ELS10} showed that every graph with degeneracy $d$ has at most $(n-d)3^{\frac d3}$ maximal cliques. As every $n$-vertex $K_r$-minor-free graph has average degree at most $O(r\cdot\sqrt{\log r})$ \cite{Kostochka82,Thomason84}, it has degeneracy at most  $O(r\cdot\sqrt{\log r})$, and hence $O_r(n)$ maximal cliques.
} 

Let $Q$ be a maximal clique in $G_{\Sigma}\cup W$, according to \Cref{lm:clique-leaf}, there is a leaf $\alpha \in \Phi$ such that $Q\subseteq X_{\alpha}$. Recall that $P_{\alpha}$ is the set of portals of paths associated with $\alpha$, and $I_{\alpha}$ is the set of internal vertices of $\alpha$.  Let $\widetilde{I_{\alpha}}$ and $\widetilde{P_{\alpha}}$ be the images of $I_{\alpha}$ and $P_{\alpha}$ in $H$, respectively.  By \Cref{obs:Ialpha-Palpha}, $\widetilde{I_{\alpha}}\cup \widetilde{ P_{\alpha}}$ induces a clique in $H$.   Let $B_{\alpha}$ be the bag of $\mathcal{T}$ that contains $\widetilde{I_{\alpha}}\cup \widetilde{ P_{\alpha}}$. We add a new copy $\widetilde{Q}$ of $Q$ to $H$, add the edge set $\widetilde{Q}\times(\widetilde{I_{\alpha}}\cup \widetilde{ P_{\alpha}})$ to $E(H)$, and update $f$ accordingly. 
We then make a new bag $B_{Q} = \tilde{Q}\cup \widetilde{I_{\alpha}}\cup \widetilde{ P_{\alpha}}$ and connect $B_Q$ to $B_{\alpha}$ by adding an edge $(B_Q,B_\alpha)$ to $\mathcal{E}(\mathcal{T})$. Note that $|Q| = O(h)$ since $G_{\Sigma} \cup W$ is an $h$-vortex planar graph. Thus, the asymptotic bound on the treewidth of $\mathcal{T}$ remains unchanged.
As usual, for every newly added edge $(u',v')$ in $H$, set its weight to be $d_G(u,v)$ where $u = f^{-1}(u')$ and $v = f^{-1}(v')$.
Denote by $\widetilde{H}$ the version of $H$ after this last modification step.
The clique preservation property holds by the construction. It is left to show that $f$ has an additive distortion $O(\eps)D$; we can get additive distortion $+\eps D$ by scaling $\eps$. 

\begin{claim}\label{clm:nearlyPlanar-close} For any $\tilde{u} \in \widetilde{Q}$, there is a vertex $\tilde{x}\in \widetilde{P_{\alpha}}\cup \widetilde{I_{\alpha}}$ such that $d_H(\tilde{u},\tilde{x}) \leq \eps D$. 
\end{claim}
\begin{proof}
	Let $u = f^{-1}(\tilde{u})$. If $u \in I_{\alpha}$, then there is another copy $\tilde{u}'$ of $u$ in $\widetilde{I_\alpha}$. By construction, there is an edge $(\tilde{u},\tilde{u}')$ of weight $0$. Thus, $\tilde{x}$ is $\tilde{u}'$ in this case. Otherwise, $u$ must belong to some boundary path of $X_{\alpha}$. Thus, there exists a portal $x \in P_{\alpha}$ such that $d_G(x,u)\le\delta< \eps D$.  By construction, the copy, say $\tilde{x}$,  of $x$ in $\widetilde{P_{\alpha}}$ has an edge to $\tilde{u}$ in $H$. The claim then follows from the fact that $w_H(\tilde{x}, \tilde{u})= d_G(x,u) \leq \eps D$.  
\end{proof}
 
We are now ready to bound the distortion of $f$. Let $\tilde{u}$ and $\tilde{v}$ be any two vertices in $\widetilde{H}$. Let $\tilde{x}$ and $\tilde{y}$ be closest vertices to $\tilde{u}$ and $\tilde{v}$, respectively, that belong to $H$ (it is possible that $\tilde{x}=\tilde{u}$ or $\tilde{y}=\tilde{v}$). 
By \Cref{clm:nearlyPlanar-close}, $d_{\widetilde{H}}(\tilde{x}, \tilde{u}) \leq \eps D$ and $d_{\widetilde{H}}(\tilde{y}, \tilde{v}) \leq \eps D$. Let $x = f^{-1}(\tilde{x})$ and $y = f^{-1}(\tilde{y})$. Observe that $d_H(x,y) \leq d_G(x,y) + c\eps D$ for some constant $c$. By the triangle inequality, it holds that:
\begin{equation*}
	d_G(x,y) \leq d_G(x,u) + d_G(u,v)  + d_G(v,y) \leq d_G(u,v) + 2\eps D~.
\end{equation*} 
Also by the triangle inequality, we have:
\begin{equation*}
	\begin{split}
			d_{\widetilde{H}}(\tilde{u},\tilde{v}) &\leq d_{\widetilde{H}}(\tilde{u},\tilde{x}) + d_H(\tilde{x},\tilde{y})  + d_{\widetilde{H}}(\tilde{y},\tilde{v}) ~\leq~ d_H(\tilde{x},\tilde{y})  + 2\eps D~\\
			&\leq d_G(x,y) + (c+2)\eps D ~\leq~ d_G(u,v) + (c+4)\eps D,
	\end{split}
\end{equation*} 
as desired. 
\end{proof}

\subsubsection{From nearly embeddable graphs to general $K_r$-minor-free graphs}\label{subsec:general-minor}

In this section, we prove a reduction from nearly embeddable graphs to general $K_r$-minor-free graphs as described in \Cref{lm:nearlyEmb-to-Minor}. Our reduction removes a loss of an additive $O(\log n)$ factor in the treewidth in the reduction of Cohen-Addad \etal~\cite{CFKL20}. To give context to our improvement, we briefly review their construction below. 

By the Robertson-Seymour theorem (\Cref{thm:RS}), a $K_r$-minor-free graphs $G$ can be decomposed into a tree $\mathbb{T}$ of nearly $h(r)$-embeddable graphs, called \emph{pieces}, that are glued together by taking clique-sums along the edges of $\mathbb{T}$. We then embed each piece into a graph with treewidth $t(h(r),\eps,n)$ by a clique-preserving embedding. The clique-preserving property allows us to take clique-sum of the host graphs of the pieces to obtain an embedding of $G$. The remaining problem is that the distortion of the embedding could be arbitrarily large due to long paths in $\mathbb{T}$.  Cohen-Addad \etal \cite{CFKL20} resolved this issue  by ``shortcutting'' long paths via centroids\footnote{A centroid of a tree $T$ is a vertex $v$ such that every subtree of $T\setminus v$ has at most $|V(T)|/2$ vertices}. Specifically, they apply the following steps: (a) find a centroid piece $G_c$ of $\mathbb{T}$, (b) recursively construct the embedding for each subgraph $K$ obtained by removing the centroid piece from $\mathbb{T}$, (c) glue the embedding of $G_c$ and $K$ via clique-sum, and (d) add (the image of) the adhesion $J$ between $K$ and $G_c$ to every bag in the tree decomposition of the host graph of $K$. Step (d) helps ``shortcut'' the shortest path, say $Q_{uv}$, between any two vertices $u$ and $v$ between two different subgraphs $K_i$ and $K_j$ of $\mathbb{T}\setminus G_c$  in the following way: there is a pair of vertices $x\in Q_{uv}\cap J_i$ and $y\in Q_{uv}\cap J_j$ such that the distances between $u$ and $x$, and between $v$ and $y$ are preserved \emph{exactly}. Thus, the final embedding only incurs an additive distortion $+\eps D$ for the distance between $x$ and $y$ (in the embedding of $G_c$). However, step (d) increases the treewidth of the embedding additively by $O(h(r)\log n)$, since each adhesion has size $O(h(r))$ and the recursion has depth $O(\log n)$---by using the centroids for shortcutting---where each recursive step adds one adhesion to the embedding a piece in $\mathbb{T}$.

Our idea to remove the $\log n$ factor is as follows. We first root the tree $\mathbb{T}$. Then for each piece $G_i$, we simply add the adhesion $J_0$ between $G_i$ and its parent bag $G_0$ to every bag in the tree decomposition of the host graph $H_i$ in the embedding of $G_i$; this increases the treewidth of $G_0$ by only $+O(h(r))$, and the overall treewidth of the final embedding is also increased only by $+O(h(r))$. For the stretch, we show by induction that for every vertex $u \in G_a$ and $v\in G_b$ for any two pieces $G_a,G_b$ such that $G_i$ is their lowest common ancestor, there is a vertex $x \in Q_{uv}\cap G_i$ and  $y \in Q_{uv}\cap G_i$, the distances between $u$ and $x$, and between $v$ and $y$ are preserved \emph{exactly}. Thus, the distortion for $d_G(u,v)$ is just $+\eps D$.  Since we apply this construction to \emph{every} piece of $\mathbb{T}$, $+\eps D$ is also the distortion of the final embedding. 

We are now ready to prove \Cref{lm:nearlyEmb-to-Minor}, which we restate below.
\nearlyEmbToMinor*
  \begin{proof}
  	We apply \Cref{thm:RS} to obtain a decomposition of $G$ into a tree $\mathbb{T}$ of pieces $\{G_{1}, \ldots, G_{\ell}\}$. Since every adhesion involved in a clique-sum operation is a clique, $G_i$ has a diameter at most $D$ for every $i \in [\ell]$. Root $\mathbb{T}$ at a piece $G_{rt}$ for some $rt \in [\ell]$. For each $G_i$, let $f_i$ be the clique-preserving embedding $G_i$ into $H_i$ with treewidth $t(h(r),\eps,|V(G_i)|) \leq t(h(r),\eps,n)$ and expected additive distortion $+\eps D$. Let $\mathcal{T}_i$ be the tree decomposition of $H_i$ with width at most  $t(h(r),\eps,n)$. 
  	
  	For each pair of pieces $G_i$ and $G_j$ where $G_j$ is the parent piece of $G_i$ in $\mathbb{T}$. Let  $J_{ij}$ be the adhesion involving the clique sum of $G_i$ and $G_j$. Let $Q^i_{ij}$ and $Q^j_{ij}$ be the clique copies of $J_{ij}$ in $H_i$ and $H_j$. We apply the following two steps:
  		\begin{itemize}
		\item \textbf{(Step 1).~} We add to $H_i$ the edge set $Q^i_{ij}\times V(H_i)$ and set the weight of each edge to be the distance between the pre-images of its endpoints in $G$; recall that $f_i(V(G_i)) = V(H_i)$. We then add $Q^i_{ij}$ to every bag of the tree decomposition $\mathcal{T}_i$ so that it remains a valid tree decomposition of $H_i$ after adding edges. 
		
		\item \textbf{(Step 2).~} We identify $Q^i_{ij}$ and $Q^j_{ij}$ into a single clique copy $Q'_{ij}$. This effectively connects $H_i$ and $H_j$. We then replace the vertices of $Q^i_{ij}$ and $Q^j_{ij}$ in $\mathcal{T}_i$ and $\mathcal{T}_j$, respectively, by the vertices of $Q'_{ij}$. We then take a pair of arbitrary two bags $B_i$ and $B_j$ of $\mathcal{T}_i$ and $\mathcal{T}_j$, respectively, that contain $Q'_{ij}$, and add an edge $(B_i,X_j)$ to connect $\mathcal{T}_i$ and $\mathcal{T}_j$ into a single tree decomposition. $X_i$ and $X_j$ exist since $f_i$ and $f_j$ are clique-preserving.
  		\end{itemize}  
 	 Let $H$ and $\mathcal{T}$ be the host graph and the tree decomposition obtained by applying the above two steps to every pair of pieces $G_i$ and $G_j$ where $G_j$ is the parent piece of $G_i$ in $\mathbb{T}$. For each vertex $\tilde{v} \in V(H)$, we add $\tilde{v}$ to $f(v)$, which is initially empty, if $\tilde{v}$ is the image of $v$ in the embedding of some piece in the decomposition of $G$. This completes our construction of the embedding of $G$.
 	 
 	 Observe that $\mathcal{T}$ is a valid tree decomposition of $H$ since every time we add an edge $(\tilde{u},\tilde{v})$ to $E(H)$ that is not present in any $\{H_i\}_{i=1}^{\ell}$ (in Step 1), we guarantee that $\tilde{u}$ and $\tilde{v}$ are in the same bag.  Additionally, the construction in Step 1 only increases the treewidth of $H_i$ by at most $h(r)$, which is the upper bound on the size of $J_{ij}$ (this happens only once per piece). The width of the resulting tree in Step 2 is bounded by the maximum width of $\mathcal{T}_i$ and $\mathcal{T}_j$. Thus, the width of $\mathcal{T}$ is $t(h(r),\eps,n) + h(r)$.
 	 
 	 Observe that $f$ is a dominating embedding since $\{f_i\}_i$ are dominating. Thus, it is left to show that the expected distortion of the embedding is $+\eps D$. To that end, we show the following claim.
 	 
 	 \begin{claim}\label{clm:minor-dist-ind} Let $G_a$ and $G_b$ be two pieces of $\mathbb{T}$ such that $G_a$ is a proper ancestor of $G_b$. Let $J^{\downarrow}_{ab}$ be the adhesion in the clique-sum between $G_a$ and its child on the path from $G_a$ to $G_b$. For any two vertices $u\in J^{\downarrow}_{ab}$ and $v\in G_b$, there exists a copy $\tilde{u}\in f_a(u)$ such that for every $\tilde{v} \in f_b(v)$:
 	 		\begin{equation*}
 	 			d_H(\tilde{u}, \tilde{v}) = d_G(u,v)~.
 	 		\end{equation*}
 	 \end{claim}
   \begin{proof}
   		We prove by induction. The base case is when $G_a$ is a parent of $G_b$. Then $J^{\downarrow}_{ab} = J_{ab}$. Let $Q^a_{ab}$ and $Q^{b}_{ab}$ be two clique copies of $J^{ab}$ in $H_a$ and $H_b$, respectively, that are identified into a single clique copy $Q'_{ab}$ in the construction of Step 2. Let $\tilde{u}$ be the copy of $u$ in $Q'_{ab}$. By the construction in Step 1, there is an edge from $\tilde{u}$ to every vertex in $f(V(G_b))$. Thus for every $\tilde{v} \in f_b(v)$, $d_H(\tilde{u}, \tilde{v}) = d_G(u,v)$.
   		
   		For the inductive case, let $G_c$ be the parent piece of $G_b$.  Let $S_{uv}$ be a shortest path between $u$ and $v$ in $G$. Since $G_c\in \mathbb{T}[G_a,G_b]$, there must be a vertex $x \in J_{cb}\cap S_{uv}$. By the induction hypothesis, there exists a copy $\tilde{x} \in f_c(x)$ such that for every $\tilde{v}\in f_b(v)$, $d_H(\tilde{x}, \tilde{v}) = d_G(x,v)$. Also by induction, there exists a copy $\tilde{u} \in f_a(u)$ such that $d_H(\tilde{u},\tilde{x}) = d_G(u,x)$. By the triangle inequality, for every $\tilde{v}\in f_b(v)$:
   		\begin{equation*}
   			d_H(\tilde{u},\tilde{v}) \leq d_H(\tilde{u},\tilde{x}) + d_H(\tilde{x}, \tilde{v}) = d_G(u,x) + d_G(x,v) = d_G(u,v)~.
   		\end{equation*}
	As $f$ is dominating, we conclude that  $d_H(\tilde{u},\tilde{v}) = d_G(u,v)$.
   \end{proof}

  	We now show that for every $u,v$,
  	\begin{equation}\label{eq:minor-dist}  		
  		\mathbb{E}\left[\max_{\tilde{u}\in f(u),\tilde{v}\in f(v)}d_{H}(\tilde{u},\tilde{v})\right]\le d_{G}(u,v)+\eps D~.
  	\end{equation}
  	Let $G_a$ and $G_b$ be the two pieces closest to the root of $\mathbb{T}$, containing $u$ and $v$ respectively.
  	  If $G_a = G_b$, then $d_{G_a}(u,v) = d_G(u,v)$, and by the assumption that $f_a$ has expected additive distortion $+\eps D$, we have that $\mathbb{E}[\max_{\tilde{u}\in f_a(u),\tilde{v}\in f_a(v)}d_H(\tilde{u},\tilde{v})]\le d_H(u,v)+\eps D$.
  	  By \Cref{clm:minor-dist-ind}, for any copy $\tilde{u}\in f(u)$ (resp. $\tilde{v}\in f(v)$) coming from a decedent piece of $G_a$ (resp. $G_b$), there is a copy $\hat{u}\in f_a(u)$ (resp. $\hat{v}\in f_a(v)$) such that $d_H(\hat{u},\tilde{u})=d_G(u,u)=0$ (resp. $d_H(\hat{v},\tilde{v})=0$). It follows that 
  	  \begin{align}
  	  	\mathbb{E}\left[\max_{\tilde{u}\in f(u),\tilde{v}\in f(v)}d_{H}(\tilde{u},\tilde{v})\right] \nonumber& \le\mathbb{E}\left[\max_{\tilde{u}\in f(u),\tilde{v}\in f(v)}d_{H}(\tilde{u},\hat{u})+d_{H}(\hat{u},\hat{v})+d_{H}(\hat{v},\tilde{v})\right]\\
  	  	& \le\mathbb{E}\left[\max_{\hat{u}\in f_{a}(u),\hat{v}\in f_{a}(v)}d_{H}(\hat{u},\hat{v})\right]~~\le~~ d_{G}(u,v)+\eps D~.\label{eq:lowerCopies}
  	  \end{align}
    
  	  Otherwise ($G_a\ne G_b$), let $G_c$ be the least common ancestor of $G_a$ and $G_b$ in $\mathbb{T}$. Let $S_{uv}$ be the shortest path between $u$ and $v$ in $G$. Then, there must exist two vertices $x\in J^{\downarrow}_{ca}\cap S_{uv}$ and $y \in J^{\downarrow}_{cb}\cap S_{uv}$. By \Cref{clm:minor-dist-ind}, there exist $\hat{x} \in f_c(x)$ and  $\hat{y}\in f_c(y)$ such that for every two copies $\tilde{u}\in f_a(u)$ and $\tilde{v}\in f_b(v)$,
  	  \begin{equation*}
  	  	d_H(\hat{x}, \tilde{u}) = d_G(x,u) \qquad \& \qquad d_H(\hat{y},\tilde{v}) = d_G(y,v).
  	  \end{equation*}
      By the triangle inequality, linearity of expectation, the fact that $f_c$ has expected distortion $+\eps D$, and the fact that $x$ and $y$ lie on a shortest $u$-$v$ path, it holds that:
  	  \begin{align*}
  	  	\mathbb{E}\left[\max_{\tilde{u}\in f_{a}(u),\tilde{v}\in f_{b}(v)}d_{H}(\tilde{u},\tilde{v})\right] & \le\mathbb{E}\left[\max_{\tilde{u}\in f_{a}(u),\tilde{v}\in f_{b}(v)}d_{H}(\tilde{u},\hat{x})+d_{H}(\hat{x},\hat{y})+d_{H}(\hat{y},\tilde{v})\right]\\
  	  	& \le d_{G}(u,x)+\mathbb{E}\left[\max_{\tilde{x}\in f_{c}(x),\tilde{y}\in f_{c}(y)}d_{H}(\tilde{x},\tilde{y})\right]+d_{G}(y,v)\\
  	  	& =d_{G}(u,x)+d_{G}(x,y)+d_{G}(y,v)+\eps D~~=~~d_{G}(u,v)+\eps D~.
  	  \end{align*}
	For any copy $\tilde{u}\in f(u)$ (resp. $\tilde{v}\in f(v)$) coming from a decedent piece of $G_a$ (resp. $G_b$), by \Cref{eq:lowerCopies}, we conclude that \Cref{eq:minor-dist}  holds, as desired.
  \end{proof}

\subsection{An Efficient PTAS for the Bounded-capacity VRP (Proof of \Cref{thm:VRP-minor})}

In this section, we show that the embedding in \Cref{thm:MinorToTreewidth} can be used to obtain an efficient PTAS for the bounded-capacity VRP in $K_r$-minor-free graphs as claimed in \Cref{thm:VRP-minor}. Our result significantly improves the previous result by Cohen-Addad \etal~\cite{CFKL20} who showed an approximation scheme with \emph{quasi-polynomial time} for this problem. 

The algorithm for minor-free graphs closely resembles the algorithm for planar graphs in \Cref{subsec:VHR-planar}. The following lemma is similar to \Cref{lm:Root-Embedding-Planar}.

\begin{lemma}\label{lm:Root-Embedding-Minor}
	Let $G(V,E,w)$ be an $n$-vertex  $K_r$-minor-free graph and $r$ be a distinguished vertex in $V$. For any given parameter $\eps \in (0,\frac14)$, we can construct in $O_{\eps,r}(1)\cdot n^{O(1)}$   time an $(r,\eps)$-rooted stochastic embedding of $G$ into graphs $H$ of treewidth $O_{r,\eps}((\log \log n)^2)$.
\end{lemma}
\begin{proof}
Our construction follows exactly the same steps as the construction of \Cref{lm:Root-Embedding-Planar}; the only difference is that we use \Cref{thm:MinorToTreewidth} instead of \Cref{thm:PlanarToTreewidth} (and the standard shortest path algorithm instead of \cite{HKRS97}). The running time is clearly polynomial.

For the distortion analysis, consider a pair of vertices $u,v\in V$, and let $i$ be the index such that $u\in B_i$, and let $f_i$ be the embedding of $G_i$ to the low treewidth graph $H_i$.
Recall that the vertex $u$ is successful with probability $1-2\eps$. 
If $u$ is successful and the shortest path $P$ between $u$ and $v$ is not fully contained in $B_i$, then by the case analysis in \Cref{lm:Root-Embedding-Planar}, $d_{H}(v,u)\le d_{G}(v,u)+2\cdot\epsilon\cdot d_{G}(r,u)$  (as that argument used only edges incident to $r$, which are present in our embedding as well).
Otherwise, $P\subseteq B_i$, and hence $d_{G_i}(u,v)=d_{G}(u,v)$. Following the calculation in \Cref{eq:rootedPContained}, we have
\[
\mathbb{E}\left[d_{H}(f(u),f(v))\mid d_{G_{i}}(u,v)=d_{G}(u,v)\right]\le d_{G}(u,v)+\delta\cdot2U_{i}\le d_{G}(u,v)+2\epsilon\cdot d_{G}(r,u)~.
\]
Thus, conditioning on $u$ being successful, we have
\[
\mathbb{E}[d_{H}(f(u),f(v))\mid u\text{ is successful}]\le d_{G}(v,u)+O(\epsilon)\cdot\left(d_{G}(r,v)+d_{G}(r,u)\right)~,
\]
and it follows that:
\begin{align*}
	\mathbb{E}[d_{H}(f(u),f(v))] & \leq\Pr\left[v\text{ is successful}\right]\cdot\mathbb{E}[d_{H}(f(u),f(v))\mid u\text{ is successful}]\\
	& \qquad+\Pr\left[v\text{ is unsuccessful}\right]\cdot\left(d_{G}(v,r)+d_{G}(r,u)\right)\\
	& \le d_{G}(v,u)+O(\epsilon)\cdot\left(d_{G}(r,v)+d_{G}(r,u)\right)~.
\end{align*}
\end{proof}

We are now ready to prove \Cref{thm:VRP-minor}, which we restate below for convenience. 
\CVRPMinor*
\begin{proof}
	Following exactly the same steps in the proof of \Cref{thm:VRP-planar} in \Cref{subsec:VHR-planar}, we plug  \Cref{lm:Root-Embedding-Minor} and \Cref{lm:DP-VRP} into \Cref{lm:BKS} to obtain a $(1+\eps)$-approximate solution.
	Following \Cref{PlanarVRPruntime}, the running time hence is:
	\[
	(Q\epsilon^{-1}\log n)^{O_{\frac{\epsilon}{9Q}}(Q\cdot\log\log n)^{2}}\cdot n+O_{\epsilon,r}(1)\cdot n^{O(1)}=O_{\epsilon,r}(1)\cdot n^{O(1)}~.
	\]
\end{proof}

\section{Ramsey Type and Clan Embeddings of Minor-free Graphs and their Applications} \label{sec:RamseyClan}

\subsection{Ramsey Type and Clan Embeddings (Proofs of \Cref{thm:MinorRamsey,thm:MinorClan})}

In this section, we prove \Cref{thm:MinorClan} and \Cref{thm:MinorRamsey}. A key idea in the construction of the authors~\cite{FL21} is that one can construct a clan/Ramsey type embedding from a deterministic one-to-one embedding of $h$-multi-vortex-genus graphs with an $O(\log n)$ loss in the treewidth. \Cref{lm:NearlyClanRamsey} below is a reinterpretation of Lemma 15 and Lemma 16 for Ramsey type embedding and Lemma 18 and Lemma 19 for clan embedding in~\cite{FL21}.

\begin{lemma}[Lemmas 15,16, 18 and 19 \cite{FL21}, adapted]\label{lm:NearlyClanRamsey} Suppose that every $n$-vertex $h$-multi-vortex-genus graph of diameter $D$ can be deterministically embedded into a graph with treewidth $t(h,\eps,n)$ and additive distortion $+\eps D$ via a clique-preserving embedding.
	Then for every parameter $\delta\in(0,1)$, and every $n$-vertex $K_r$-free graph $G=(V,E,w)$ of diameter $D$, there exist an $(\eps D, \delta)$-Ramsey type embedding and an  $(\eps D, \delta)$-clan embedding, both with treewidth $t(O_r(1),O(\frac{\delta \eps}{\log(n)}),n)+O_r(\log n)$.
\end{lemma}

To obtain their treewidth bound $O_r(\frac{\log^2}{\delta \eps})$, \blind \cite{FL21} used the embedding of  $h$-multi-vortex-genus graphs  with treewidth $O_r(\frac{\log(n)}{\eps})$ by Cohen-Addad \etal~\cite{CFKL20}.

We show below that \Cref{lm:NearlyClanRamsey} and \Cref{lm:CKFL-embed-genus-vortex} provide a reduction from \emph{$h$-vortex planar graphs} to $K_r$-minor-free graphs.

\begin{lemma}\label{lm:Red-NearlyPlanar-ClanRamsey} Suppose that every $h$-vortex planar graph of $n$ vertices and diameter $D$ can be deterministically embedded into a graph with treewidth $t(h,\eps,n)$ and additive distortion $+\eps D$ via a clique-preserving embedding.
	Then for every parameter $\delta\in(0,1)$ and every $K_r$-free graph $G=(V,E,w)$ of $n$ vertices and diameter $D$, there exist an $(\eps D, \delta)$-Ramsey type embedding and an  $(\eps D, \delta)$-clan embedding, both with treewidth $O_r(t(O_1(r), O_r(\frac{\delta \eps}{\log n}), n) + \log(n))$.
\end{lemma}
\begin{proof} By \Cref{lm:CKFL-embed-genus-vortex}, any  $h$-multi-vortex-genus graph $G$ of $n$ vertices and diameter $D$ can be deterministically embedded into a graph with treewidth $\hat{t}(h,\eps,n) = O_h(t(O_h(1),O_h(\eps),n))$ via a clique-preserving embedding.  This embedding and \Cref{lm:NearlyClanRamsey} give a $(\eps D, \delta)$-Ramsey type embedding and a  $(\eps D, \delta)$-clan embedding with treewidth:
	\begin{equation*}
		\begin{split}
		\hat{t}(O_r(1),O(\frac{\delta \eps}{\log(n)}),n)+O_r(\log n) &= O_r(t(O_r(1), O_r(\frac{\delta \eps}{\log n}), n) + O_r(\log(n)))\\
			&=	O_r(t(O_1(r), O_r(\frac{\delta \eps}{\log n}), n) + \log(n))~,	
		\end{split}
	\end{equation*}
as claimed.
\end{proof}

We are now ready to prove \Cref{thm:MinorClan} and \Cref{thm:MinorRamsey}.

\begin{proof}[Proof of \Cref{thm:MinorClan} and \Cref{thm:MinorRamsey}] By \Cref{lm:nearlyPlanar-emb}, any $h$-vortex planar graph of $n$ vertices and diameter $D$ has a clique-preserving embedding with additive distortion $+\eps D$  and treewidth: 
	\begin{equation*}
		t(h,\eps,n) = O(\frac{h(\log \log n)^2}{\eps}). 
	\end{equation*}
It follows from \Cref{lm:Red-NearlyPlanar-ClanRamsey} that $K_r$-free graphs of $n$ vertices and diameter $D$ has an $(\eps D, \delta)$-Ramsey type embedding and an $(\eps D, \delta)$-clan embedding, both with treewidth:
	\begin{equation*}
O_r(t(O_1(r), O_r(\frac{\delta \eps}{\log n}), n) + \log(n)) = O_r\left(O_r\left(\frac{(\log \log n)^2}{(\delta\eps)/\log(n)}\right) + \log(n)\right) =  O_{r}\left(\frac{\log n(\log\log n)^2}{\eps\delta}\right)~,
\end{equation*}
as desired.
\end{proof}
\subsection{Applications to Metric Baker's Problems (Proof of \Cref{thm:MetricBakerMinor})}

Given a clan embedding, one could obtain a bicriteria approximation scheme for the $\rho$-dominating set problem by dynamic programming on the host graph. Similarly, given a Ramsey type embedding, one could obtain a bicriteria approximation scheme for the $\rho$-independent set problem, also by dynamic programming. \Cref{lm:KLS-IS-DS} provides efficient dynamic programs for both problems in graphs with low treewidth. The following lemmas specify the running time of the resulting algorithms.

\begin{lemma}[Theorems 7 and 8 \cite{FL21}, adapted]\label{lm:dominating} Suppose that for every parameters $\eps,\delta\in(0,1)$, and every $K_r$-free graph $G=(V,E,w)$ of diameter $D$ and $n$ vertices,
	there are  $(\eps D,\delta)$-clan and $(\eps D,\delta)$-Ramsey type embeddings, both with treewidth $\tw(\eps,\delta)$. Then given an $n$-vertex $K_r$-minor-free graph $G(V,E,w)$, two parameters $\epsilon \in (0,1)$ and $\rho > 0$, and a vertex weight function $\mu: V\rightarrow \mathbb{R}^+$, one can find in  $2^{O(\tw(\Theta(\eps),\Theta(\eps))\cdot\log\frac{\tw(\Theta(\eps),\Theta(\eps))}{\eps})}\cdot\poly(n)$ time:\\
	(1) a $(1-\epsilon)\rho$-independent  set $I$ such that for every $\rho$-independent  set $\tilde{I}$, $\mu(I) \geq (1-\epsilon)\mu(\tilde{I})$, and\\
	(2) a  $(1+\epsilon)\rho$-dominating set $S$ such that for every $\rho$-dominating set $\tilde{S}$, $\mu(S) \leq (1+\epsilon)\mu(\tilde{S})$.
\end{lemma}

We are now ready to prove \Cref{thm:MetricBakerMinor}.

\begin{proof}[Proof of \Cref{thm:MetricBakerMinor}] By \Cref{thm:MinorClan} and \Cref{thm:MinorRamsey}, $K_r$-minor-free graphs  of diameter $D$ and $n$ vertices have $(\eps D,\delta)$-clan and $(\eps D,\delta)$-Ramsey type embeddings, both with treewidth $\tw(\eps,\delta) = O_r(\frac{\log n (\log\log n)^2)}{\eps \delta})$. Thus, by \Cref{lm:dominating}, we obtain a bicriteria approximation scheme for $\rho$-independent set and $\rho$-dominating set problems in time:
	\begin{equation*}
		\begin{split}
		2^{O(\tw(\Theta(\eps),\Theta(\eps))\cdot\log\frac{\tw(\Theta(\eps),\Theta(\eps))}{\eps})}\cdot\poly(n) &= 2^{O(O_r(\frac{\log n (\log\log n)^2}{\eps^2}\log(\frac{\log n (\log\log n)^2}{\eps^3})}\\
				&= n^{\tilde{O}_r(\eps^{-2}(\log\log(n))^3)}~,			
		\end{split}
	\end{equation*}
as claimed.
\end{proof}

\section{Treewidth Lower Bounds}\label{sec:TwLb}

\subsection{Additive Distortion (Proof of \Cref{thm:PlanarEmbeddingLB})}

\begin{wrapfigure}{r}{0.1\textwidth}
	\begin{center}
		\vspace{-20pt}
		\includegraphics[width=0.95\textwidth]{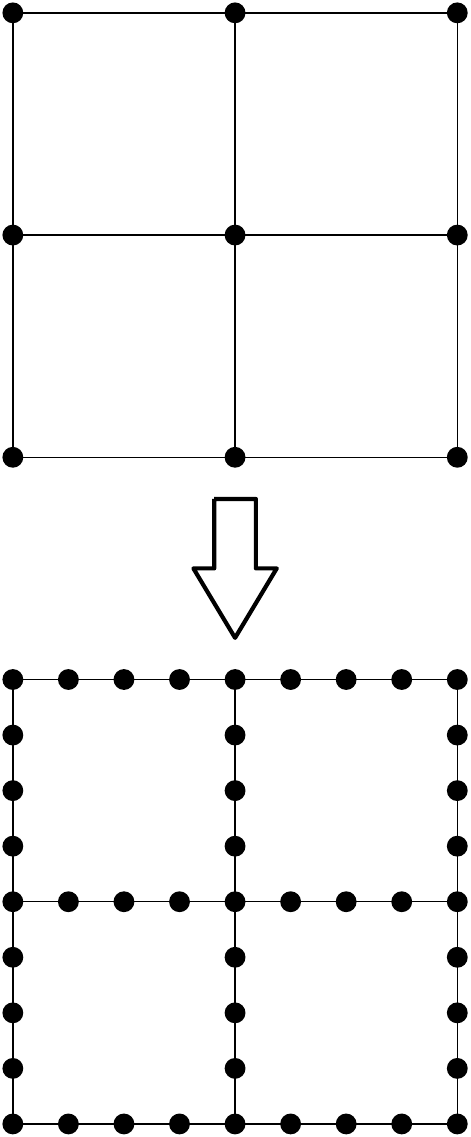}
		\vspace{-40pt}
	\end{center}
	\vspace{-10pt}
\end{wrapfigure}
In this section, we prove a lower bound on the treewidth of the embedding planar graphs with additive distortion. For a given graph $J = (V_J,E_J)$, we denoted by $J_k$ a $k$-subdivision of $J$. That is, $J_k$ is obtained from $J$ by replacing each edge with a path of length $k$. See the illustration on the right for a $4$-subdivision. Our lower bound uses the following lemma by Carrol and Goel~\cite{CG04}.

\begin{lemma}[Lemma 1~\cite{CG04}]\label{lm:CGLB} Let $H$ be a (possibly weighted) graph that excludes $J$ as a minor.
	Every dominating embedding from $J_k$ to $H$ has a multiplicative distortion at least $\frac{k-3}{6}$. 
\end{lemma}

\PlanarEmbeddingLB*
\begin{proof}	W.l.o.g, we assume that $\eps$ is sufficiently smaller than $1$ since otherwise, the theorem trivially holds. Set $n'=\frac{1}{42\eps}+1$. Let $Q$ be an $n'\times n'$-grid. The diameter of $Q$ is $2\cdot \frac{1}{42\eps} = \frac{1}{21\eps}$. Let $Q_{k}$ be a $k$-subdivision of $Q$ for $k = 21$. Observe that the diameter of  $Q_k$ is $21\cdot\frac{1}{21\eps} = \frac{1}{\eps}$. Furthermore, $|V(Q_k)| = O(1/\eps^2)$. Finally, we add $n - |V(Q_k)|$ dummy vertices and connect all of them to some vertex of $Q_k$ in a star-like way. Let $G$ be the resulting graph; $G$ is planar by construction. Furthermore, the diameter of $G$ is $D =  1/\eps + 1$.  
	
	Suppose that there is a deterministic, dominating embedding $f$ of $G$ with an additive distortion at most $\eps D$ into a treewidth $o(1/\eps)$ graph $H$. Observe that $\eps D = \eps(\frac{1}{\eps} + 2) = 1+2\eps \leq 2$ when $\eps \leq 1/2$. It follows that $f$ has multiplicative distortion at most $2$.   As $H$ has treewidth $o(1/\eps)$ (see e.g. \cite{NST94}), it excludes $Q$ as a minor. Thus by \Cref{lm:CGLB}, any embedding of $Q_k$ into $H$, including $f_{Q_k}$, must have a multiplicative distortion at least $\frac{k-3}{6} = 3$, a contradiction. Thus, the treewidth of $H$ must be $\Omega(1/\eps)$ as claimed.
\end{proof}

\subsection{Multiplicative Distortion (proof of \Cref{thm:multLBStochastic})}\label{subsec:multDistLB}

In this section, we follow the lower bound constructions of  Carroll and Goel~\cite{CG04} and Chakrabarti, Jaffe, Lee, and Vincent~\cite{CJLV08} (which by itself is based on \cite{GNRS04})
to derive lower bounds on the multiplicative distortion for stochastic embeddings of planar graphs into low-treewidth graphs.
Chakrabarti \etal \cite{CJLV08} showed that there exists an infinite family of graphs such that for any graph $G$ in the family with $n$ vertices, any embedding of $G$ into  a distribution over graphs of \emph{constant} treewidth must have an expected (multiplicative) distortion $\Omega(\log n)$. By carefully adapting the parameters of the construction of  Chakrabarti \etal~\cite{CJLV08}, we obtain a quantitative lower bound w.r.t. the treewidth. The rest of the section is devoted to proving \Cref{thm:multLBStochastic}.
\begin{theorem}\label{thm:multLBStochastic} There exists an infinite family of planar graphs $\mathcal{G}$ such that for any graph $G_n\in \mathcal{G}$ of $n$ vertices for a sufficiently large $n$, any embedding of $G_n$ into a distribution over graphs of treewidth $O(\frac{\log^{\nicefrac13} n}{\log\log n})$ must have expected multiplicative distortion $\Omega(\frac{\log^{\nicefrac13} n}{\log\log n})$.
\end{theorem}

Given an edge-weighted graph $G = (V,E,w)$ and two vertices $s\ne t$ in $G$, we say that $G$ is \emph{$(s,t)$-geodesic} if we can partition the edge set of $G$ into a set of edge-disjoint shortest $s$-to-$t$ paths. 

Denote by $Q$ the $(n+1)\times (n+1)$ unweighted grid, and by $Q_{p}$ the $p$-subdivision of $Q$. 
Denote by $\mathcal{X}_n$ the family of graphs with treewidth at most $n$. As every graph of treewidth at most $n$ excludes $Q$ as a minor~\cite{RS86}, by \Cref{lm:CGLB} we observe:
\begin{observation}\label{obs:multLBDeterministic} 
	For every $p\in\N$, every \emph{deterministic} embedding of $Q_{p}$ into a graph in $\mathcal{X}_n$ must incur multiplicative distortion $\frac{p-3}{6}$.
\end{observation}

Next we construct a geodesic version of the grid graph $Q$. Metric embedding $f:X\rightarrow Y$ between two metric spaces $(X,d_X),(Y,d_Y)$ is called isometric if for every $x,y\in X$, $d_X(x,y)=\alpha\cdot d_Y(f(x),f(y))$ for some constant $\alpha$.
\begin{lemma}\label{lem:GeodesicGrid}
	There is an unweighted $(s,t)$-geodesic planar graph $\hat{Q}$ with $<8n^2$ edges, and $2n+2$ disjoint $s-t$ shortest paths $\mathcal{P}$ of length $2n$, such that $Q$ embeds isometrically into $\hat{Q}$.
\end{lemma}
\begin{proof}
	Denote by $s,t$ the upper-left, and lower-right vertices of the grid $Q$, respectively.
	There is a set of $2(n+1)$ shortest paths $\{P_i\}_{i=1}^{2(n+1)}$ from $s$ to $t$ such that every edge of $Q$ belongs to some path in $Q$ (each path can ``cover'' all the edges in a single row/column). By duplicating edges, we can obtain a planar multi-graph $Q^M$, such that the edges of $Q^M$ can be partitioned into disjoint paths $\{P_i^M\}_{i=1}^{2(n+1)}$, while the shortest path distance of $Q$ and $Q^M$ is the same.
	$Q^M$ has $|E(Q^M)|=2n\cdot2(n+1)$ edges (as it a collection of $2(n+1)$ paths of length $2n$). The graph $Q^M_2$ (where each edge is subdivided to two edges) is a simple graph, with  $|E(Q^M_2)|=2|E(Q^M)|=4n\cdot2(n+1)$ edges, and $|\mathcal{P}(Q^M_2)|=2(n+1)$ disjoint $s-t$ shortest paths of length $2n$. Note that $Q$ indeed embeds isometrically into $Q^M_2$, as all distances are increased by a factor of $2$. 
	See \Cref{fig:geodesic} for an illustration.
\end{proof}
	\begin{figure}[t]
	\centering{\includegraphics[scale=.64]{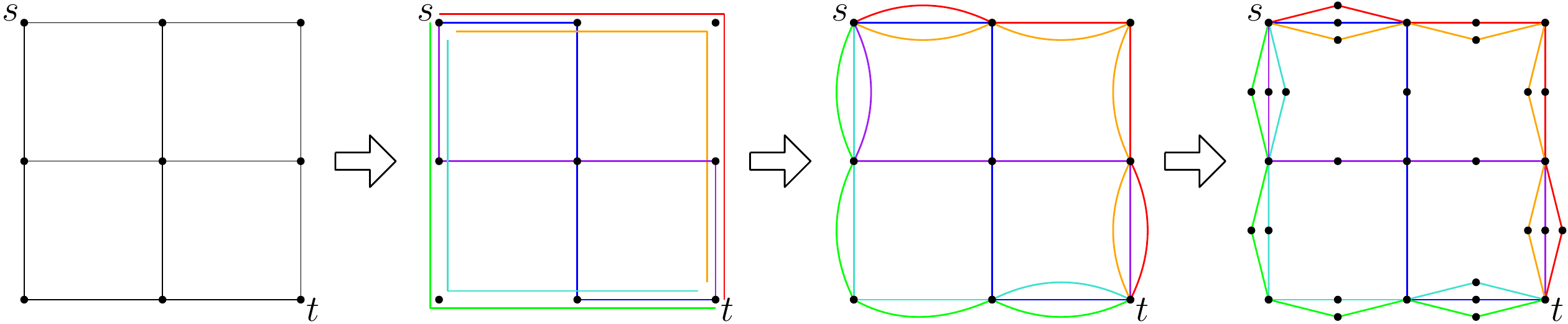}}
	\caption{\label{fig:geodesic}\small The $3\times3$ grid $Q$ is illustrated in the first figure on the left. In the second figure, all $Q$ edges are covered by $6$ paths. In the third figure, we obtain an isometric multi-graph whose edges can be partitioned into disjoint shortest paths. In the right figure, we obtain a simple graph by subdividing each edge.}
\end{figure}
	\begin{figure}[t]
	\centering{\includegraphics[scale=.27]{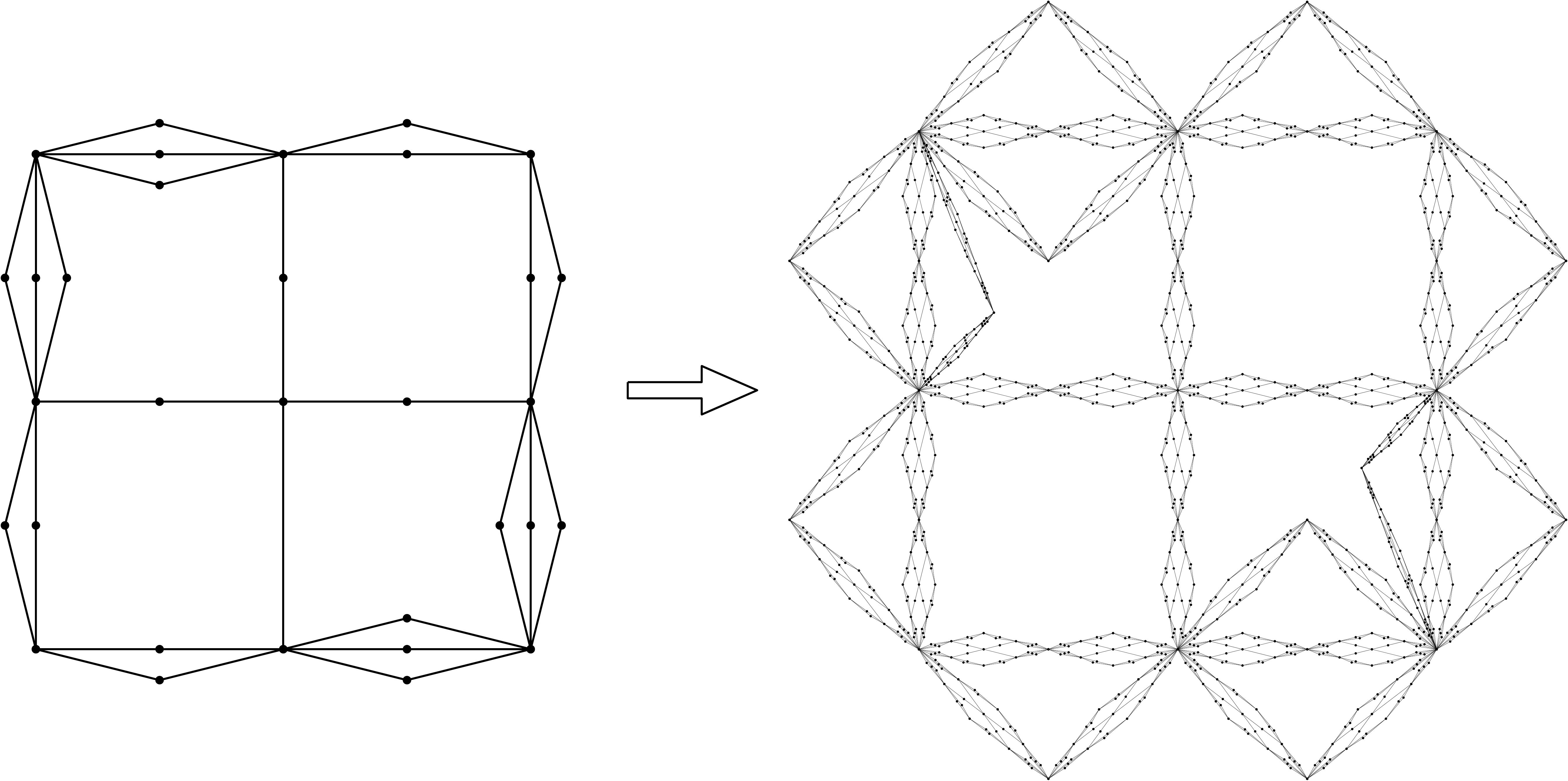}}
	\caption{\label{fig:fractal}\small On the left illustrated the geodesic version $\hat{Q}$ of the $3\times3$ grid $Q$, also denoted by $H_1$.
	On the right the graph $H_2$ is illustrated, where we replaced each edge of $H_1$ with a copy of $H_1$.}
\end{figure}
Note that as $Q$ embeds isometrically into $\hat{Q}$, every $p$-subdivision of $Q$, denoted $Q_p$, embeds isometrically into the $p$-subdivision of $\hat{Q}$, denoted $\hat{Q}_p$. In particular, 
every deterministic embedding of $\hat{Q}_{p}$ into $\mathcal{X}_n$ incurs a multiplicative distortion at least $\frac{p-3}{6}$.

Let $H_1$ by the graph $\hat{Q}$. For $k\ge 2$, we construct $H_k$, by replacing every edge in $H_{k-1}$ with a copy of $H_1$, identifying $s,t$ with the endpoints of the edge arbitrarily.
Note that $H_k$ is a planar graph. See \Cref{fig:fractal} for an illustration.
Following a mini-max argument in \cite{GNRS04}, in order to show that every stochastic embedding of $H_k$ into graphs of $\mathcal{X}_n$ has expected multiplicative distortion $\alpha$, it is enough to show that for every deterministic embedding, the average distortion over the edges is $\alpha$. More formally, it is enough to show that for every dominating embedding $f$ into $G\in\mathcal{X}_n$ it holds that $\sum_{\{u,v\}\in E(H_k)}d_G(f(u),f(v))\ge\alpha\cdot|E(H_k)|$.

Denote $m=|E(H_1)|=|E(\hat{Q})|$. The following claim follows by a simple induction:
\begin{claim}\label{clm:disjointCopies}
	For all $1\le i\le k$, $H_k$ contains $m^{k-i}$ disjoint copies of $H_i$.
\end{claim}
$H_1$ contains $\rho=2(n+1)$ shortest paths, all of length $\ell=2n$. Note that $\ell\cdot\rho=m$. By another induction, for $i\ge 1$ it follows:
\begin{claim}\label{clm:disjointSubdivisions}
	$H_i$ contains $\rho^{i-1}$ edge-disjoint $\ell^{i-1}$-subdivisions of $H_1$ (that is $\hat{Q}_{(2n)^{i-1}}$).
\end{claim}
Fix some embedding $f$ of $H_k$ into a graph in $\mathcal{X}_n$.
By \Cref{obs:multLBDeterministic}, each such subdivision  $\hat{Q}_{\ell^{i-1}}$ has distortion $\frac{\ell^{i-1}-3}{6}\ge \frac{\ell^{i-1}}{7}$ w.r.t. $f$ (we assume $n$ is large enough, and for $i=1$ the distortion is at least $1\ge\frac17$).
In particular,  some edge $\{u,v\}$ in each subdivision has $d_G(f(u),f(v))\ge \frac{1}{7}\cdot\ell^{i-1}$.
Let $C_i=\left\{\{u,v\}\in E(H_k)\mid d_G(f(u),f(v))\ge \frac{1}{7}\cdot\ell^{i-1} \right\}$ be the set of of all edges that are distorted by at least $\frac{1}{7}\cdot\ell^{i-1}$.
By combining \Cref{clm:disjointCopies} and \Cref{clm:disjointSubdivisions}, $|C_i|\ge m^{k-i}\cdot\rho^{i-1}$.
For an edge $\{u,v\}\in E(H_k)$, let $C_{\{u,v\}}=\left\{i\mid \{u,v\}\in C_i\right\}$.
It holds that $\sum_{i\in C_{\{u,v\}}}\frac{1}{7}\cdot\ell^{i-1}\le\frac{2}{7}\cdot\ell^{\max C_{\{u,v\}}-1}\le2\cdot d_{G}(f(u),f(v))$.
We conclude:
\begin{align*}
	\sum_{\{u,v\}\in E(H_{k})}d_{G}(f(u),f(v)) & \ge\sum_{\{u,v\}\in E(H_{k})}\sum_{i\in C_{\{u,v\}}}\frac{1}{14}\cdot\ell^{i-1}=\frac{1}{14}\cdot\sum_{i=1}^{k}|C_{i}|\cdot\ell^{i-1}\\
	& =\frac{1}{14}\cdot\sum_{i=1}^{k}m^{k-i}\cdot\rho^{i-1}\cdot\ell^{i-1}=\frac{1}{14}\cdot\sum_{i=1}^{k}m^{k-i}\cdot m^{i-1}=\frac{k}{14m}\cdot m^{k}=\frac{k}{14m}\cdot|E(H_{k})|~.
\end{align*}
Set $k=n^3$. Then $H_k$ has $N=|V(H_k)|\le |E(H_k)|=m^k<(8n)^{2k}=2^{O(n^3\log n)}$ vertices, and every stochastic embedding into a graphs of treewidth $n=O(\frac{\log^{\nicefrac13} N}{\log\log N})$, will have expected distortion at least $\frac{k}{14m}=\Omega(n)=\Omega(\frac{\log^{\nicefrac13} N}{\log\log N})$. \Cref{thm:multLBStochastic} now follows.

\begin{remark}
	It is interesting to understand the exact trade-off between the treewidth of the host graphs, and the expected distortion. \cite{FRT04} implies that one can embed every $N$-vertex planar graph (actually general graph) into a distribution over treewidth-$1$ graphs with expected distortion $O(\log N)$, while \Cref{thm:multLBStochastic} states that one cannot obtain expected distortion smaller that $\Omega(\frac{\log^{\nicefrac13} N}{\log\log N})$ for embedding planar graphs into treewidth-$\Omega(\frac{\log^{\nicefrac13} N}{\log\log N})$ graphs.
	By choosing $k=n^{2+\alpha}$ in the proof of \Cref{thm:multLBStochastic} for constant $\alpha>0$, we derive that one cannot obtain distortion smaller than $\Omega\left(\frac{\log^{\frac{\alpha}{2+\alpha}}N}{\log\log N}\right)$
	for embedding planar graphs into treewidth-$O\left(\frac{\log^{\frac{1}{2+\alpha}}N}{\log\log N}\right)$ graphs.
\end{remark}

\section*{Acknowledgments} The authors are grateful to Philip Klein and Vincent Cohen-Addad for helpful conversations.

	\bibliographystyle{alphaurlinit}
	\bibliography{RamseyTreewidthBib,RPTALGbib,spanner}

\end{document}